\documentclass[11.5pt]{article}
\newcommand{\blind}{1}

\usepackage{amssymb,amsbsy,amsfonts,amsmath,xspace,amsthm}
\usepackage{mathrsfs}
\usepackage{graphicx}
\usepackage[boxed]{algorithm2e}

\usepackage{caption}
\usepackage{setspace}
\usepackage{comment}
\usepackage[margin=1in]{geometry}
\usepackage{enumitem}
\usepackage{subcaption}
\usepackage{multirow}
\usepackage{epstopdf}
\usepackage{epsfig}
\usepackage[colorlinks,citecolor=blue]{hyperref}

\usepackage{url}
\usepackage[toc,page]{appendix}
\usepackage{float}
\usepackage{natbib}
\usepackage{color}
\usepackage[dvipsnames]{xcolor}
\usepackage{verbatim}
\usepackage{authblk}
\usepackage[normalem]{ulem}
\newcommand{\ord}{\mathrm{ord}}

\theoremstyle{definition}
\newtheorem{thm}{Theorem}

\newtheorem{defi}{Definition}
\newtheorem{algo}{Algorithm}
\newtheorem{prop}{Proposition}
\newtheorem{coro}{Corollary}
\newtheorem{lem}{Lemma}

\newcommand{\e}{\mathbb{E}}
\newcommand{\p}{\mathbb{P}}

\newcommand{\hcal}{\mathcal{H}}
\newcommand{\ccal}{\mathcal{C}}

\newcommand{\scal}{\mathcal{S}}
\newcommand{\one}{\textbf{1}}

\newcommand{\td}{\tilde}
\newcommand{\bR}{\mathbb{R}}

\newcommand{\mc}{\mathrm{mc}}

\newcommand{\mbone}{{\mathbf{1}}}

\newcommand{\diag}{\text{diag}}

\newcommand{\norm}[1]{\Vert{#1}\Vert}
\newcommand{\sign}{\mathrm{sign}}

\newcommand{\lb}{\left(}
\newcommand{\rb}{\right)}

\def\ack{\section*{Acknowledgements}%
  \addtocontents{toc}{\protect\vspace{6pt}}%
  \addcontentsline{toc}{section}{Acknowledgements}%
}

\def\text#1{\mbox{\rm #1}}

\usepackage{hyperref}
\hypersetup{
    colorlinks=true,
    linkcolor=blue,
    filecolor=magenta,      
    urlcolor=cyan,
    citecolor=blue
}
 
\begin{document}

\def\spacingset#1{\renewcommand{\baselinestretch}%
{#1}\small\normalsize} \spacingset{1}

\if1\blind
{
  \title{\bf Hierarchical community detection by recursive partitioning}
  \author{Tianxi Li\\
    Department of Statistics, University of Virginia\\
    Lihua Lei \\
    Department of Statistics, Stanford University\\
        Sharmodeep Bhattacharyya       \\
    Department of Statistics, Oregon State University\\
    Koen Van den Berge \thanks{K. Van den Berge is a postdoctoral fellow of the Belgian American Educational Foundation (BAEF) and is supported by the Research Foundation Flanders (FWO), grant 1246220N}\\
    Department of Statistics, University of California, Berkeley and Department of Applied Mathematics, Computer Science and Statistics, Ghent University\\
        Purnamrita Sarkar \thanks{P. Sarkar was supported in part by an NSF grant (DMS-1713082)}\\
    Department of Statistics and Data Sciences, University of Texas at Austin
        Peter J. Bickel\thanks{P. Bickel is supported in part by an NSF grant (DMS-1713083)} \\
    Department of Statistics, University of California, Berkeley\\
    and\\
            Elizaveta Levina\thanks{
    E. Levina is supported in part by  NSF grants (DMS-1521551 and DMS-1916222) and an ONR grant (N000141612910)} \\
    Department of Statistics, University of Michigan}

  \maketitle
} \fi

\if0\blind
{
  \bigskip
  \bigskip
  \bigskip
  \begin{center}
    {\LARGE\bf Hierarchical community detection by recursive partitioning}
\end{center}
  \medskip
} \fi

\bigskip
\begin{abstract}
  The problem of community detection in networks is usually formulated
  as finding a single partition of the network into some 
  ``correct'' number of communities.  We argue that it is more interpretable and
  in some regimes more accurate to construct a hierarchical tree of
  communities instead.   This can be done with a simple top-down recursive
  partitioning algorithm,  starting with a single community and
  separating the nodes into two communities by spectral clustering
  repeatedly, until a stopping rule suggests there are no further
  communities.  This class of algorithms is model-free, computationally efficient, and
  requires no tuning other than selecting a stopping rule.    We show
  that there are regimes where this approach outperforms $K$-way spectral
  clustering, and propose a natural framework for analyzing  the algorithm's theoretical performance, the
  binary tree stochastic block model.   Under this model, we prove that the
  algorithm correctly recovers the entire community tree under
  relatively mild assumptions.    We apply the algorithm to a gene network based on gene co-occurrence in 1580 research papers on anemia, and identify six clusters of genes in a meaningful hierarchy.  We also illustrate the algorithm on a dataset of statistics papers.  
  
\end{abstract}

\noindent%
\vfill

\newpage
\spacingset{1.5} 

\section{Introduction}\label{sec:intro}
Data collected in the form of networks have become increasingly common in many fields, with interesting scientific phenomena discovered through the analysis of biological, social, ecological, and various other networks;   see   \cite{newman2010networks} for a review.   Among various network analysis tasks, community detection has been one of the most studied, due to the ubiquity of communities in different types of networks and the appealing mathematical formulations that lend themselves to analysis;  see for example reviews by \cite{fortunato2010community}, \cite{goldenberg2010survey}, and \cite{abbe2017community}.     Community detection is the task of  
clustering network nodes into groups with similar connection patterns, and in many applications, communities provide a useful and parsimonious representation of the network.   There are many statistical models for networks with communities, including the  stochastic block model \citep{holland1983stochastic} and its many variants and extensions, such as, for example,  \cite{handcock2007model, hoff2008modeling, airoldi2008mixed, karrer2011stochastic, xu2013dynamic, zhang2014detecting, matias2017statistical}.    One large class of methods focuses on fitting such models based on their likelihoods or approximations to them 
\citep{bickel2009nonparametric, mariadassou2010uncovering, celisse2012consistency, bickel2013asymptotic,amini2013pseudo};   another class of methods takes an algorithmic approach, designing algorithms, often based on spectral clustering,  that can sometimes be proven to work well under specific models  \citep{newman2004finding, newman2006modularity, rohe2011spectral, bickel2011method, zhao2012consistency, chen2012clustering, lei2014consistency, cai2015robust, chen2014statistical, amini2018semidefinite, joseph2016impact, gao2015achieving,le2017concentration, gao2016optimal}.

Most work on community detection to date has focused on finding a single $K$-way partition of the network into $K$ groups, which are sometimes allowed to overlap.   This frequently leads to a mathematical structure that allows for sophisticated analysis, but for larger $K$ these partitions tend to be unstable and not easily interpretable.  These methods also typically require the ``true'' number of clusters $K$ as input. Although various methods have been proposed to estimate $K$ \citep[e.g.][]{chen2014network, chatterjee2015matrix, wang2015likelihood, le2015estimating, li2016network}, none of them have been especially tested or studied for large $K$, and in our experience, empirically they perform poorly when $K$ is large.    Finally, a single ``true''  number of communities may not always be scientifically meaningful, since in practice different community structures can often be observed at different scales.

Communities in real networks are often hierarchically structured, and the hierarchy can be scientifically meaningful, for example, a phylogenetic tree.  A hierarchical tree of communities, with larger communities subdivided into smaller ones further down, offers a natural and very interpretable representation of communities.    It also simplifies the task of estimating $K$, since, instead of estimating a large $K$ from the entire network we only need to check  whether a particular subnetwork contains more than one community.    We can also view a hierarchy as regularizing an otherwise unwieldy model with a large number of communities, which in theory can approximate any exchangeable graph \citep{olhede2014network}, by imposing structural constraints on parameters.    We would expect that for large networks with many communities, such regularization can lead to improvements in both computational costs and theoretical guarantees.     

Hierarchical community detection methods can be generally divided into three types: estimating the hierarchy directly and all at once, typically with either a Bayesian or an optimization method;   agglomerative algorithms that merge nodes or communities recursively in a bottom-up fashion;  and partitioning algorithms which split communities recursively in a top-down fashion. The earliest work in the first category we are aware of is  \cite{kleinberg2002small}, generalized by \cite{clauset2008hierarchical} and \cite{peel2015detecting}.  These models directly incorporate a tree by modeling connection probabilities between pairs of nodes based on their relative distance on the tree.  One line of work treats the tree as a parameter and takes a Bayesian approach, e.g., \cite{clauset2008hierarchical,blundell2013bayesian}.     Bayesian inference on these models is computationally prohibitive, and thus infeasible for large networks.   Even more importantly, treating each node as a leaf involves a large number of parameters, therefore sacrificing interpretability. Agglomerative clustering algorithms for networks date back to at least \cite{clauset2004finding}, which combined  well-known Ward's hierarchical clustering \citep{ward1963hierarchical} with the Girvan-Newman modularity \citep{newman2004finding} as the objective function. The idea of modularity-based clustering was further explored by \cite{pons2005computing, reichardt2006statistical, wakita2007finding, arenas2008analysis, blondel2008fast}. Divisive algorithms were once very popular in machine learning problems such as graph partitioning and image segmentation \citep{spielman1996spectral, shi2000normalized, kannan2004clusterings}.    This class of methods appears to have first been applied to networks by \cite{girvan2002community}, who proposed an ``elbow-finding'' algorithm deleting the edge with highest betweenness centrality recursively, thereby building a top-down hierarchy. This idea was further developed by \cite{wilkinson2002finding, holme2003subnetwork, gleiser2003community}, and \cite{radicchi2004defining}.

In spite of many practical benefits and applied work on hierarchical community detection, it is hard to come by a rigorous analysis. The first such analysis of a hierarchical algorithm we are aware of was given by \cite{dasgupta2006spectral}, for a recursive bi-partitioning algorithm based on a modified version of spectral clustering. Their analysis allows for sparse networks with average degree growing poly-logarithmically in $n$, but the procedure involves multiple tuning parameters with no obvious default values.   Later on, \cite{balakrishnan2011noise} considered a top-down hierarchical clustering algorithm based on unnormalized graph Laplacian and the model of \cite{clauset2008hierarchical}, for a pairwise similarity matrix instead of a network. They did not propose a practical stopping rule, but did provide a rigorous frequentist theoretical guarantee for clustering accuracy.  However, as we will further discuss in Section \ref{sec:theory}, their analysis only works for dense networks which are rare in practice. \cite{lyzinski2017community} proposed another hierarchical model based on a mixture of random dot product graph (RDPG) models \citep{young2007random}. In contrast to \cite{balakrishnan2011noise}, they use a two-stage procedure which first detects all communities, and then applies agglomerative hierarchical clustering to build the hierarchy from the bottom up.  They proved strong consistency of their algorithm, but it hinges on perfect recovery of all communities in the first step, which leads to very strong requirements on network density.

In this paper, we consider a framework for hierarchical community detection based on recursive bi-partitioning, an algorithm similar to \cite{balakrishnan2011noise}.   The algorithm needs a partitioning method, which divides any given network into two, and a stopping rule, which decides if a given network has at least two communities;  in principle, any partitioning method and any stopping rule can be used.  The algorithm starts by splitting the entire network into two and then tests each resulting leaf with the stopping rule, until the stopping rule indicates there is nothing left to split.   We prove that the algorithm consistently recovers the entire hierarchy, including all low-level communities,  under the binary tree stochastic block model (BTSBM), a hierarchical network model with communities we propose, in the spirit of \cite{clauset2008hierarchical}. Our analysis applies to networks with average degree as low as $(\log n)^{2+\epsilon}$ for any $\epsilon > 0$, while existing results either require the degree to be polynomial in $n$, or $\log^{a}n$ for large $a$ (e.g. $a = 6$ in \cite{dasgupta2006spectral}) at the cost of numerous tuning  parameters. We also allow the number of communities $K$ to grow with $n$, which is natural for a hierarchy, at a strictly faster rate than previous work, which for the most part treats $K$ as fixed.   Even more importantly, when $K$ is too big to recover the entire tree, we can still consistently recover mega-communities at the higher levels of the hierarchy, whereas $K$-way clustering will fail.       Since the stopping rule only needs to decide whether $K > 1$ rather than estimate $K$, we can use either hypothesis tests  \citep{bickel2013hypothesis, gao2017testing, jin2019optimal}, or various methods for estimating $K$.   Importantly, the main weakness of methods for of $K$, which is underestimating $K$ when it is large, since empirically they never underestimate it so severely as to conclude $K = 1$.    Unlike previous studies of hierarchical community detection, we are able to provide theoretical guarantees for a data-driven stopping rule rather than known $K$. Finally, our procedure has better computational complexity than $K$-way partitioning methods.

The rest of the paper is organized as follows. In Section~\ref{sec:method}, we present our general recursive bi-partitioning framework, a specific recursive algorithm,  and discuss the interpretation of  the resulting hierarchical structure.    In  Section~\ref{sec:theory}, we introduce a special class of stochastic block models under which a hierarchy of communities can be naturally defined, and provide theoretical guarantees on recovering the hierarchy for that class of models.  Section~\ref{sec:sim} presents extensive simulation studies demonstrating advantages of recursive bi-partitioning for both community detection and estimating the hierarchy.     Section~\ref{sec:app} applies the proposed algorithm to a gene co-occurrence network and obtains a readily interpretable hierarchical community structure. Section~\ref{sec:discussion} concludes with discussion.   Proofs and an additional data example on a dataset of statistics papers can be found in the Appendix.   


\section{Community detection by recursive partitioning}
\label{sec:method}
\subsection{Setup and notation}
We assume an undirected network on nodes  $1, 2, \cdots, n$. The corresponding $n \times n$ symmetric adjacency matrix $A$ is defined by $A_{ij} = 1$ if and only if node $i$ and node $j$ are connected, and 0 otherwise.   We use $[n]$ to denote the integer set $\{1, 2, \cdots, n\}$. We write $I_n$ for the $n\times n$ identity matrix and $\mbone_n$  for $n\times 1$ column vector of ones, suppressing the dependence on $n$ when the context makes it clear.   For any matrix $M$, we use $\norm{M}$ to denote its spectral norm (the largest singular value of $M$), and  $\norm{M}_F$  the Frobenius matrix norm.   Community detection will output a partition of nodes into $K$ sets, $V_1 \cup V_2 \cup \cdots \cup V_K = [n]$ and $V_i \cap V_j = \emptyset$ for any $i\ne j$, with $K$ typically unknown beforehand.   

\subsection{The recursive partitioning algorithm}

In many network problems where a hierarchical relationship between communities is expected,  estimating the hierarchy accurately is just as important as finding the final  partition of the nodes.    A natural framework for producing a  hierarchy is recursive partitioning,  an old idea in clustering that has not resurfaced much in the current statistical network analysis literature \citep[e.g.][]{kannan2004clusterings, dasgupta2006spectral, balakrishnan2011noise}.  The framework is general and can be used in combination with any community detection algorithm and model selection method; we will give a few options that worked very well in our experiments.   In principle, the output can be any tree, but we focus on binary trees, as is commonly done in hierarchical clustering;  we will sometimes refer to partitioning into two communities as bi-partitioning.  

Recursive bi-partitioning does exactly what its name suggests:  
\begin{enumerate}
\item Apply a decision / model selection rule to the network to decide if it contains more than one community.   If no, stop;  if yes, split into two communities.
  \item Repeat step 1 for each of the resulting communities, and continue until no further splits are indicated.    
\end{enumerate}
This is a top-down clustering procedure which produces a binary tree, but the leaves are small communities, not necessarily single nodes.    Intuitively, as one goes down the tree, the communities become closer, so the tree distance between communities reflects their level of connection.  

Computationally, while we do have to partition multiple times, each community detection problem we have to solve is only for $K=2$, which is faster, easier and more stable than for a general $K$, and the size of networks decreases as we go down the tree and thus it becomes faster.    When $K$ is large and connectivity levels between different communities are heterogeneous, we expect recursive partitioning to outperform $K$-way clustering, which does best for small $K$ and when everything is balanced.

We call this approach hierarchical community detection (HCD). As input, it takes a network adjacency matrix $A$; an algorithm that takes an adjacency matrix $A$ as input and partitions it  into two communities, outputting their two induced submatrices, $\ccal(A) = \{A_1, A_2\}$;   and 
a stopping rule $\scal: \bR^{n\times n} \to \{0,1\}$, where $\scal(A)=1$ indicates there is no evidence $A$ has communities and we should stop, and $\scal(A)=0$ otherwise.   Its output $\hcal_{\ccal,\scal}(A) = (c, T)$ is the community label vector $c$ and the hierarchical tree of communities $T$.
The algorithm clearly depends on the choice of the partitioning algorithm $\ccal$ and the stopping rule $\scal$;  we describe a few specific options next.

\subsection{The choice of partitioning method and stopping rule}

Possibly the simplest partitioning algorithm is a simple eigenvector sign check, used in \cite{balakrishnan2011noise, gao2015achieving, le2017concentration,abbe2017entrywise}:  
\begin{algo}\label{algo:SS}
	Given the adjacency matrix $A$:
	\begin{enumerate}
		\item Compute the eigenvector $\tilde{u}_2$ corresponding to the second largest eigenvalue in magnitude of $A$.
		\item Let $\hat{c}(i) =  0$ if $\tilde{u}_{2,i} \ge 0$ and $\hat{c}(i) = 1$ otherwise. 
		\item Return label $\hat{c}$.
	\end{enumerate}
\end{algo}

A more general and effective partitioning method is regularized spectral clustering (RSC), especially for sparse networks.    Several regularized versions are available;   in this paper, we use the proposal of \cite{amini2013pseudo}, shown to improve performance of spectral clustering for sparse networks \citep{joseph2016impact, le2017concentration}. 

\begin{algo}\label{algo:RSC}
	\SetAlgoLined
	Given the adjacency matrix $A$ and regularization parameter $\tau$ (by default, we use $\tau =0.1$), do the following:
	\begin{enumerate}
		\item Compute the regularized adjacency matrix as
		$$A_{\tau} = A+\tau\frac{\bar{d}}{n}\mbone\mbone^T$$
		where $\bar{d}$ is the average degree of the network.
		\item Let $D_{\tau} = \diag(d_{\tau 1}, d_{\tau2}, \cdots, d_{\tau n})$ where $d_{\tau i} = \sum_j A_{\tau,ij}$ and calculate the regularized Laplacian 
		$$L_{\tau} = D_{\tau}^{-1/2}A_{\tau}D_{\tau}^{-1/2}.$$
		\item Compute the leading two eigenvectors of $L_{\tau}$, arrange them in a $n\times 2$ matrix $U$, and apply $K$-means algorithm to the rows, with $K=2$.
		\item Return the cluster labels from the $K$-means result.
	\end{enumerate}
\end{algo}

The simplest stopping rule is to fix the depth of the tree in advance, though that is not what we will ultimately do.   A number of recent papers focused on estimating the number of communities in a network, typically assuming that each community in the network is generated from either the Erd\"{o}s-Renyi model or the configuration model \citep{van2016random}.     The methods proposed include directly estimating rank by the USVT method of \cite{chatterjee2015matrix}, hypothesis tests of \cite{bickel2013hypothesis}, \cite{gao2017testing} and \cite{jin2019optimal}, the BIC criteria of \cite{wang2015likelihood}, the spectral methods of \cite{le2015estimating} and cross-validation methods of \cite{chen2014network, li2016network}.    The cross-validation method of \cite{li2016network} works for both unweighted and weighted networks under a low rank assumption, while the others use the block model assumption. 

Under block models,  empirically we found that the most accurate and computationally feasible stopping criterion is the non-backtracking method of \cite{le2015estimating}. Let $B_{\mathrm{nb}}$ be the non-backtracking matrix, defined by 
\begin{equation}\label{eq:NBmatrix}
B_{\mathrm{nb}} = 
\begin{pmatrix}
0 &D-I \\
-I & A \\
\end{pmatrix}.
\end{equation}

Let $\lambda_i, i \in [2n]$ be the real parts of the eigenvalues of $B_{\mathrm{nb}}$ (which may be complex).    The number of communities  is then estimated as the number of eigenvalues that satisfy  $|\lambda_i| > \norm{B_{\mathrm{nb}}}^{1/2}$. This is because if the network is generated from an SBM with $K$ communities, the largest $K$ eigenvectors of $B_{\mathrm{nb}}$ will be well separated from the radius $\norm{B_{\mathrm{nb}}}^{1/2}$ with high probability, at least in sparse networks  \citep{krzakala2013spectral, le2015estimating}. 
We approximate the norm $\norm{B_{\mathrm{nb}}}$  by $\frac{\sum_i d_i^2}{\sum_i d_i}-1$, as suggested by \cite{le2015estimating}.   For our purposes, we only need the real parts of the two leading eigenvalues, not the full spectrum.  If we want to avoid the block model assumption, the edge cross-validation (ECV) method of \cite{li2016network} can be used instead to check whether a rank 1 model is a good approximation to the subnetwork under consideration.     

The main benefit of using these estimators as stopping rules (i.e.,  checking at every step if the estimated $K$ is greater than 1) is that the tree can be of any form;  if we fixed $K$ in advance, we would have to choose in what order to do the splits in order to end up with exactly the chosen $K$.    Moreover, empirically we found the local stopping criterion is more accurate  than directly estimating $K$, especially with larger $K$.     For the rest of the paper, we will focus on two versions,  ``HCD-Sign"  which uses splitting by eigenvalue sign (Algorithm~\ref{algo:SS}), and "HCD-Spec", which uses regularized spectral clustering (Algorithm~\ref{algo:RSC}).   Any of the stopping rules discussed above can be used with either method.

\subsection{Mega-communities and a similarity measure for binary trees}\label{secsec:binarytree}

The final communities (leaves of the tree) as well as the intermediate {\em mega-communities}  can be indexed by their position on the tree.   Formally, 
each node or (mega-)community of the binary tree can be represented by a sequence of binary values  $x \in \{0,1\}^{l_x}$, where $l_x$ is the depth of the node (the root node has depth 0).   The string $x$ records the path from the root to the node, with $x_q=1$ if step $q$ of the path is along the right branch of the split and $x_q=0$ otherwise.   We define the $x$ for the root node to be an empty string.     Intuitively, the tree induces a  similarity measure between communities: two communities that are split further down the tree should be more similar to each other than two communities that are split higher up.   The similarity between two mega-communities does not depend on how they are split further down the tree, which is a desirable feature.   Note that we do not assume an underlying hierarchical community model; the tree is simply the output of the HCD algorithm.

To quantify this notion of tree similarity, we define a similarity measure between two nodes $x,x'$  on a binary tree by 
$$s(x,x') = \min\{q:x_q \ne x'_q\}.$$
For instance, for the binary tree in Figure~\ref{fig:BTSBM}, we have $s(000,001) = 3$, while $ s(000,11) =s(000,110) =s(000,111) = 1$.    Note that comparing values of $s$ is only meaningful for comparing pairs with a common tree node. So $s(000,111) < s(000,001)$ indicates that  community $000$ is closer to community $001$ than to $111$, but the comparison between $s(000,001)$ and $s(10,11)$ is not meaningful.

A natural question is whether this tree structure and the associated similarity measure tell us anything about the underlying population model.    Suppose that the network is in fact generated from the SBM.    The probability matrix $P = \e A$ under the SBM is block-constant, and  applying either HCD-Sign or HCD-Spec to $P$  will recover the correct communities and produce a binary tree.    This binary tree may not be unique; for example, for the planted partition model where all communities have equal sizes, all within-block edge probabilities are $a$ and all between-block edge probabilities are $b$. However, in many situations $P$ does correspond to a unique binary tree (up to a permutation of labels), for example, under the hierarchical model introduced in Section~\ref{sec:theory}.     For the moment, assume this is the case.  Let $c$ and $T$ be the binary string community labels and the binary tree produced by applying the HCD algorithm to $P$ in exactly the same way we previously applied it to $A$. Let  $\hat{c}$ and $\hat{T}$  be the result of applying HCD to $A$.     The estimated tree $\hat{T}$ depends on the stopping rule and may be very different in size from $T$;  however, we can always compute the tree-based similarity between nodes based on their labels.   Let 
$S_T = \left(s_T(c(i),c(j))\right)$ be the $n\times n$ matrix where $s_{T}$ is the pairwise similarities induced by $T$, and  $S_{\hat{T}} = \left(s_{\hat{T}}(\hat{c}(i),\hat{c}(j))\right)$  the  corresponding similarity matrix based on $\hat{T}$. $S_{\hat{T}} $ can be viewed as an estimate of $S_T$, and we argue that comparing $S_T$ to $S_{\hat T}$ may give a more informative measure of performance than just comparing $\hat c$ to $c$.    This is because with a large $K$ and weak signals it may be hard or impossible to estimate all communities correctly, but if the tree gets most of the mega-communities right, it is still a useful and largely correct representation of the network.

Finally, we note that an estimate of $S_T$ under the SBM can be obtained for any community detection method:  if $\tilde{c}$ are estimated community labels, we can always estimate the corresponding $\tilde{P}$ under the SBM and apply HCD to $\tilde{P}$ to obtain an estimated tree $\tilde{T}$.   However, our empirical results  in Section~\ref{sec:sim} show that applying HCD directly to the adjacency matrix $A$ to obtain $S_{\hat{T}}$ gives a better estimate of $S_T$ than the $S_{\tilde{T}}$ constructed from post-processing the estimated probability matrix produced by a $K$-way partitioning method.


\section{Theoretical properties of the HCD algorithm}\label{sec:theory}
\subsection{The binary tree stochastic block model}
We now proceed to study the properties of HCD on a class of SBMs that naturally admit a binary tree community structure.  We call this class the Binary Tree Stochastic Block Models (BTSBM), formally defined in Definition \ref{defi:BTSBM} and illustrated in Figure \ref{fig:BTSBM}.  
\begin{figure}[H]
  \centering
  \includegraphics[width = 0.6\textwidth]{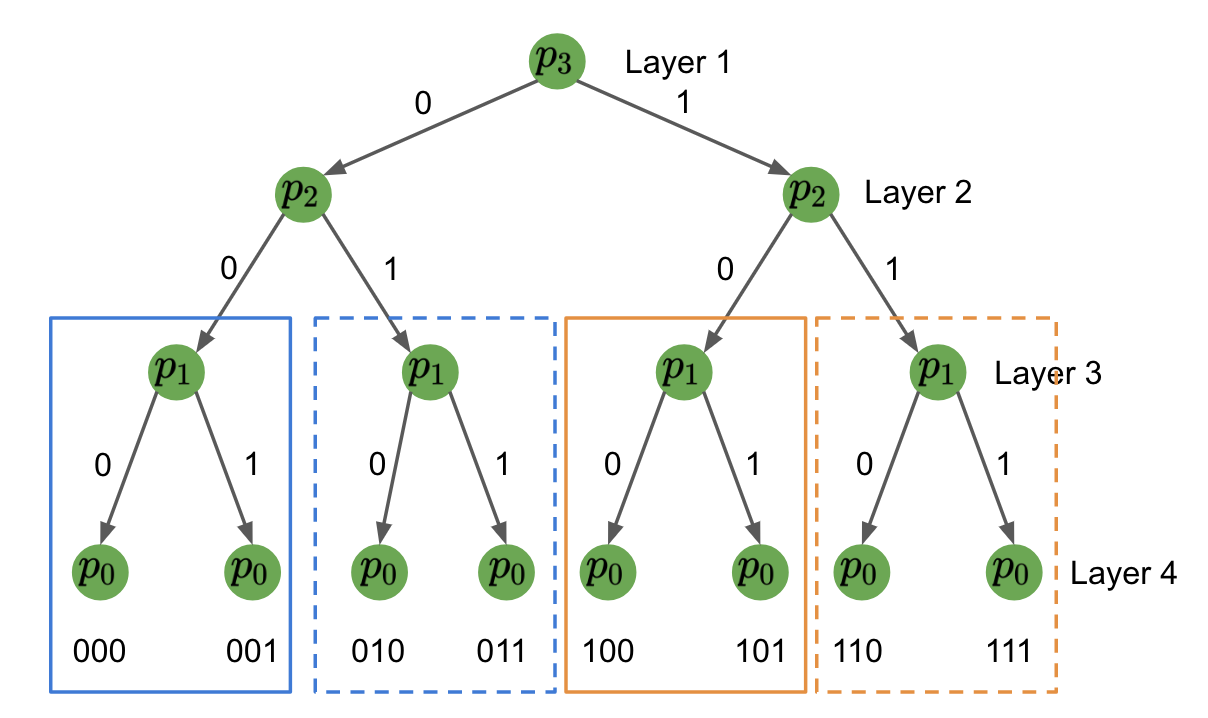}
  \caption{An $8$-cluster Binary Tree SBM.  Rectangles correspond to mega-communities.}\label{fig:BTSBM}
\end{figure}

\begin{defi}[The binary tree stochastic block model (BTSBM]\label{defi:BTSBM}
  Let $S_d:=\{0,1\}^d$ be the set of all length $d$ binary sequences and let $K =  |S_d| = 2^d$. Each binary string in $S_d$ encodes a community label and has a 1-1 mapping to an integer in $[K]$ via standard binary representation $\mathcal{I}: S_d \to [K]$. 
  For node $i \in [n]$, let $c(i) \in S_d$ be its community label, let $C_x = \{i:  c(i) = x\}$ be the set of nodes labeled with string $x \in S_d$, and let $n_x = |C_x|$.  
\begin{enumerate}
\item Let  $B\in \bR^{K\times K}$ be a matrix of probabilities defined by  $$B_{\mathcal{I}(x),\mathcal{I}(x')} = p_{\scriptscriptstyle{D(x,x')}}$$ 
   where $p_0, p_1, \ldots p_d$ are arbitrary $d + 1$ parameters in $[0, 1]$ and 
$$D(x, x') = (d + 1 - s(x, x'))I(x \not = x'),$$
for $s(x,x') = \min\{q:x_q \ne x'_q\}$ defined in Section \ref{secsec:binarytree}.
   
 \item Edges between all pairs of distinct nodes $i,
 j$ are independent Bernoulli, with
   $$P(A_{i, j} = 1) = B_{\mathcal{I}(c(i)),\mathcal{I}(c(j))} $$
 corresponding to the $n \times n$ probability matrix $P = E A$.
   \end{enumerate}
 \end{defi}

For instance, the BTSBM in Figure \ref{fig:BTSBM} corresponds to the matrix 
  \[B = \left[
      \begin{array}{cc|cc|cccc}
        p_{0} & p_{1} & p_{2} & p_{2} & p_{3} & p_{3} & p_{3} & p_{3}\\
        p_{1}& p_{0} & p_{2} & p_{2} & p_{3} & p_{3} & p_{3} & p_{3}\\
        \hline
        p_{2} & p_{2} & p_{0} & p_{1} & p_{3} & p_{3} & p_{3} & p_{3}\\
        p_{2} & p_{2} &p_{1}& p_{0} & p_{3} & p_{3} & p_{3} & p_{3}\\
        \hline
        p_{3} & p_{3} & p_{3} & p_{3} & p_{0} & p_{1} & p_{2} & p_{2}\\
        p_{3} & p_{3} & p_{3} & p_{3} &p_{1}& p_{0} & p_{2} & p_{2}\\
        p_{3} & p_{3} & p_{3} & p_{3} & p_{2} & p_{2} & p_{0} & p_{1}\\
        p_{3} & p_{3} & p_{3} & p_{3} & p_{2} & p_{2} &p_{1}& p_{0}
      \end{array}
      \right].
\]

A nice consequence of defining community labels through binary strings is that they naturally embed the communities in a binary tree.   We can think of each entry of the binary string as representing one level of the tree, with the first digit corresponding to the first split at the top of the tree, and so on.   We then define a {\em mega-community} labeled by a binary string $x\in S_q$ at any level $q$ of the tree as the set $\{i: c(i)_h = x_h, 1 \le h \le q\}$, defined on a binary tree $T$.  The mega-communities are unique up to community label permutations, and give a multi-scale view of the community structure;  for example, Figure~\ref{fig:BTSBM} shows four mega-communities in layer 3 and two mega-communities in layer 2.  

The idea of embedding connection probabilities in a tree, to the best of our knowledge, was first introduced as the  hierarchical random graph (HRG) by \cite{clauset2008hierarchical}, and extended by \cite{balakrishnan2011noise} to weighted graphs and by \cite{peel2015detecting} to general dendrograms. The BTSBM can be viewed as a hybrid of the original HRG and the SBM, maintaining parsimony by estimating only community-level parameters while imposing a natural and interpretable hierarchical structure.    It also provides us with a model that can be used to analyze recursive bi-partitioning on sparse graphs. 

\subsection{The eigenstructure of the BTSBM}\label{subsec:eigen}
 Let $Z\in \bR^{n\times K}$ be the membership matrix with the $i$-th row $Z_{i} = e_{\mathcal{I}(c(i))}$ containing $e_{j}$, the $j$-th canonical basis vector in $\bR^{K}$, where $\mathcal{I}$ is the integer given by the binary representation. Then it is straightforward to show that
\[P = \e A = ZBZ^{T} - p_{0}I.\]
The second term comes from the fact that $A_{ii} = 0$. For the rest of the theoretical analysis, we assume equal block sizes, i.e., 
\begin{equation}
  \label{eq:equal_block}
  n_{1} = n_{2} = \ldots = n_{K} = n / K = m.
\end{equation}
This assumption is stringent but standard in the literature and can be relaxed to a certain extent, as indicated in our Theorem~\ref{thm:unequal} in Appendix C.    For the BTSBM, this assumption leads to a particularly simple and elegant eigenstructure for $P$.

Given a (mega)-community label denoted by a binary string $x$, we write $x0$ and $x1$ as the binary strings obtained by appending $0$ and $1$ to $x$, respectively.    We further define $\{x+\}$ to be the set of all binary strings starting with $x$.   The following theorem gives a full characterization of the eigenstructure for the BTSBM.

\begin{thm}\label{thm:eigen2}
  Let $P$ be the $n\times n$ community connection probability matrix of the BTSBM with $K = 2^d$ and define $\tilde{P} = P+p_0I =ZBZ^T$. 
  Then the following holds:  \\
  1. (Eigenvalues) The distinct nonzero eigenvalues of $\tilde{P}$, denoted by $\lambda_{1},\lambda_{2} \cdots,  \lambda_{ d+1}$,  are given by 
\begin{equation}\label{eq:lambda}
  \lambda_{1} = m(p_0+\sum_{r=1}^{d}2^{r-1}p_r), \quad \lambda_{q+1} = m\left(p_0+\sum_{r=1}^{d - q} 2^{r-1}p_r - 2^{d - q}p_{d - q + 1}\right), ~~q = 1, \ldots, d.
\end{equation}
2. (Eigenvectors) For any $1\le q \le d$ and each $x \in S_{q-1}$, let $\nu_{x}^{q+1}$ be an $n$-dimensional vector, such that for any $i \in [n]$, 
\[\nu_{x, i}^{q+1} = \left\{
    \begin{array}{rl}
      1 & \mbox{if }c(i) \in S_d\cap \{x0+\} \, , \\
      -1 & \mbox{if }c(i) \in S_d\cap \{x1+\} \, , \\
      0 & \mbox{otherwise}   \, . 
    \end{array}
\right.\]
Then the eigenspace corresponding to eigenvalue $\lambda_{q+1}$ is spanned by $\{\nu_{x}^{q+1}: x \in S_{q-1}\}$. The eigenspace corresponding to $\lambda_1$ is spanned by $\mbone$.
\end{thm}

It is easy to see that each $\nu_{x}^{q+1}$ corresponds to a split of the two mega-communities in layer $q$, at an internal tree node $x$. For instance, consider the colored rectangles in Figure \ref{fig:BTSBM}, which correspond to $d = 3$ and $q = 2$.  The vector $\nu_{0}^{3}$ has entry $1$ for all nodes in the (solid blue) mega-community 00, entry $-1$ for  all nodes in the  (dashed blue) mega-community $01$ and $0$ for all the other nodes, thus separating mega-communities $00$ and $01$. Similarly, $\nu_{1}^{3}$ has entry $1$ for the nodes in (solid orange) mega-community $10$, and entry $-1$ for nodes in the  (dashed orange) mega-community $11$.  The binary tree structure is thus fully characterized by the signs of eigenvectors' entries. Note that due to multiplicity of eigenvalues, the basis of the eigenspace is not unique. In Appendix~\ref{app:eigen}, we use another basis which, though less interpretable, is used in the proof of Theorem~\ref{thm:consistency-ass} and Theorem~\ref{thm:consistency-dis-ass} to obtain better theoretical guarantees.

While the previous theorem is stated for general configurations of $p_0, \cdots, p_d$, the two most natural situations where a hierarchy is meaningful are either 
{\em assortative} communities, with 
\begin{equation}
  \label{eq:assortative}
  p_{0} > p_{1} > \ldots > p_{d},
\end{equation}
or {\em dis-assortative} communities, with 
\begin{equation}
  \label{eq:disassortative}
  p_{0} < p_{1} < \ldots < p_{d}.
\end{equation}

Recall that the HCD-sign algorithm only depends on the eigenvector corresponding to the second largest eigenvalue (in magnitude).  Theorem \ref{thm:eigen2} directly implies that under either the assortative or dis-assortative setting, such eigenvalue is unique (has multiplicity $1$) with an eigenvector that yields the first split in the tree according to the signs of the corresponding eigenvector entries.  

\begin{coro}\label{cor:eigen}
Let $P$ be the $n\times n$ probability matrix of the BTSBM with $K = 2^d$ and balanced community sizes as in \eqref{eq:equal_block}. Under either \eqref{eq:assortative} or \eqref{eq:disassortative}, the second largest eigenvalue (in absolute value) of $P$ for a BTSBM is unique and given by
\[(m - 1)p_{0} + m\sum_{i=1}^{d-1}2^{i-1}p_{i} - m2^{d - 1}p_{d}\]
and the gap between it and the other eigenvalues is $\Delta_{2}= n\min\{p_d, (p_{d-1} - p_{d}) / 2\}$ in the assortative case and $\Delta_{2}= n(p_{d} - p_{d - 1}) / 2$ in the dis-assortative case. The corresponding (normalized) eigenvector is  
$$u_{2} = \frac{1}{\sqrt{n}} ( \one_{n / 2}^T,  -\one_{n / 2}^T)^T.$$
\end{coro}
In a slight abuse of notation, we still denote the $k$th eigenvalue of $P$ (instead of $\tilde{P}$) by $\lambda_k$ whenever it is clear from the context. 
%


\subsection{Consistency of HCD-sign under the BTSBM}\label{subsec:consistency}
The population binary tree $T$ defined in Section~\ref{secsec:binarytree} is unique under the BTSBM, and thus we can evaluate methods under this model on how well they estimate the population tree.   Given a community label $c$ and the corresponding balanced binary tree $T$ of depth $d$, define $\mc(T,c, q) \in [2^q]^n$ to be the community partition of all nodes into the mega-communities at level $q$ corresponding to $c$.    In particular, at level $d$,  $\mc(T,c,d)$ gives the true community labels $c$, up to a label permutation.    This quantity is well defined only if the binary tree is balanced (i.e., all leaves are at the same depth $d$), and we will restrict our analysis to the balanced case.  

The convention in the literature is to scale all probabilities of connection by a common factor that goes to 0 with $n$, and have no other dependency on $n$; see, e.g., the review \cite{abbe2017community}.    We similarly reparametrize the BTSBM as 
\begin{equation}\label{eq:rhon}
(p_{0}, p_1, \ldots, p_d) = \rho_{n}(1, a_{1}, \ldots, a_{d}).
\end{equation}
Let $\td{u}_{2}$ be the eigenvector of the second largest eigenvalue (in magnitude) of $A$. If 
\begin{equation}
  \label{eq:sign_equal}
  \sign(\td{u}_{2i}) = \sign(u_{2i}) \mbox{ for all }i,
\end{equation}
with high probability, then the first split will achieve exact recovery. A sufficient condition for \eqref{eq:sign_equal} is concentration of $\td{u}_2$ around $u_2$ in the $\ell_{\infty}$ norm. 
The $\ell_{\infty}$ perturbation theory for random matrices is now fairly well studied  \citep[e.g.][]{eldridge2017unperturbed, abbe2017entrywise}. By recursively applying an $\ell_{\infty}$ concentration bound,  we can guarantee recovery of the entire binary tree with high probability, under regularity conditions.

We start from a condition for the stopping rule. Recall that we defined the stopping rule as a function $\Psi$ such that $\Psi(A) =1$ indicates the adjacency matrix $A$ contains communities and $\Psi(A) = 0$ indicates there is no evidence of more than one community.  

\begin{defi}\label{defi:stopping-rule-consistency}
  A stopping rule for a network of size $n$ generated from an SBM with $K$ communities is {\em consistent with rate $\phi$} if $\p(\Psi(A) = 1) \ge 1-n^{-\phi} $ when $K>1$
  and $\p(\Psi(A) = 0) \ge 1-n^{-\phi}$ when $K=1$.
 
\end{defi}

 With a consistent stopping rule, the strong consistency of binary tree recovery can be guaranteed, as stated in the next two theorems.
\begin{thm}[Consistency of HCD-Sign in the assortative setting]\label{thm:consistency-ass}
Let  $A\in \bR^{n\times n}$ be generated from a BTSBM with parameters $(n, \rho_{n}; a_{1}, \ldots, a_{d})$ as defined in \eqref{eq:rhon}, with $n=Km = 2^dm$. Let $\hat{c}$ be the community labels  and $\hat{T}$ the corresponding binary tree computed with the HCD-Sign algorithm with stopping rule $\Psi$. Suppose the model satisfies the assortative condition \eqref{eq:assortative}.   Let $a_{0} = 1$ and for any $\ell \in [d]$, define
\begin{equation}
  \label{eq:delta}
  \eta_{(\ell)} =  \min\{2^{\frac{\ell-r+1}{4}}\eta_{d-r+1}: r\in [\ell]\} , \quad \mbox{where }\eta_{r} = \min\{ a_{r}, |a_{r-1} - a_{r}| / 2\}. 
\end{equation}
 Fix any $\xi > 1, \phi > 1$ and $\phi' > 0$. Then there exists a constant $C(\phi)$, which only depends on $\phi$, such that,  for any $\ell\in [d]$, if
\begin{equation}
  \label{eq:main_cond}
  \sqrt{\frac{K^{\ell / d}}{n\rho_{n}}}\frac{\max\{\log^{\xi} n, \eta_{(\ell)}^{-1}\}}{\eta_{(\ell)}} < C(\phi),
\end{equation}
and the stopping rule $\Psi$ is consistent for all the subgraphs corresponding to mega-communities up to the $\ell+1$ layer with rate $\phi'$,
then for a sufficiently large $n$, 
$$\min_{\Pi \in \text{Perm}(q)}\Pi(\mc(\hat{T},\hat{c},q) )= \mc(T,c,q), \quad \mbox{for all } q\le \ell,$$
with probability at least $1-2 K^{(\phi + 1) \ell/ d}n^{-\phi}-K^{(\phi'+1)\ell/d}n^{-\phi'}$.   The mega-community partition $\mc(T,c,q)$ is defined at start of Section \ref{subsec:consistency} and $\text{Perm}(q)$ is the set of all label permutations on the binary string set $S_q$.  Further, if the conditions hold for $\ell = d$, then with probability at least $1-2 K^{(\phi + 1) }n^{-\phi}-2K^{(\phi'+1)}n^{-\phi'}$, the algorithm exactly recovers the entire binary tree and stops immediately after it is recovered. 
\end{thm}

Theorem \ref{thm:consistency-ass} essentially says that each splitting step of HCD-Sign consistently recovers the corresponding mega-community, provided that the condition \eqref{eq:main_cond} holds for that layer.  Note that, according to \eqref{eq:delta}, in the assortative setting,
$$\sqrt{\frac{K^{\ell / d}}{n\rho_{n}}} \frac{1}{\eta_{(\ell)}^2} = \sqrt{\frac{2^{\ell}}{n\rho_{n}}} \frac{1}{\eta_{(\ell)}^2} = \frac{1}{\sqrt{n\rho_n}}\frac{1}{\min\{2^{-\frac{r-1}{2}}\eta_{d-(r-1)}^2, r \in [\ell]\}}.$$
As $\ell$ increases,  the set over which we minimize grows,  while each individual term remains the same.  Thus the whole term increases with $\ell$. We also have
$$\sqrt{\frac{K^{\ell / d}}{n\rho_{n}}} \frac{\log^{\xi}{n}}{\eta_{(\ell)}} = \frac{2^{\ell/4}\log^{\xi}{n}}{\sqrt{n\rho_n}}\frac{1}{\min\{2^{-\frac{r-1}{4}}\eta_{d-(r-1)}, r \in [\ell]\}},$$
which also increases in $\ell$. Therefore, \eqref{eq:main_cond} gets strictly harder to satisfy as $\ell$ increases, and even if recovering the entire tree is intrinsically hard or simply impossible (condition \eqref{eq:main_cond} fails to hold for $\ell = d$),  HCD-Sign can still consistently recover mega-communities at higher levels of the hierarchy, as long as they satisfy the condition.   This is a major practical advantage of recursive partitioning compared to both $K$-way partitioning and agglomerative hierarchical clustering of \cite{lyzinski2017community}.   A similar result holds in the dis-assortative setting.

\begin{thm}[Consistency of HCD-Sign in the dis-assortative setting]\label{thm:consistency-dis-ass}
Suppose the model satisfied the dis-assortative condition \eqref{eq:disassortative}.   Under the setting of Theorem \ref{thm:consistency-ass}, the conclusion continues to hold if \eqref{eq:main_cond} is replaced by 
\begin{equation}
  \label{eq:main_cond-dis-ass}
  \sqrt{\frac{K^{\ell / d}}{n\rho_{n}}}\frac{\nu_{(\ell)}\max\{\log^{\xi} n, \nu_{(\ell)}\eta_{(\ell)}^{-1}\}}{\eta_{(\ell)}}  < C(\phi),
\end{equation}
where 
\begin{equation}
  \label{eq:nu-new}
  \nu_{(\ell)} =  \max\left\{2^{-\frac{\ell-r+1}{2}}a_{d-r+1}: r\in [\ell]\right\}.
\end{equation}
\end{thm}
It is easy to verify that condition \eqref{eq:main_cond-dis-ass} also becomes harder to satisfy for larger $\ell$.   Therefore in dis-assortative settings we may also be able to recover mega-communities even if we cannot recover the whole tree.  

The theorems apply to any consistent stopping rule satisfying Definition~\ref{defi:stopping-rule-consistency}. In particular, the non-backtracking matrix method we use in the implementation is a consistent stopping rule based on the recently updated result of \cite{le2015estimating}, as we show next.

\begin{prop}\label{prop:NB-consistency}
Define
\begin{equation}\label{eq:zeta}
\zeta_{(\ell)} = \min\left\{\frac{(1+\sum_{j=1}^{\ell-r}2^{j-1}a_j-2^{\ell-r}a_{\ell-r+1})^2}{1+\sum_{j=1}^{\ell-r+1}2^{j-1}a_j}, r\in [\ell] \right\}.
\end{equation}
If 
\begin{equation}\label{eq:stopping-rule-2'}
\frac{K}{n\rho_n}\log{n} \le \min\left\{  \frac{\zeta_{(\ell)}}{25}\log{n}, 1+\sum_{j=1}^{d-\ell+1}2^{j-1}a_j \right\},
\end{equation}
and
\begin{equation}\label{eq:stopping-rule-3'}
n\rho_n\max\{1, a_d\} \le n^{2/13},
\end{equation}
then for a sufficiently large $n$, the non-backtracking matrix stopping rule, described in \eqref{eq:NBmatrix}, is consistent with rate $1$ under BTSBM for all mega-communities up to the layer $\ell + 1$.
\end{prop}

  Proposition~\ref{prop:NB-consistency} directly implies that if \eqref{eq:stopping-rule-2'},  \eqref{eq:stopping-rule-3'} and \eqref{eq:main_cond} or \eqref{eq:main_cond-dis-ass} hold at the same time, the conclusions of Theorem~\ref{thm:consistency-ass} and Theorem~\ref{thm:consistency-dis-ass}  hold when using the non-backtracking matrix as the stopping rule. In the proposition, condition \eqref{eq:stopping-rule-3'} constrains the network from being too dense. We believe this to be an artifact of the proof technique of \citep{le2015estimating}.   Intuitively, if the method works for a sparser network, it should work for a denser one as well, so we expect this condition can be removed, but we do not pursue this direction since the non-backtracking estimator is not the focus of the present paper. A similar argument shows that another class of stopping rules based on Bethe-Hessian matrices \citep{le2015estimating} also gives consistent stopping rules which will also give consistent recovery.  

Next,  we illustrate conditions \eqref{eq:stopping-rule-2'} and \eqref{eq:main_cond} in a simplified setting.   Consider the assortative setting when the whole tree can be recovered ($\ell = d$).

\noindent Example 1.  Assume an arithmetic sequence $a_r$, given by $a_{r} = 1 - r / (d + 1)$. In this case, taking $\xi = 1.01$, it is easy to see that \eqref{eq:main_cond} is
\begin{equation}
  \label{eq:arithmetic_cond}
  K(d + 1)^{2} \log^{2.02}n = O(n\rho_{n}).
\end{equation}
Some simple algebra shows that condition \eqref{eq:stopping-rule-2'} simplifies to 
$$K\log{n} = O(n\rho_n) \text{~~~~and~~~~} Kd^2 = O(n\rho_n).$$
%
%

Thus the additional requirement for strong consistency with the non-backtracking matrix stopping rule is redundant and we only need the condition
$$K(d+1)^2\log^{2.02}{n} = O(n\rho_n).$$

\noindent Example 2.  Assume a geometric sequence $a_{r}$, given by $a_{r} = \beta^{r}$ for a constant $\beta < 1$. Let $\beta = 2^{-\gamma}$ for $\gamma>0$.
The condition \eqref{eq:stopping-rule-2'}, after some simplifications, becomes 
$$K\log n = O(n\rho_{n}).  $$
On the other hand, for  \eqref{eq:main_cond}, we have 
\begin{align*}
  2^{r / 4}\eta_{r} &= \min\{2^{r / 4}a_{r}, 2^{r / 4}(a_{r-1} - a_{r}) / 2\} = (2^{1/4-\gamma})^{r}\min\left\{1, \frac{\beta^{-1} - 1}{2}\right\}.
\end{align*}
If $\gamma \le 1/4$, then $\eta_{(d)} \ge \min\left\{1, \frac{\beta^{-1} - 1}{2}\right\}$, and \eqref{eq:main_cond} with $\xi = 1.01$ becomes 
\[K\log^{2.02}n = O(n\rho_{n}).\]

If $\gamma > 1/4$, then $\eta_{(d)} = (2^{1/4 - \gamma})^{d}\min\left\{1, \frac{\beta^{-1} - 1}{2}\right\} = \Theta(K^{1/4 - \gamma})$, and \eqref{eq:main_cond} becomes 
$$K^{4\gamma} = O(n\rho_{n}) \text{~~~~and~~~~} K^{2\gamma + 1 / 2}\log^{2.02}n = O(n\rho_{n}).$$
In summary,  the following conditions are sufficient for exact recovery:  
\begin{equation}
  \label{eq:geometric_cond1}
  K^{4\gamma} = O(n\rho_{n}) \text{~~~~and~~~~} K^{\max\{2\gamma + 1 / 2, 1\}}\log^{2.02}n = O(n\rho_{n}).
\end{equation}

\subsection{Comparison with existing theoretical guarantees}
In this section, we compare our result with other strong consistency results for recursive bi-partitioning. We focus on the assortative setting since this most existing results require it \citep{balakrishnan2011noise,lyzinski2017community}.  To make this comparison, we have to ignore the stopping rule, since no other method showed consistency for a data-driven stopping rule.  

Strong consistency of recursive bi-partitioning was previously discussed by \cite{dasgupta2006spectral}.  Their algorithm is far more complicated than ours, with multiple rounds of resampling and pruning, and its computational cost is much higher. For comparison, we can rewrite their assumptions in the BTSBM parametrization.  Their Theorem 1  requires $n\rho_n \ge \log^6{n}$ , and the gap between any two columns of $P$ corresponding to nodes in different communities to satisfy 
\begin{equation}\label{eq:K7}
  \min_{c(u) \ne c(v)}\norm{P_{\cdot u}- P_{\cdot v}}_2^2 = \Omega(K^6\rho_n\log{n}).
\end{equation}
Under the BTSBM,  it is straightforward to show that the minimum in \eqref{eq:K7} is achieved by two communities corresponding to sibling leaves in the last layer and
\begin{equation*}
  \min_{c(u) \ne c(v)}\norm{P_{\cdot u}- P_{\cdot v}}_2^2 = 2\frac{n}{K}\rho_{n}^{2}(1 - a_{1})^{2} \, .  
\end{equation*}
Assume that $K = O(n^{\omega})$ for some $\omega < 1$. Then for sufficiently large $\phi > 0$, $K^{\phi + 1}n^{-\phi} = o(1)$. To compare our result  (with $\ell = d$) to \eqref{eq:K7}, we consider two cases.

1.   Arithmetic sequence: Suppose that $a_{r}$ is given by $a_{r} = 1 - r / (d + 1)$. Then \eqref{eq:K7} gives   
\[K^{7} (d + 1)^{2} \log n = O(n\rho_{n})\Longleftrightarrow K^{7}\log^{2} K \log n = O(n\rho_{n}),\] 
whereas our condition \eqref{eq:arithmetic_cond} is only 
\[K(d + 1)^{2} \log^{2.02}n = O(n\rho_{n})\Longleftrightarrow K\log^{2} K \log^{2.02} n = O(n\rho_{n}),\]
which has a much better dependence on $K$.

2.  Geometric sequence: Suppose that $a_{r}$ is given by $a_{r} = \beta^{r}$. Then the condition \eqref{eq:K7} becomes
\[K^{7} \log n = O(n\rho_{n}).\] 
From \eqref{eq:geometric_cond1}, it is easy to see that HCD-sign has a better rate in $K$ if $\gamma < 7 / 4$, which is equivalent to $\beta > 2^{-7/4} = 0.2973$.

The algorithm proposed in \cite{balakrishnan2011noise} is similar to ours, though their analysis is rather different. Under BTSBM, each entry in the adjacency matrix is sub-Gaussian with parameter $1$ (which cannot be improved), i.e., $\e \exp\{a (A_{ij} - \e A_{ij})\}\le \exp\{a^{2} / 2\}$ for any $a \ge 0$. In order to recover all mega-communities up to layer $\ell$ (with size at least $n / 2^{\ell}$), 
it is easy to show that a necessary condition in their analysis (Theorem 1) is 
$$\rho_{n} = \omega(n^{15/16}).$$
Thus the consistency result in \cite{balakrishnan2011noise} only applies to dense graphs. 

\cite{lyzinski2017community} developed their hierarchical community detection algorithm under a different model they called the hierarchical stochastic blockmodel (HSBM).   The HSBM is defined recursively, as a mixture of lower level models. In particular, when $a_{r} = \beta^{r}$ with $\beta < 1/2$, the BTSBM is a special case of the HSBM where each level of the hierarchy has exactly two communities. \cite{lyzinski2017community} showed the exact tree recovery for fixed $K$ and the average expected degree of at least $O(\sqrt{n}\log^2{n})$, implying a very dense network. By contrast, our result allows for a growing $K$, and if $K$ is fixed, the average degree only needs to grow as fast as $\log^{1 + \xi}{n}$ for an arbitrary $\xi > 1$, which is a much weaker requirement, especially considering that a degree of order $\log n$ is necessary for strong consistency of community labels under a standard SBM with fixed $K$.

\subsection{Computational complexity}

We conclude this section by investigating the computational complexity of HCD, which turns out to be better than that of $K$-way partitioning, especially for problems with a large number of communities. The intuition behind this somewhat surprising result is that, even though HCD has to perform clustering multiple times, it performs a much simpler task at each step, computing no more than two eigenvectors instead of $K$, and the number of nodes to cluster decreases after each step.

We start with stating some relevant known facts.  Suppose we use Lloyd's algorithm \citep{lloyd1982least} for  $K$-means and the Lanczos algorithm \citep{larsen1998lanczos} to compute the spectrum, both with a fixed number of iterations. If the input matrix is $d_1\times d_2$, the $K$-means algorithm has complexity of $O(d_1d_2K)$. In calculating the leading $K$ eigenvectors, we take advantage of matrix sparsity, resulting in complexity $O(\|A\|_{0}K)$ where $\|A\|_{0}$ is the number of non-zero entries of $A$.    Therefore, for $K$-way spectral clustering, the computational  cost is
  \begin{equation}\label{eq:Kway_computation}
    O(nK^{2} + \|A\|_{0}K).
  \end{equation}

Turning to HCD, let $A_{j, \ell}$ denote the adjacency matrix of the $j$-th block in $\ell$-th layer. For comparison purposes, we assume the BTSBM and the conditions for exact recovery, so that we construct the entire tree.  Then $A_{j,\ell}$ corresponds to $2^{d - \ell}m$ nodes. Note that for both Algorithm \ref{algo:SS} and Algorithm \ref{algo:RSC}, the complexity is linear in size. As with \eqref{eq:Kway_computation}, the splitting step applied to $A_{j, \ell}$ has complexity $O(2^{d - \ell}m + \|A_{j, \ell}\|_{0})$ for both HCD-Sign and HCD-Spec. Adding the cost over all layers, the total computation cost becomes 
  \begin{equation}
    \label{eq:HCD_computation1}
    O\lb\sum_{\ell=1}^{d}\sum_{j=1}^{2^{\ell}}(2^{d - \ell}m + \|A_{j, \ell}\|_{0})\rb = O\lb n\log K + \sum_{\ell=1}^{d}\sum_{j=1}^{2^{\ell}}\|A_{j, \ell}\|_{0}\rb,
  \end{equation}
where we use the facts that $K = 2^{d}$ and $n = mK$. Since the blocks corresponding in $\ell$-th layer are disjoint, 
\begin{equation}\label{eq:Ajl0}
  \sum_{j=1}^{2^{\ell}} \|A_{j, \ell}\|_{0}\le \|A\|_{0},
\end{equation}
and thus \eqref{eq:HCD_computation1} is upper bounded by
\begin{equation}
  \label{eq:HCD_computation2}
  O\lb n\log K + \|A\|_{0}\log K\rb.
\end{equation}
This is strictly better than the complexity of $K$-way spectral clustering \eqref{eq:Kway_computation} for large $K$.

Moreover, the inequality \eqref{eq:Ajl0} may be overly conservative. Under the BTSBM, the expected number of within-block edges in the $\ell$-th layer is 
\[\e \sum_{j=1}^{2^{\ell}} \|A_{j, \ell}\|_{0} = 2^{\ell} \lb 2^{d - \ell}m\rb^{2}\bar{p}_{\ell}, \quad \mbox{where }\bar{p}_{\ell} = \frac{p_{0} + \sum_{i=1}^{d - \ell}2^{i-1}p_{i}}{2^{d - \ell}}.\]
As a result,
\[\e \sum_{\ell=1}^{d}\sum_{j=1}^{2^{\ell}} \|A_{j, \ell}\|_{0}  = \lb\sum_{\ell=1}^{d}\frac{\bar{p}_{\ell}}{2^\ell\bar{p}_{0}}\rb\e \|A\|_{0} = \lb\frac{dp_{0} + \sum_{i=1}^{d}(d - i)2^{i-1}p_{i}}{p_{0} + \sum_{i=1}^{d}2^{i-1}p_{i}}\rb\e \|A\|_{0}.\]
The coefficient before $\e \norm{A}_0$ is $O(1)$ in many situations,  including Examples 1 and 2 discussed at the end of Section \ref{subsec:consistency}. In these cases, the average complexity of HCD algorithms is only
\[O(n\log K + \e\|A\|_{0}).\]
Last but not least, the HCD framework, unlike $K$-way partitioning, can be easily parallelized.


\section{Numerical results on synthetic networks}\label{sec:sim}

In this section, we investigate empirical performance of HCD on
  synthetic networks, both on networks with hierarchical structure and those without.   We generate networks with hierarchical structure from the proposed BTSBM, and 
  consider BTSBMs with a fully balanced tree and equal block sizes, which are covered by our theory, as well as BTSBMs with unbalanced trees and different block sizes, which are not.    The non-hierarchical networks will be generated from the regular SBM, included as a general sanity check.
  
We compare methods using the following five measures of accuracy. 
\begin{enumerate}
\item Accuracy of community detection, measured by normalized mutual information (NMI) between true and estimated labels \citep{yao2003information}. NMI is commonly used in the network literature \citep{lancichinetti2009community, amini2013pseudo} and does not require the two sets of labels to have the same number of clusters, a property we need when comparing trees.  
\item Accuracy of recovering the hierarchical structure, measured by the error in the tree distance matrix  \\ $ h\norm{S_{\hat{T}}-S_T}^2_F/\norm{S_T}^2_F$, where $S_{T}$ is defined in Section \ref{secsec:binarytree}.   For baseline methods that do not output a hierarchy, we apply the HCD procedure to their estimate of the probability matrix estimated to construct a tree. 
\item Accuracy of mega-community detection at the top two layers, measured by the proportion of correctly clustered nodes.  We only compute this for balanced BTSBMs.    
\item  Accuracy of estimating the probability matrix $P$, measured by $\norm{\hat{P}-P}_F^2/\norm{P}_F^2$. 
\item Accuracy of estimating $K$, measured by comparing the average $\hat{K}$ with $K$.    Spectral clustering is excluded since it requires $K$ as input;   instead, we report the error of the non-backtrackign estimator of $K$, shown to be one of the most accurate options available for the SBM \citep{le2015estimating}.  

\end{enumerate}

We compare the two versions of HCD (sign-based and spectral clustering) with three other baseline methods:  regularized spectral clustering, the Louvain's modularity method \citep{blondel2008fast}, which is regarded as a most competitive modularity based community detection algorithms by empirical studies \citep{yang2016comparative}, and the recursive partitioning method of \cite{dasgupta2006spectral}.     The first two methods are not hierarchical.  Spectral clustering requires $K$ as an input, and we use the $K$ estimated by hierarchical spectral clustering (HCD-spec) for a fair comparison.   For mega-communities, we also include RSC with the true number of clusters ($K=2$ for layer $1$ and $K=4$ for layer $2$), which can be viewed as an oracle version of RSC.    We use the default parameter settings with regularized spectral clustering and Louvain modularity. Unfortunately, the algorithm of \citep{dasgupta2006spectral} involves several tuning parameters and the paper do not provide recommendations on how to tune them.   For the purposes of this comparison, we tried a number of settings and used a configuration that seemed the best on average.

While we include multiple benchmarks, the main focus of our comparisons is on hierarchical vs $K$-way clustering implemented by the same method (regularized spectral clustering), since it focuses on the central contribution of this work.  All simulation results are averaged over 100 independent replications. 

\subsection{Balanced BTSBM}

We start with a balanced BTSBM with $n=3200 = m2^d$, $K = 2^d$,  and the values of $d = 2, 3, 4, 5, 6$.    The parameter $\beta$ is set so that the average out-in ratio (between-block edge/within-block edges) for all $K$ is fixed at $0.15$ and $\rho_n$ is set so that the average node degree is 50.   These values of $\beta$ and $\rho_n$ are not too challenging, so we can be sure that the main impact on accuracy comes from changing $K$.

Figure~\ref{fig:VaryK} shows the results as a function of $K$.    First, the difference between the two versions of HCD is negligible, with the RSC-based splitting slightly better, so we will simply refer to HCD in comparisons.    HCD outperforms all other methods on both finding all the communities and constructing the hierarchy (Figures~\ref{fig:VaryK-Coclust} and \ref{fig:VaryK-SimErr}), and as predicted by theory, the gain increases with $K$.  While all methods' performance on recovering all communities gets worse as $K$ increases, Figures~\ref{fig:VaryK-Mega2} and \ref{fig:VaryK-Mega4} show that HCD recovers the two top layers of the tree essentially perfectly for all values of $K$, while other methods do not.  This is also consistent with theory which predicts that hierarchical clustering can retain full accuracy on the top levels of the tree.
HCD methods also consistently outperform other methods on estimating the probability matrix (Figure \ref{fig:VaryK-Prob}), as one would expect since the probability matrix depends on accurate community estimation.    Finally,  while all methods underestimate the number of communities when $K$ is large, HCD is more accurate than other methods, except for $K = 64$ when Dasgupta's method is better.
In summary, as $K$ grows and the problem becomes more challenging, the advantages of HCD become more and more pronounced.

Next, we use the same configuration, except now we vary  the average degree of the network, fixing $K=16$, and holding the out-in ratio at 0.15.   Results are shown in Figure~\ref{fig:VarySparsity}.   Again,  HCD performs better and the accuracy improves with degree as suggested by theory. The errors of the modularity method and \cite{dasgupta2006spectral}, on the other hand, appear to converge to a constant as the degree increases and do not suggest consistency will hold in this setting.   With an exception of one value of $K$ where Dasgupta's method is closer to the true value of $K$, the HCD algorithm gives the best performance for all tasks.

\begin{figure}[H]
\centering
\begin{subfigure}[t]{0.4\textwidth}
\centering
\includegraphics[width=\textwidth]{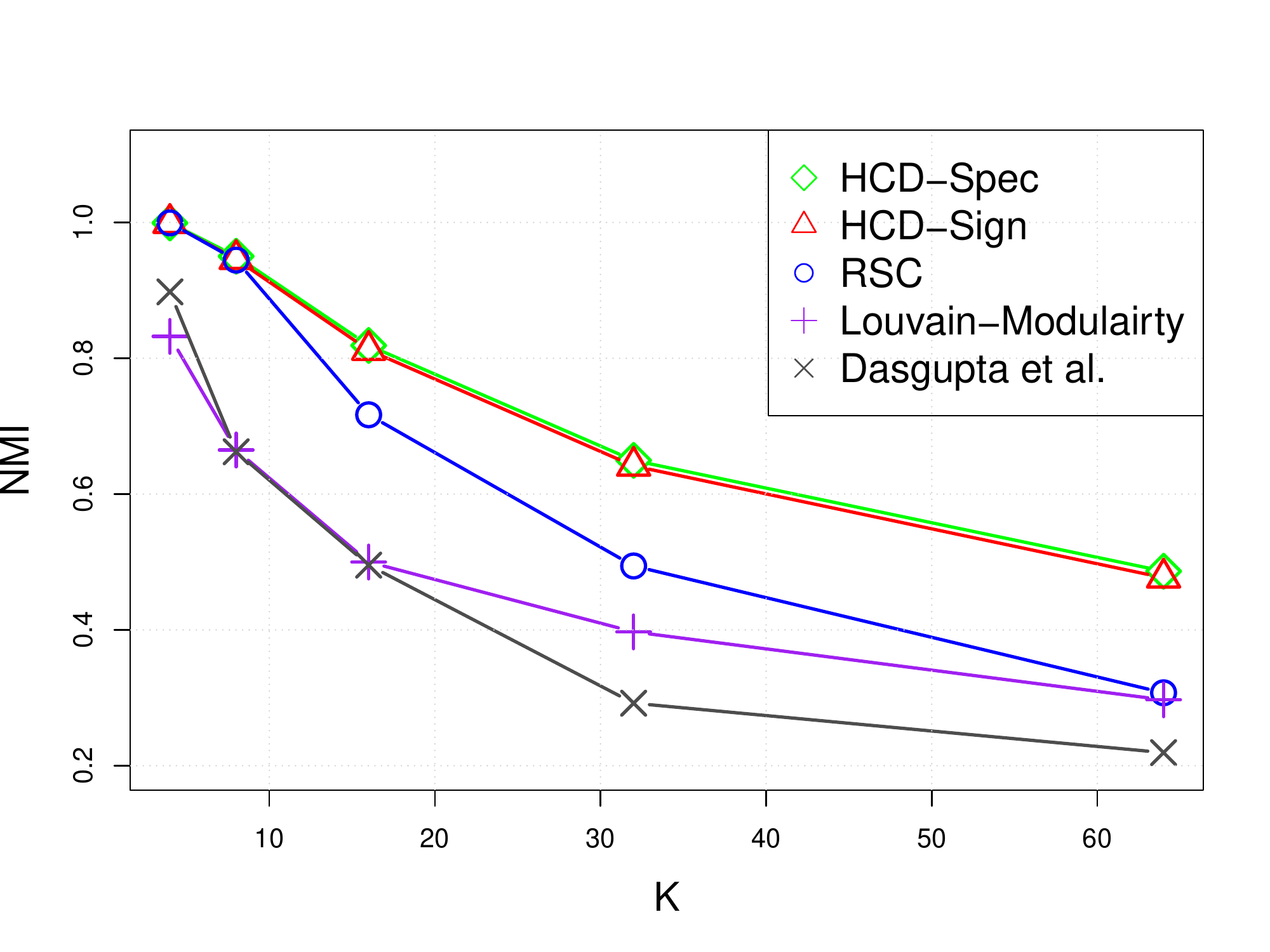}
\caption{Normalized mutual information}
\label{fig:VaryK-Coclust}
\end{subfigure}
\hfill
\begin{subfigure}[t]{0.4\textwidth}
\centering
\includegraphics[width=\textwidth]{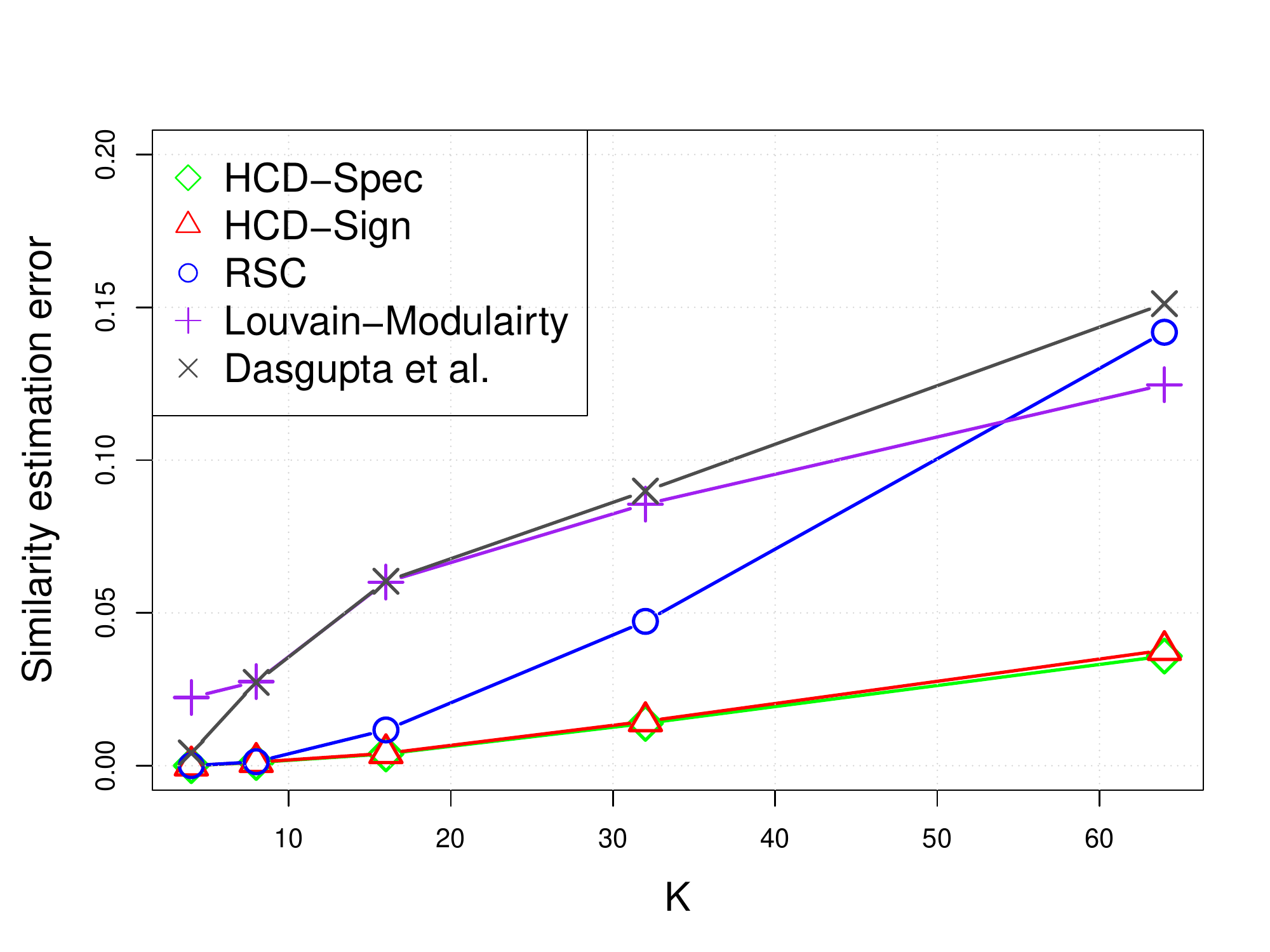}
\caption{Error in the tree similarity matrix}
\label{fig:VaryK-SimErr}
\end{subfigure}

\begin{subfigure}[t]{0.4\textwidth}
\centering
\includegraphics[width=\textwidth]{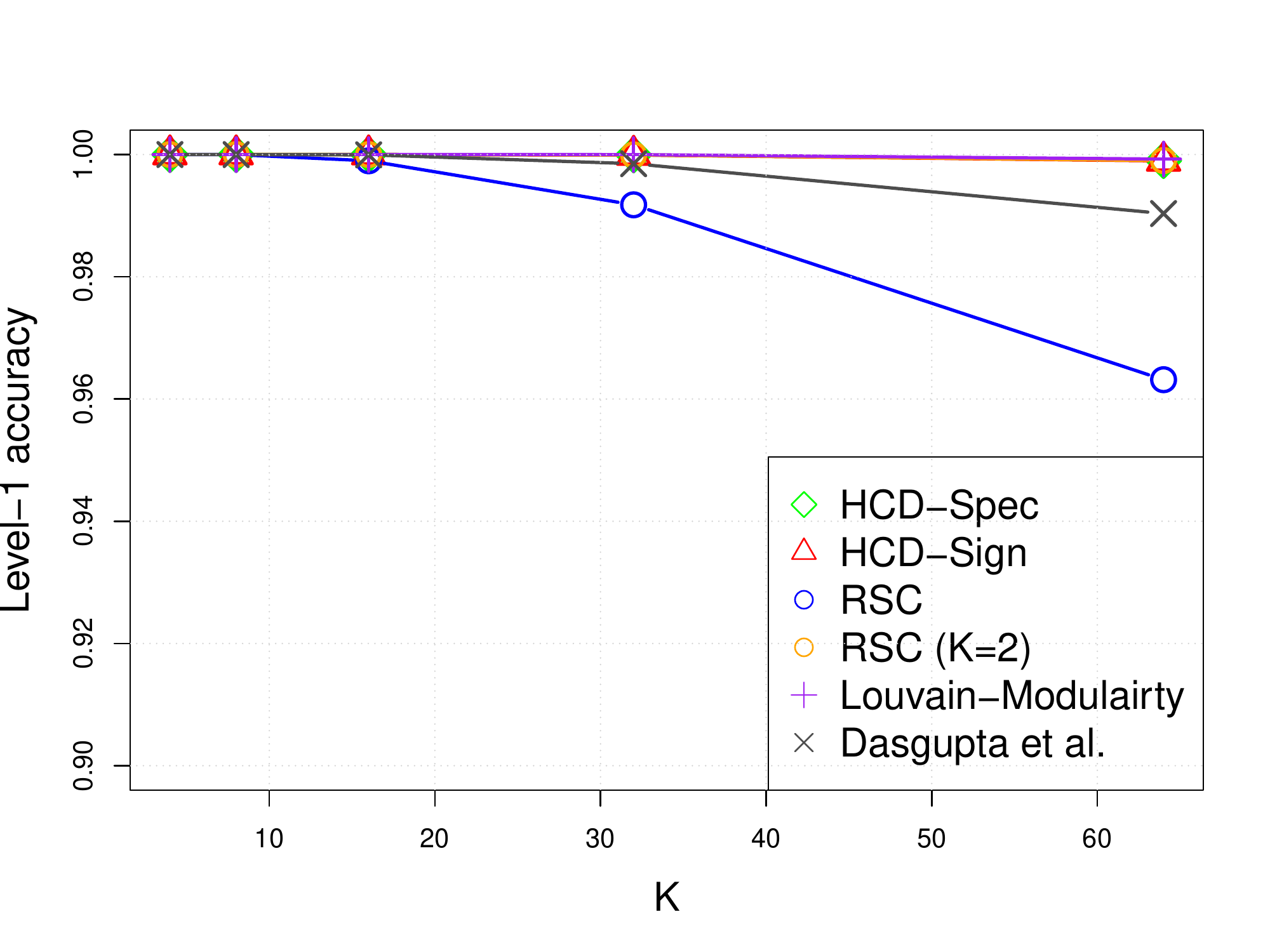}
\caption{Accuracy for level 1}
\label{fig:VaryK-Mega2}
\end{subfigure}%
\hfill
\begin{subfigure}[t]{0.4\textwidth}
\centering
\includegraphics[width=\textwidth]{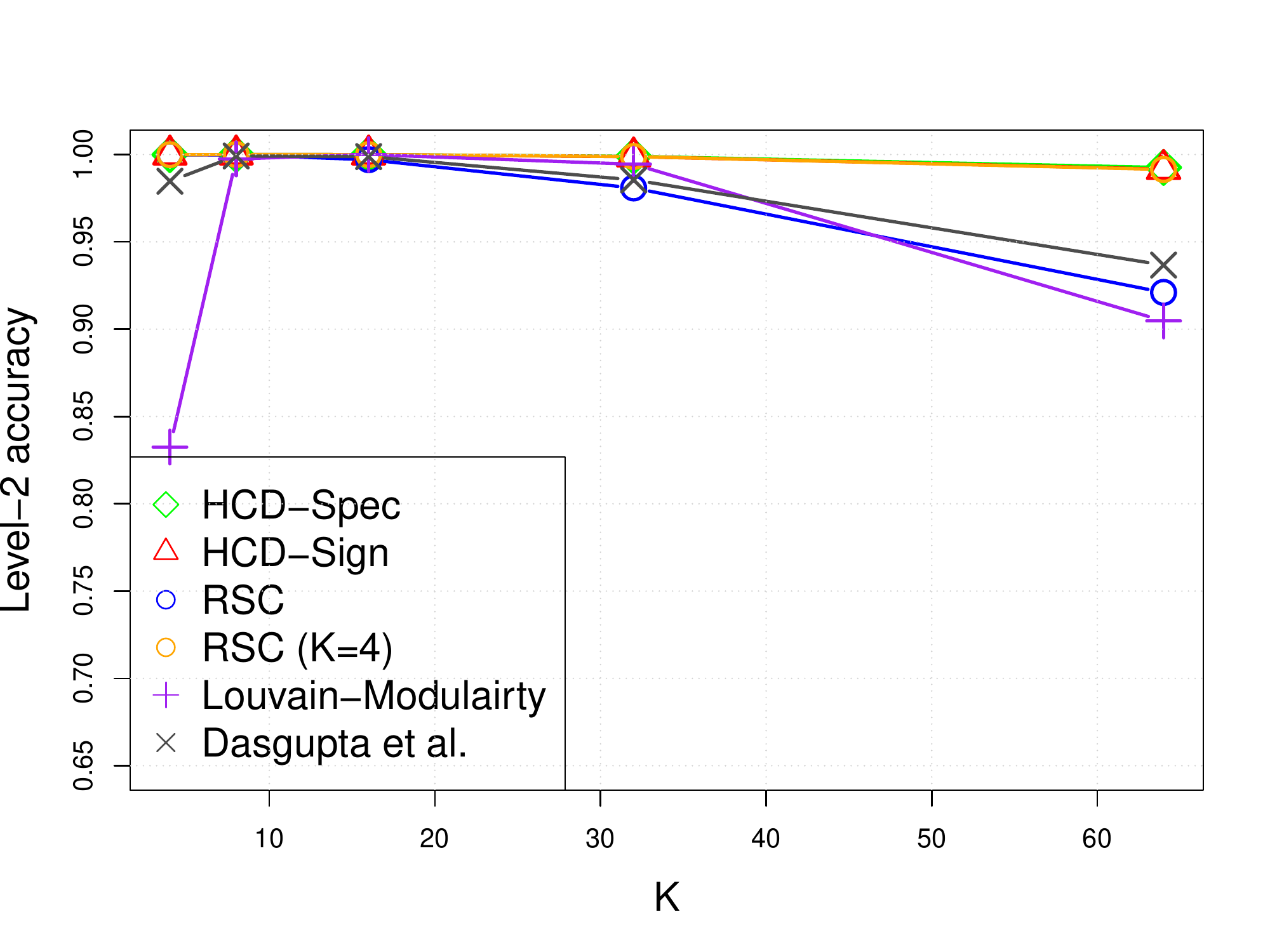}
\caption{Accuracy for level 2}
\label{fig:VaryK-Mega4}
\end{subfigure}

\begin{subfigure}[t]{0.4\textwidth}
\centering
\includegraphics[width=\textwidth]{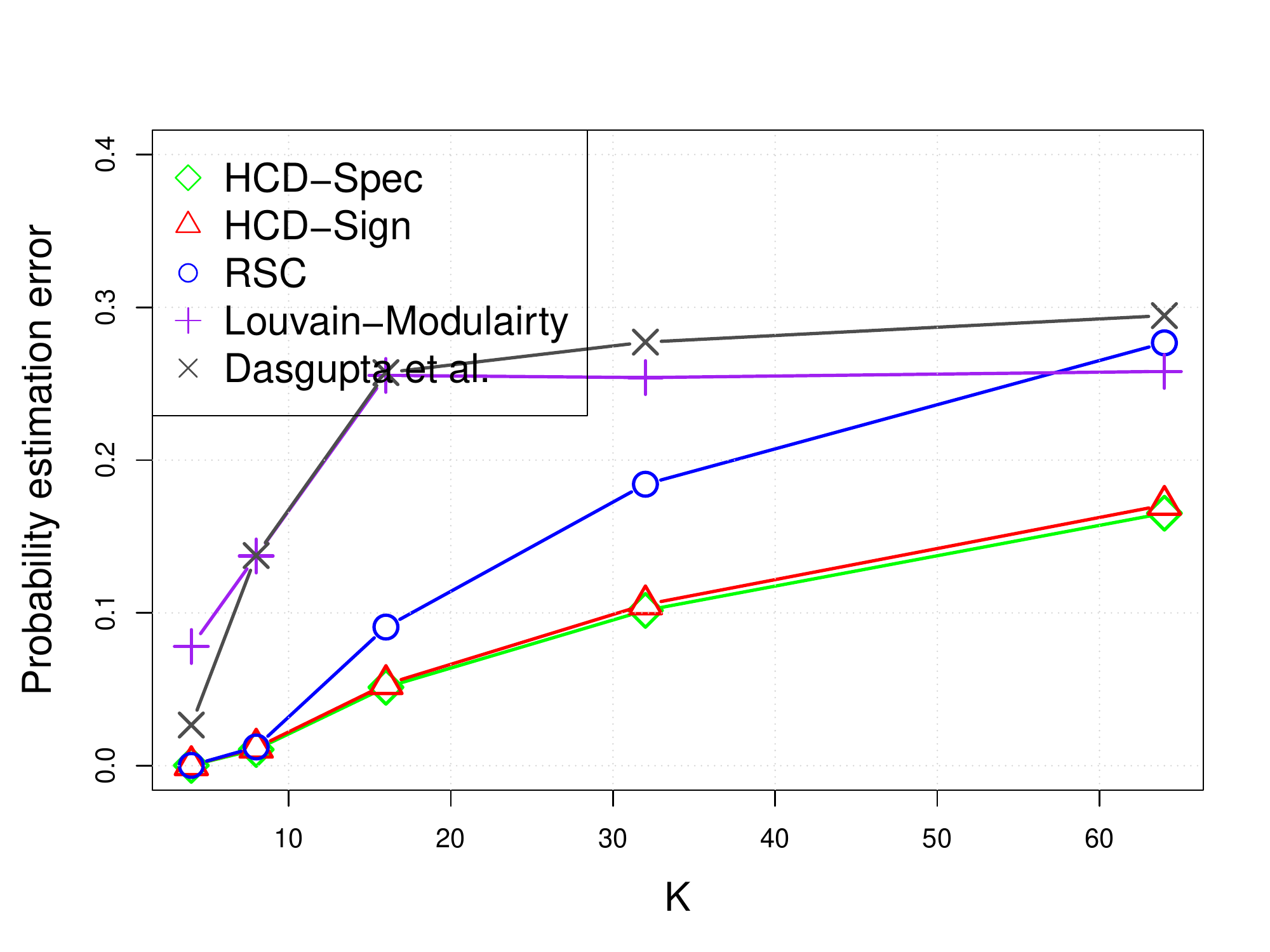}
\caption{Error in the probability matrix}
\label{fig:VaryK-Prob}
\end{subfigure}%
\hfill
\begin{subfigure}[t]{0.4\textwidth}
\centering
\includegraphics[width=\textwidth]{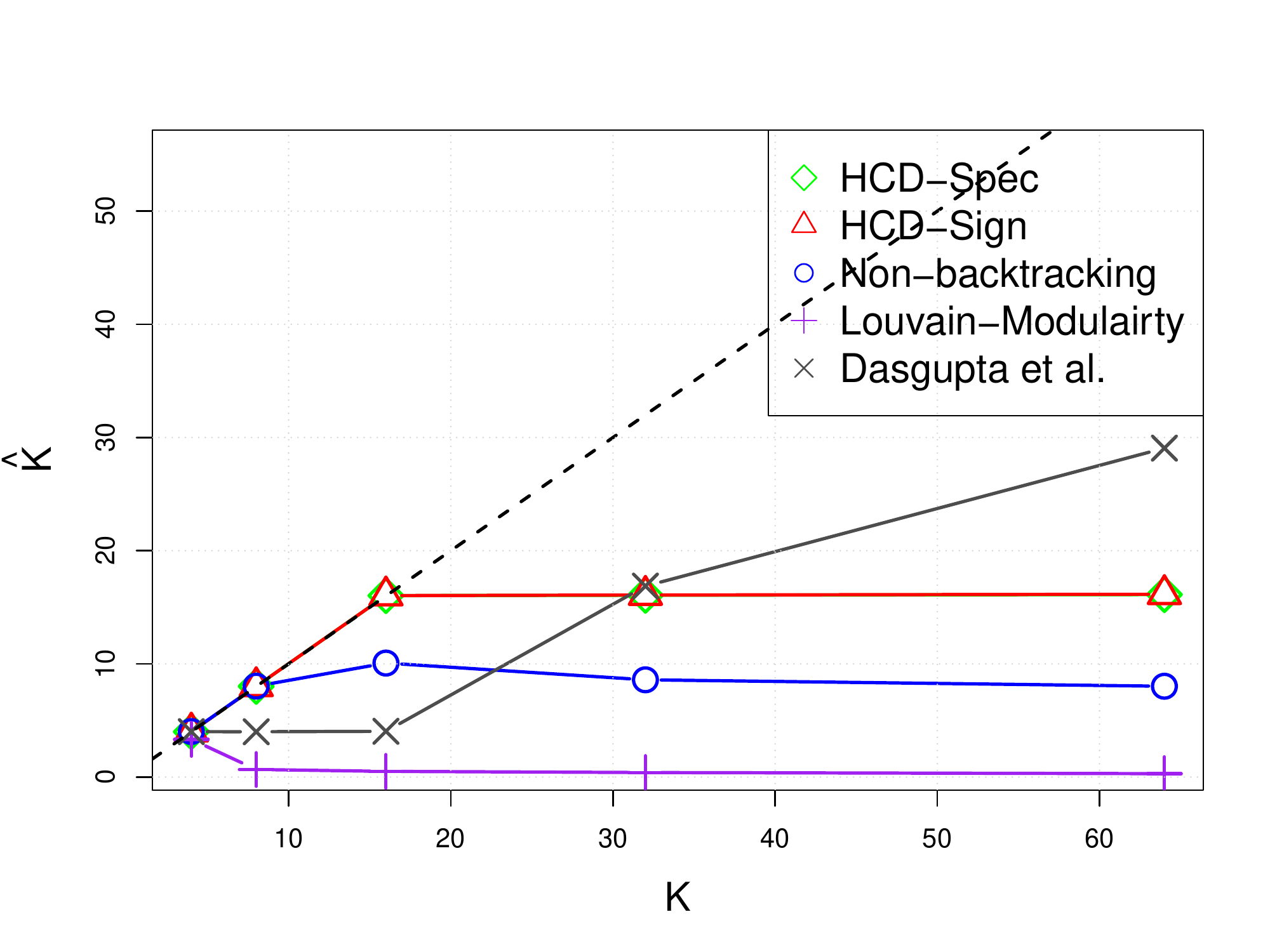}
\caption{Estimated number of communities}
\label{fig:VaryK-hatK}
\end{subfigure}%
\caption{Results for the balanced BTSBM with $n=3200$ nodes and $K$ varying in $\{4, 8, 16, 32, 64\}$.   For (a), (c), and (d) higher values indicate better performance;  for (b) and (e), lower values are better.  For (f), the truth is shown as the 45-degree line.  }
\label{fig:VaryK}
\end{figure}

\begin{figure}[H]
\centering
\begin{subfigure}[t]{0.4\textwidth}
\centering
\includegraphics[width=\textwidth]{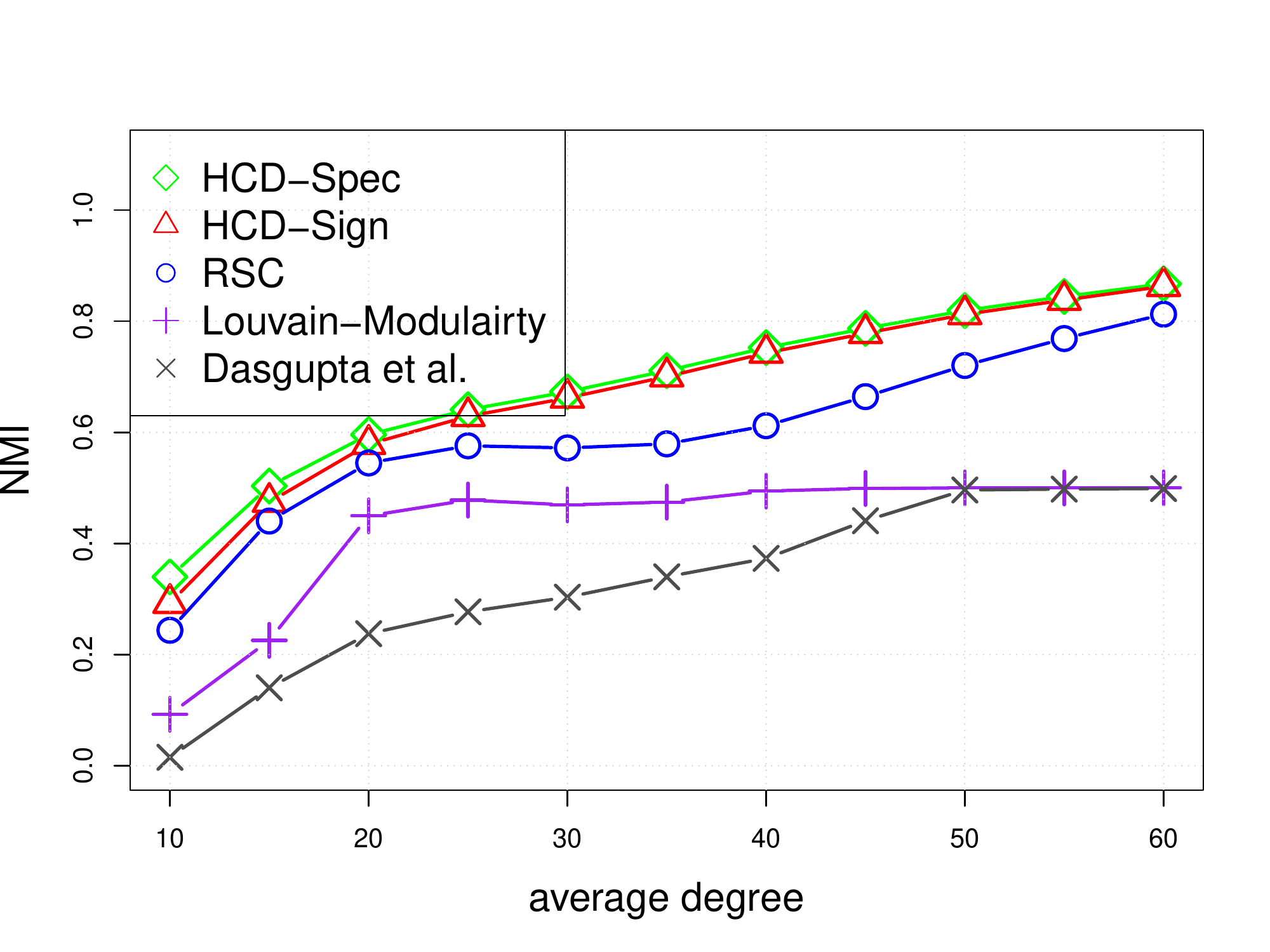}
\caption{Normalized mutual information}
\label{fig:VarySparsity-Coclust}
\end{subfigure}
\hfill
\begin{subfigure}[t]{0.4\textwidth}
\centering
\includegraphics[width=\textwidth]{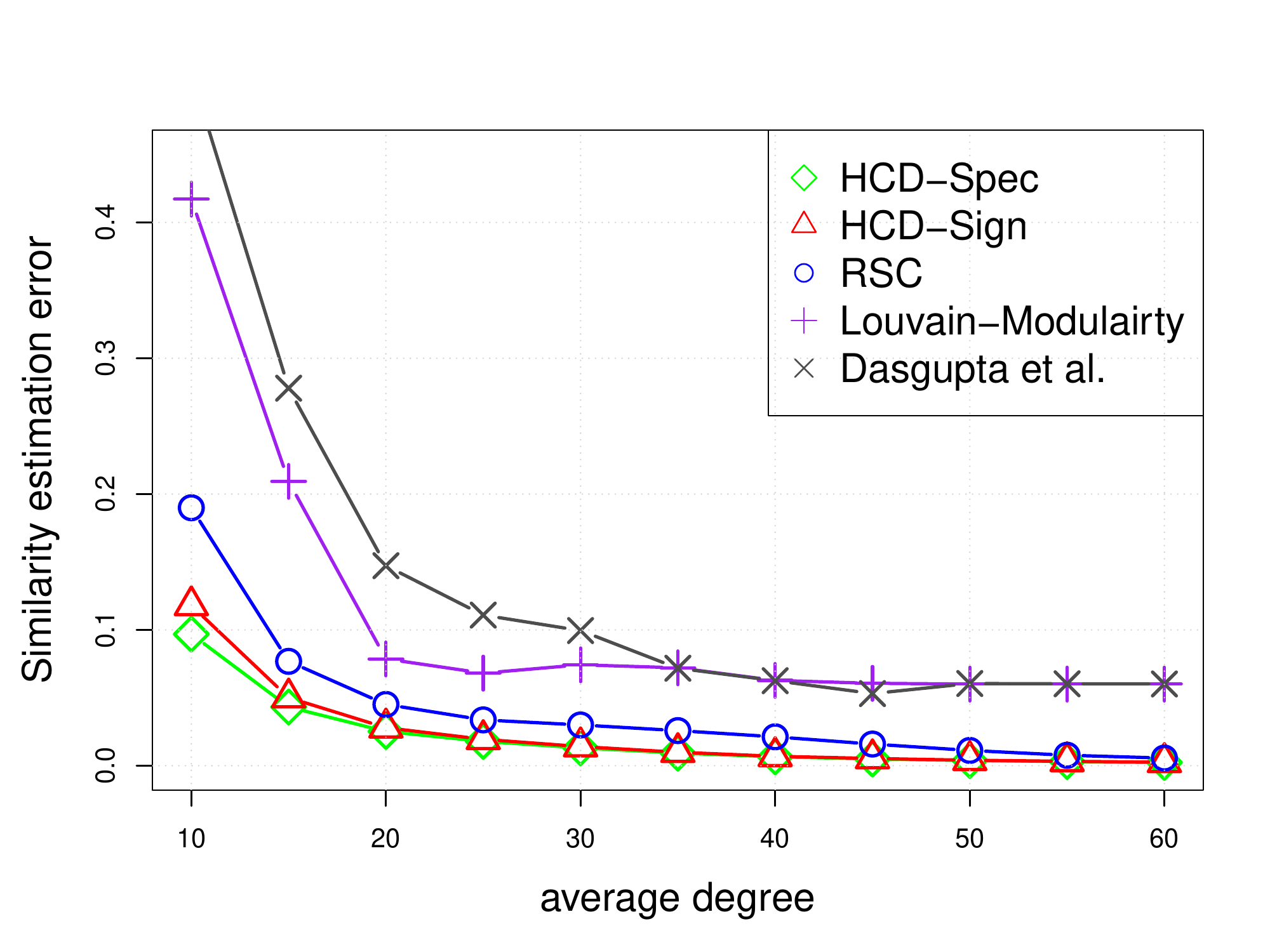}
\caption{Error in the tree similarity matrix}
\label{fig:VarySparsity-SimErr}
\end{subfigure}

\begin{subfigure}[t]{0.4\textwidth}
\centering
\includegraphics[width=\textwidth]{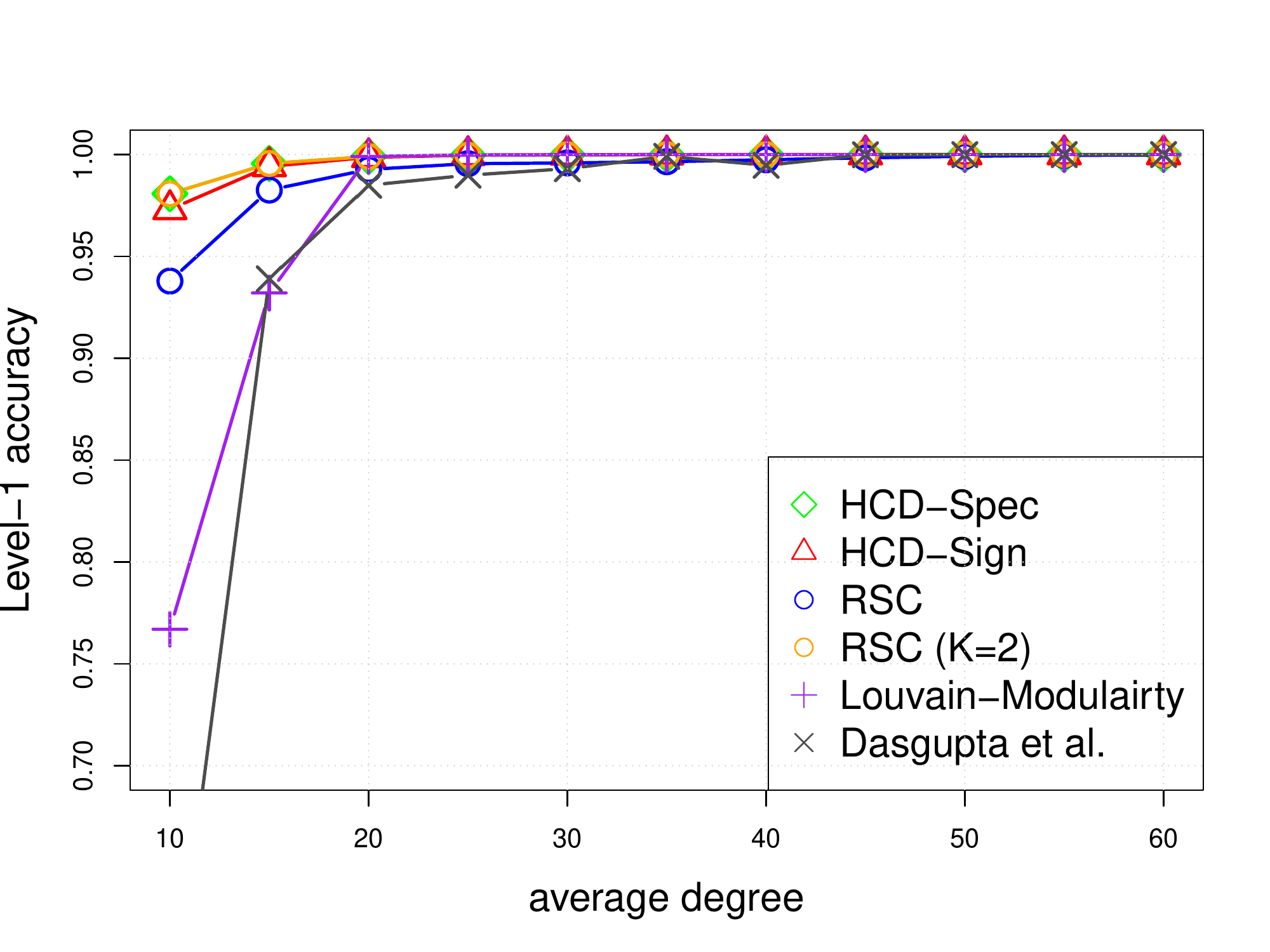}
\caption{Accuracy for  level 1 mega-communities}
\label{fig:VarySparsity-Mega2}
\end{subfigure}%
\hfill
\begin{subfigure}[t]{0.4\textwidth}
\centering
\includegraphics[width=\textwidth]{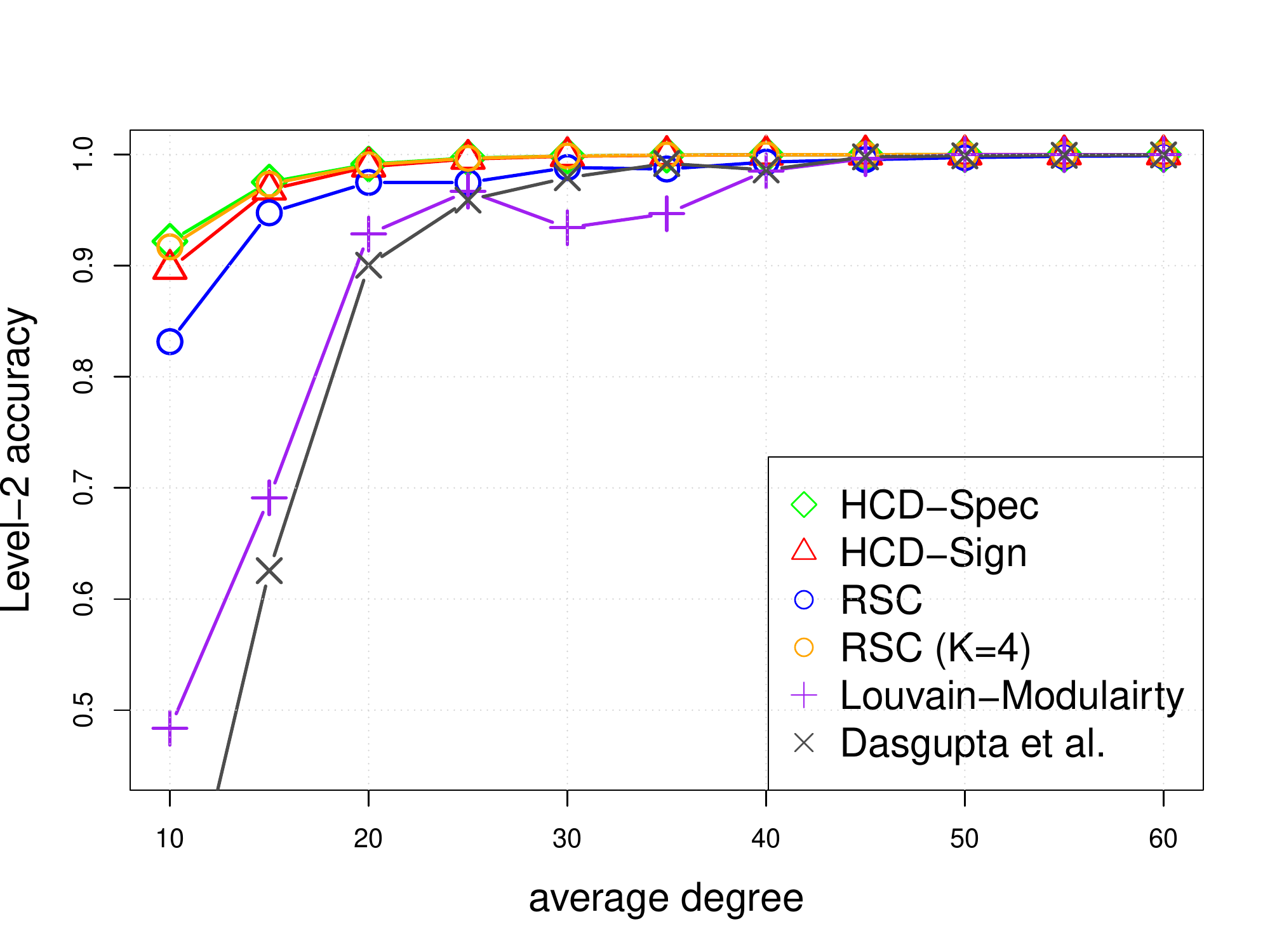}
\caption{Accuracy for level 2 mega-communities}
\label{fig:VarySparsity-Mega4}
\end{subfigure}

\begin{subfigure}[t]{0.4\textwidth}
\centering
\includegraphics[width=\textwidth]{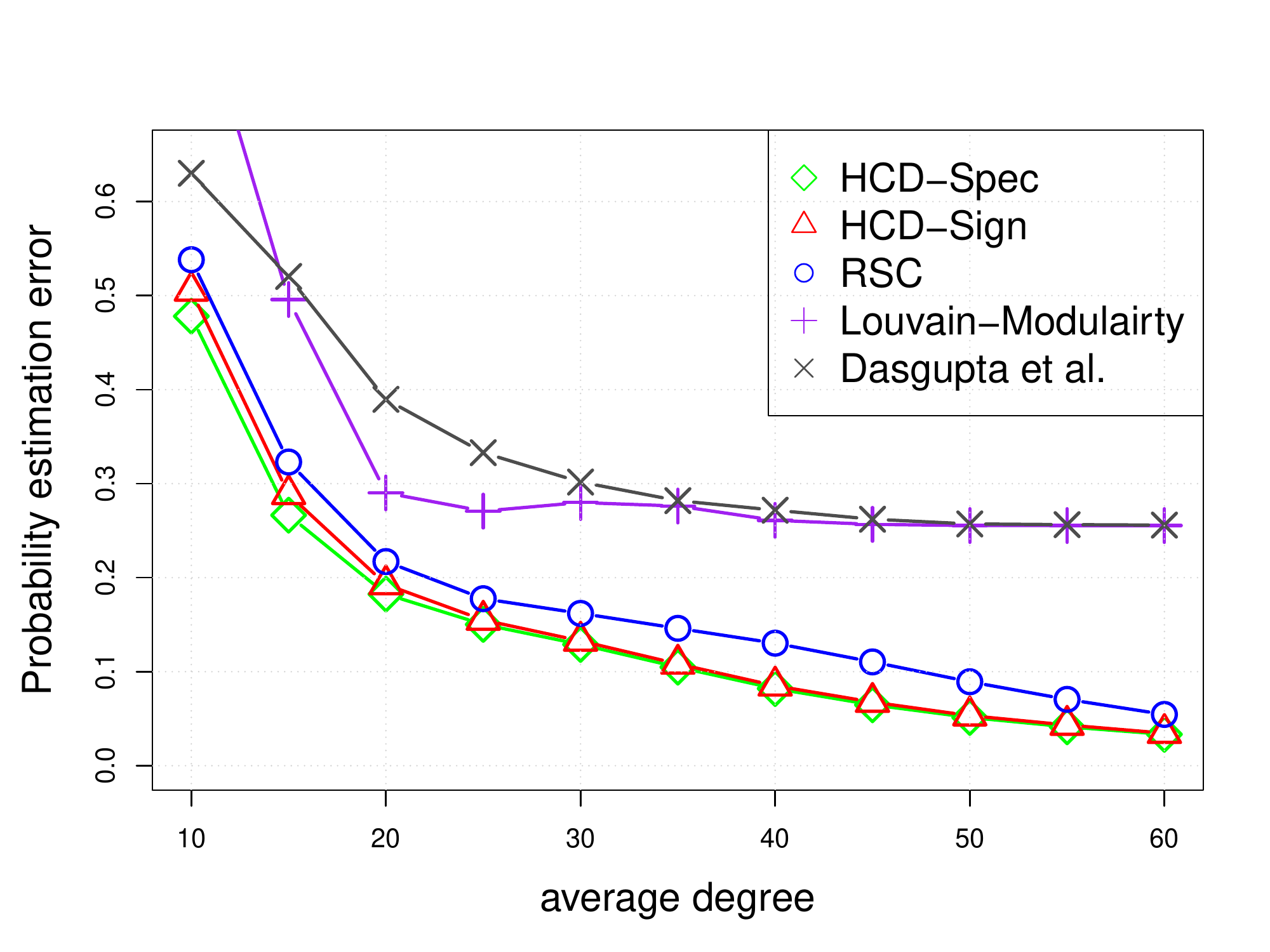}
\caption{Error in the probability matrix}
\label{fig:VarySparsity-Prob}
\end{subfigure}%
\hfill
\begin{subfigure}[t]{0.4\textwidth}
\centering
\includegraphics[width=\textwidth]{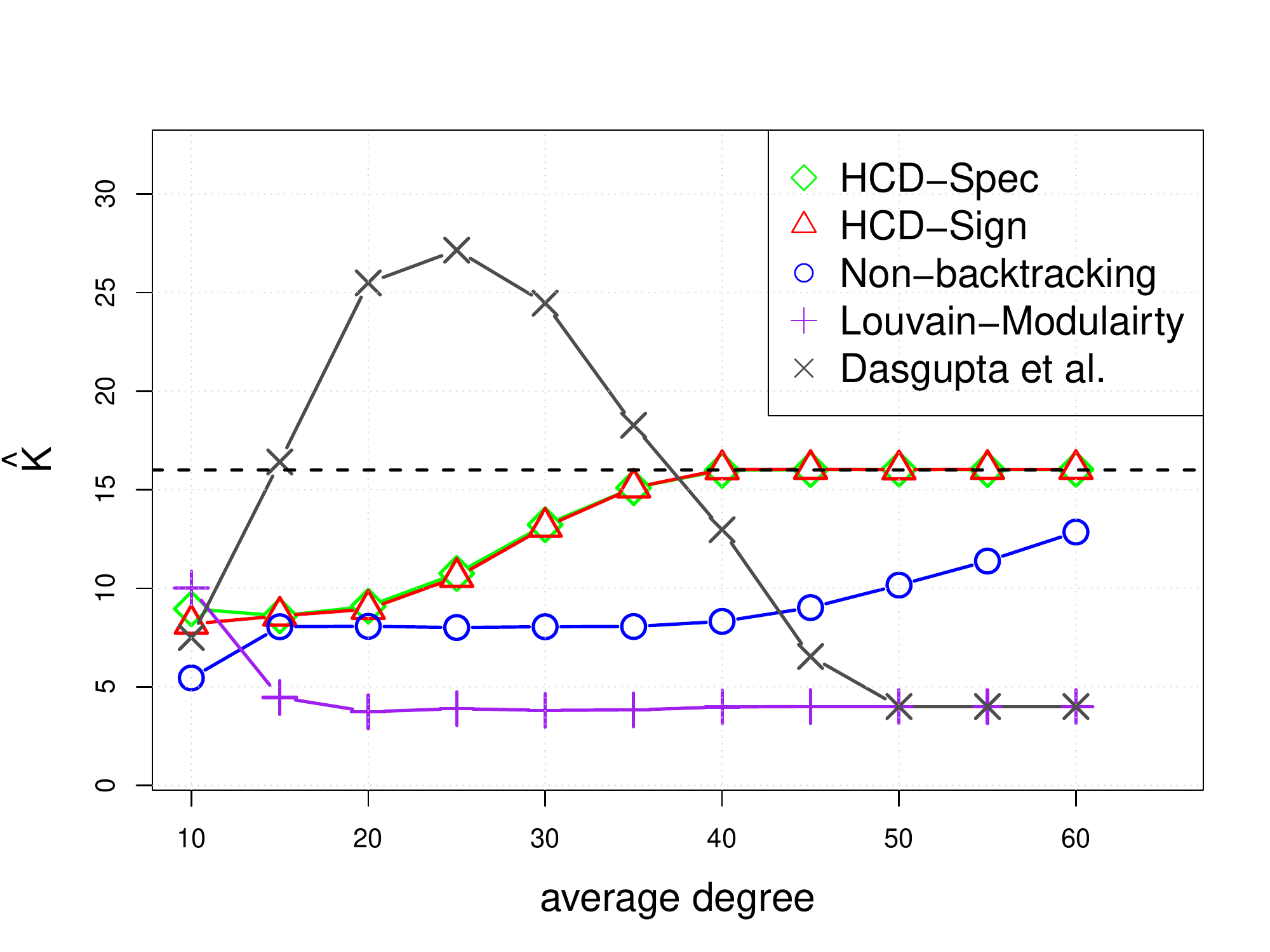}
\caption{Estimated number of communities $\hat{K}$}
\label{fig:VarySparsity-hatK}
\end{subfigure}%
\caption{Results for the balanced BTSBM with $n=3200$ nodes and varying average degree.    For (a), (c), and (d) higher values indicate better performance;  for (b) and (e), lower values are better.  For (f), the truth is shown as the horizontal line.}
\label{fig:VarySparsity}
\end{figure}

\subsection{Unbalanced BTSBM with a complex tree structure}

\begin{figure}[H]
\centering
\begin{subfigure}[t]{0.4\textwidth}
\centering
\includegraphics[width=\textwidth]{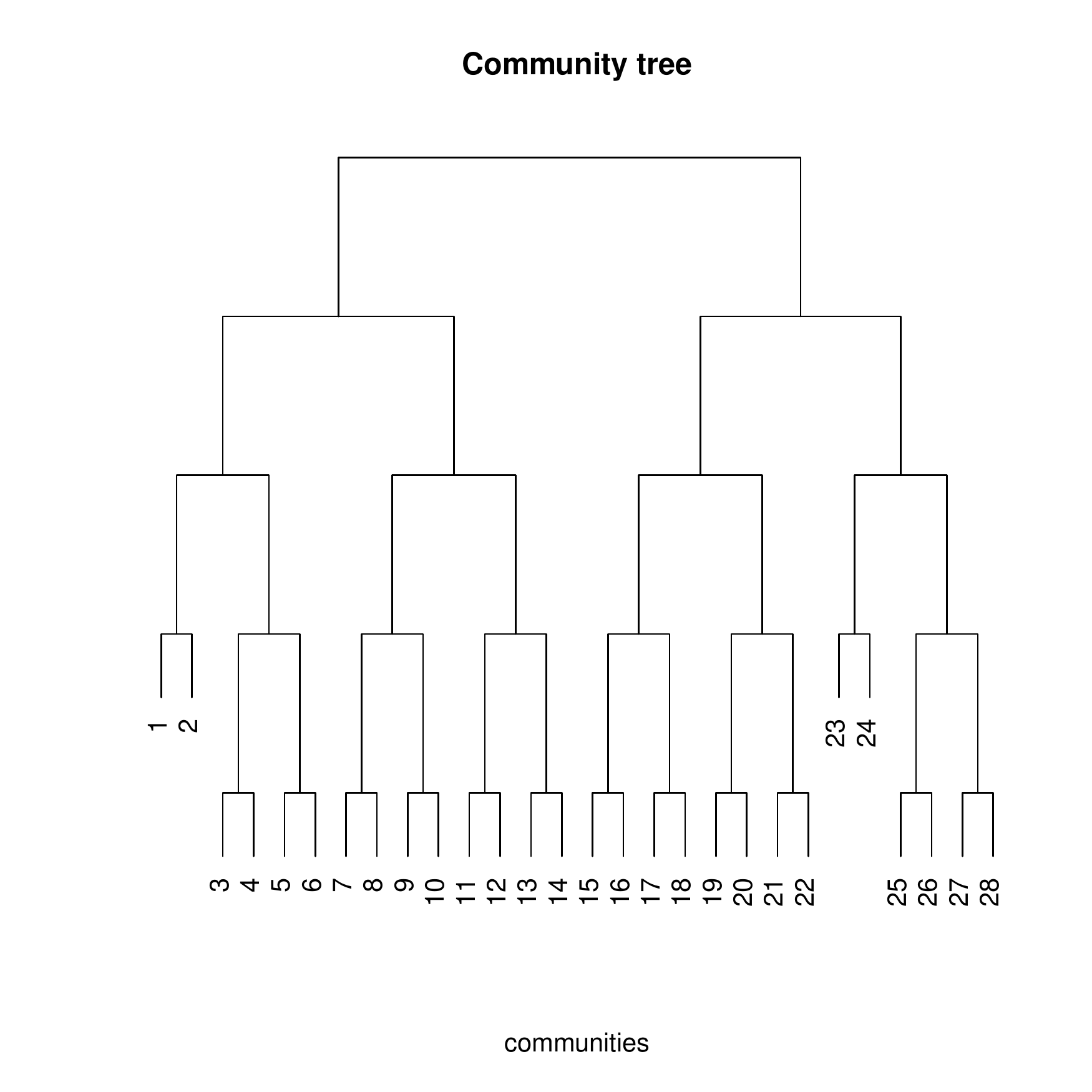}
\caption{Example 1}
\label{fig:Imbalance1}
\end{subfigure}%
\hfill
\begin{subfigure}[t]{0.4\textwidth}
\centering
\includegraphics[width=\textwidth]{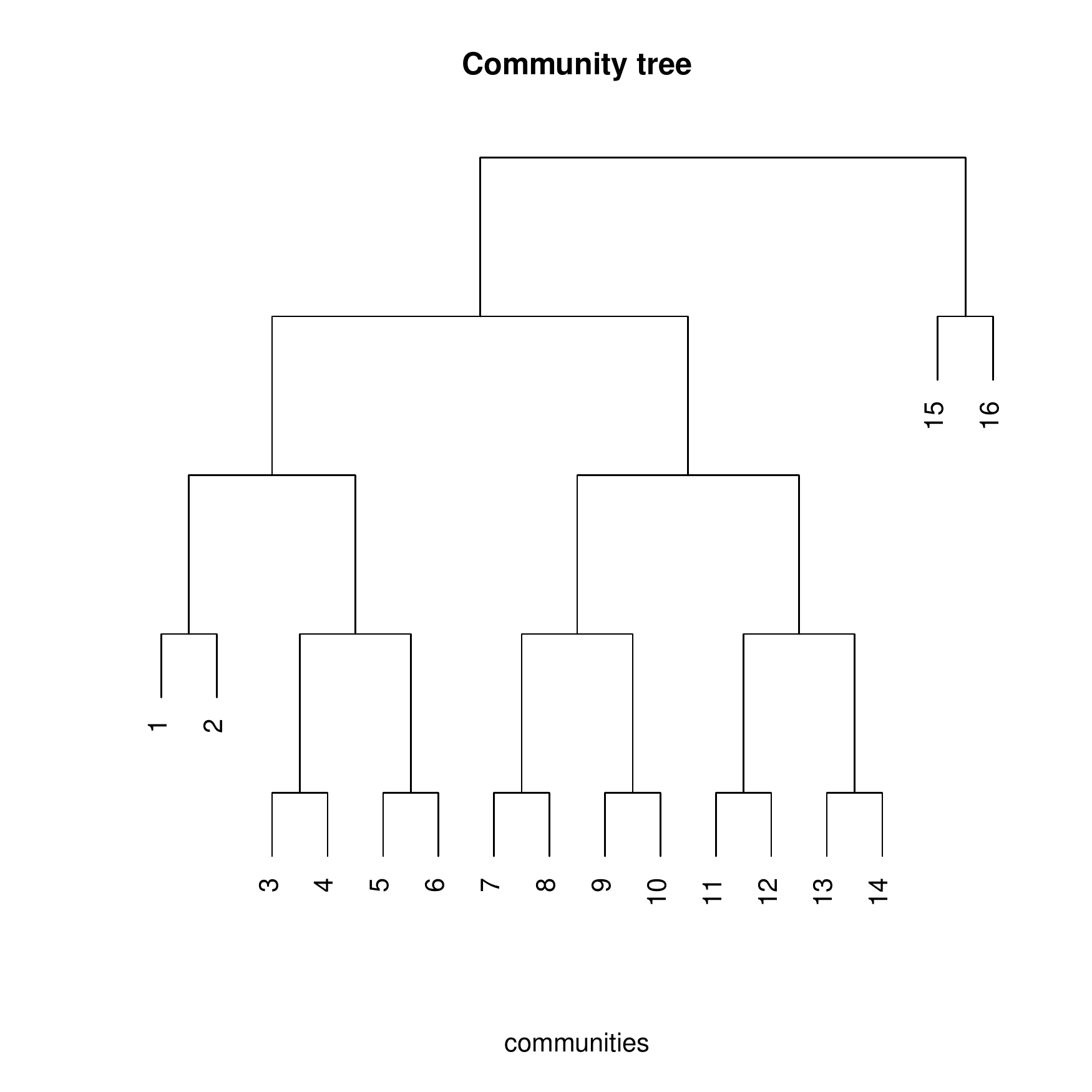}
\caption{Example 2}
\label{fig:Imbalance2}
\end{subfigure}
\caption{Examples of unbalanced community trees.}
\label{fig:Imbalance}
\end{figure}

The BTSBM  gives us the flexibility to generate complex tree structures and communities of varying sizes.    However, it is difficult to control these features with a single parameter such as $K$ or the average degree, so instead we just include two specific examples as an illustration.  The first example corresponds to the hierarchical community structure shown in Figure~\ref{fig:Imbalance1}. It is generated by merging 4 pairs of the original communities from a balanced BTSBM with 32 communities, resulting in $K=28$ total, with 4 communities of 200 nodes each and 24 communities of 100 nodes each.     This is a challenging community detection problem because of the large $K$, and the varying community sizes make it harder.     The second example is shown in Figure~\ref{fig:Imbalance2}. It is generated by merging 2 pairs of leaves one level up, and 8 pairs three levels up in 32 balanced communities again, thus making it even more unbalanced.   This tree has two communities with 800 nodes, two with 200 each, and the remaining 12 communities have 100 nodes each. In both examples,  the average degree is 35.

\begin{table}[ht]
\centering
{\small
\caption{Clustering and model estimation accuracy for Examples 1 and 2}\label{tab:imbalance}
\begin{tabular}{l|l|ccccc}
  \hline
 &Performance metric & HCD-Sign & HCD-Spec & RSC & Modularity &  Dasgupta  \\ 
  \hline
\multirow{5}{*}{Example 1} &  NMI & 0.665 & \textbf{0.677} & 0.552 &0.386 &0.282 \\ 
  &Similarity error& 0.015 & \textbf{0.014} & 0.043 & 0.103 &0.107 \\
&$\hat{P}$ error  & 0.134 &\textbf{0.128} & 0.214 & 0.343 & 0.352 \\ 
& level-1 accuracy  & 0.999 & \textbf{0.999} & 0.992 &0.996 &0.988 \\ 
& level-2 accuracy  & 0.998 & \textbf{0.998} & 0.73 &0.892 &0.973 \\ 
  \hline
\multirow{5}{*}{Example 2} & NMI & 0.753 & \textbf{0.764} & 0.626 & 0.533 &0.335 \\ 
  &Similarity error & 0.025 & \textbf{0.025} & 0.033 & 0.085 & 0.085 \\
  &$\hat{P}$ error & 0.080 & \textbf{0.075} & 0.129 & 0.223 & 0.236 \\ 
& level-1 accuracy  & 0.999 & \textbf{0.999} & 0.993 &0.999 &0.997 \\ 
& level-2 accuracy  & 0.998 & \textbf{0.999} & 0.794 &0.872 &0.982 \\ 
   \hline
\end{tabular}
}
\end{table}

  Table~\ref{tab:imbalance} shows performance for these two examples.   The HCD methods perform better on all tasks, which matches what we observed in balanced settings. All of the methods give reason clustering for level-1 and level-2 mega-communities but HCD methods are clear advantage as one moves down the tree as indicated by NMI, probability estimation error and similarity recovery. Table~\ref{tab:imbalance-K} shows the average selected $\hat{K}$, and the HCD methods are also more accurate in model selection than the others and is inferior to Dasgupta's method for Example 1. However, this only exception may be due to the over-selection of Dasgupta's method we generally observed. 

\begin{table}[ht]
\centering
{\small
\caption{Average $\hat{K}$ in Examples 1 and 2}\label{tab:imbalance-K}
\begin{tabular}{l|l|ccccc}
  \hline
 &Performance metric & HCD-Sign & HCD-Spec & NB & Modularity &  Dasgupta  \\ 
  \hline
Example 1 ($K=28$)&   $\hat{K}$ & 16.19 & 16.27 & 10.45 &3.70 &24.03 \\ 
  \hline
Example 2 ($K=16$)& $\hat{K}$ & 10.11 & 10.11 & 7.43 & 3.53 &23.75 \\ 
   \hline
\end{tabular}
}
\end{table}

\subsection{Networks with no hierarchical communities}

As a benchmark, we also test HCD methods on networks generated from the SBM without a hierarchical structure. These networks are generated from the planted-partition version of the SBM, with $P_{ij} = p$ if $c(i) = c(j)$ and $P_{ij} = \beta p$ if $c(i) \ne c(j)$.    Since all between-community probabilities are the same, no meaningful hierarchy can be defined. To make results fully comparable to the previous settings, we still use $\hat{K}$ estimated by HCD-Spec as input to RSC.

Figures~\ref{fig:SBM-VaryK} and \ref{fig:SBM-VarySparsity} show performance of all the methods under consideration on clustering accuracy (measured by NMI), estimating the probability matrix, and estimating $K$,  Figure~\ref{fig:SBM-VaryK} as a function of $K$ (with average degree set to be 50) and Figure \ref{fig:SBM-VarySparsity} as a function of average degree (with $K = 16$). In this non-hierarchical standard SBM setting, the non-backtracking method estimator gives the best estimate of $K$ and regularized spectral clustering and Louvain modularity have the best clustering accuracy, which is not surprising.  Unlike in hierarchical settings considered so far, HCD-Spec works substantially better than HCD-Sign,  but is still inferior to RSC.    Louvain modularity can be slightly better than RSC in easy settings (small $K$ or high degree), but it degrades quickly as the problem becomes harder. We also observed that the both non-backtracking method and the HCD estimate a small $K$ when the problem becomes really difficult, which could be seen as indirect evidence that there is just not enough signal to find communities in those settings.   

Results in this non-hierarchical setting match empirical observations on clustering in Euclidean space from computer science  \citep{shi2000normalized, kannan2004clusterings}.    In practice, we will not know whether the community structure is hierarchical, although the completely flat structure of the planted partition model is not likely to occur naturally.   However, as Figure~\ref{fig:SBM-VaryK} shows, all methods fail completely with $K > 16$ if there is no hierarchical structure; it might be that for large $K$ community detection is the only possible when there is a hierarchy anyway.  

\begin{figure}[H]
\centering
\begin{subfigure}[t]{0.32\textwidth}
\centering
\includegraphics[width=\textwidth]{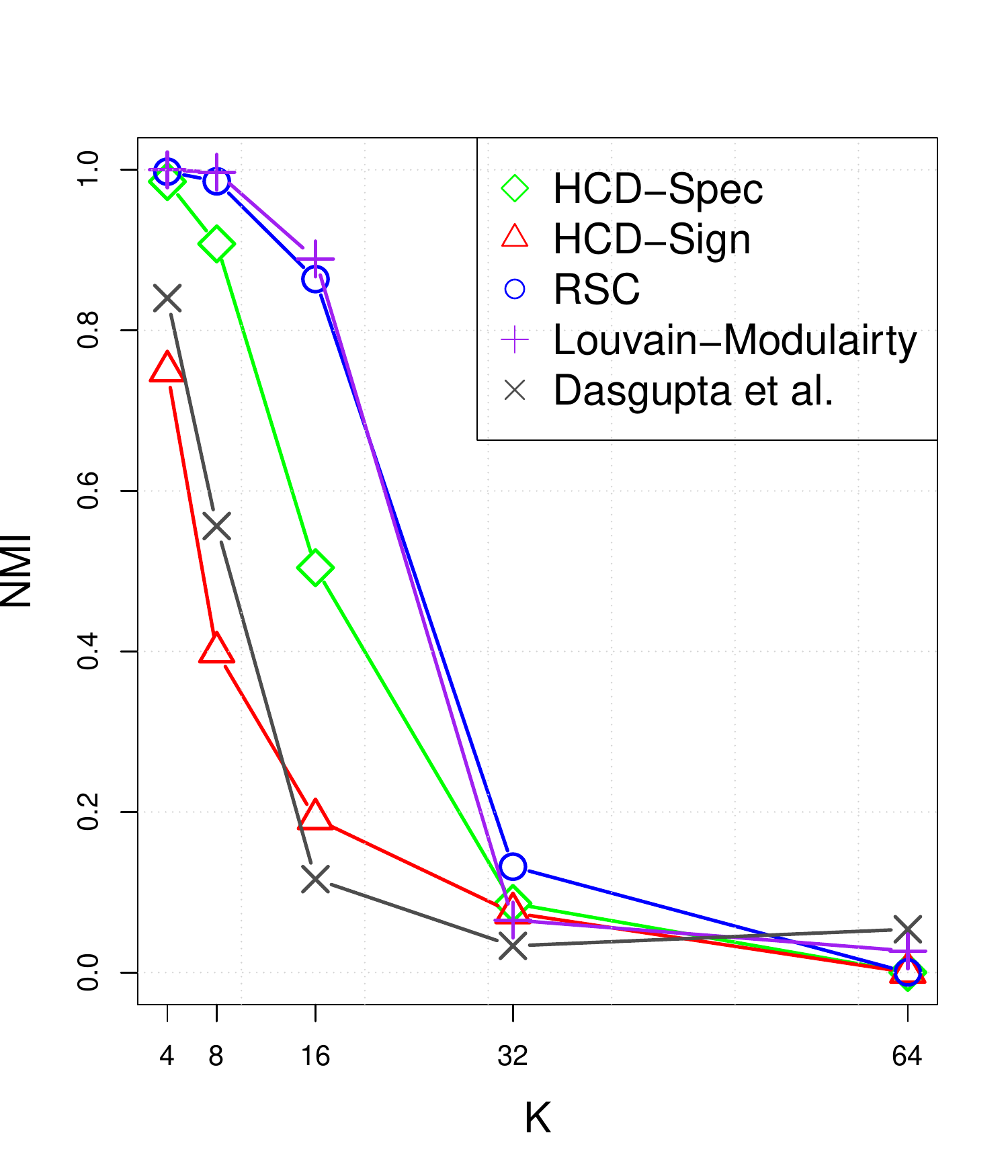}
\caption{NMI}
\label{fig:SBM-VaryK-Coclust}
\end{subfigure}
\hfill
\begin{subfigure}[t]{0.32\textwidth}
\centering
\includegraphics[width=\textwidth]{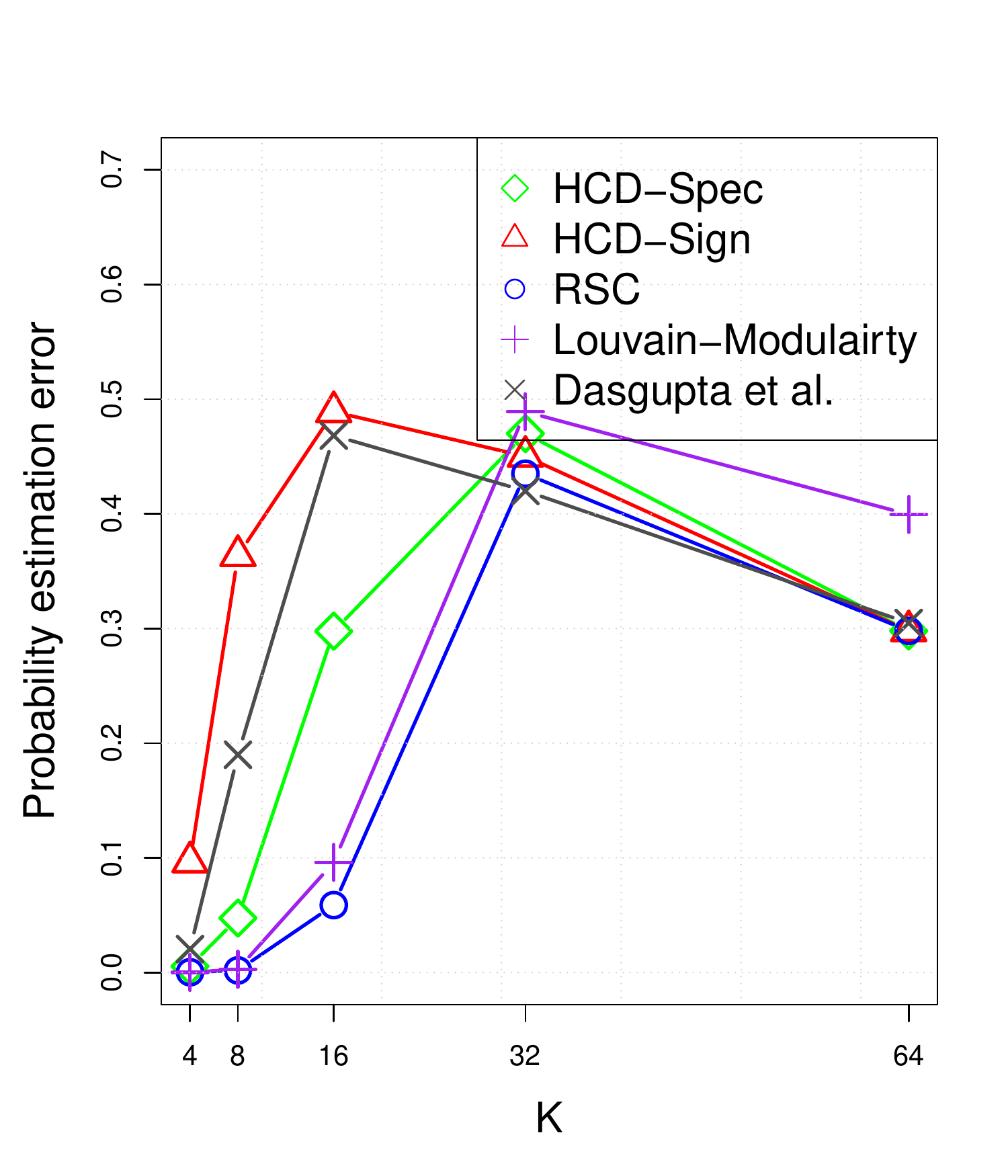}
\caption{Error in $\hat{P}$}
\label{fig:SBM-VaryK-Prob}
\end{subfigure}%
\hfill
\begin{subfigure}[t]{0.32\textwidth}
\centering
\includegraphics[width=\textwidth]{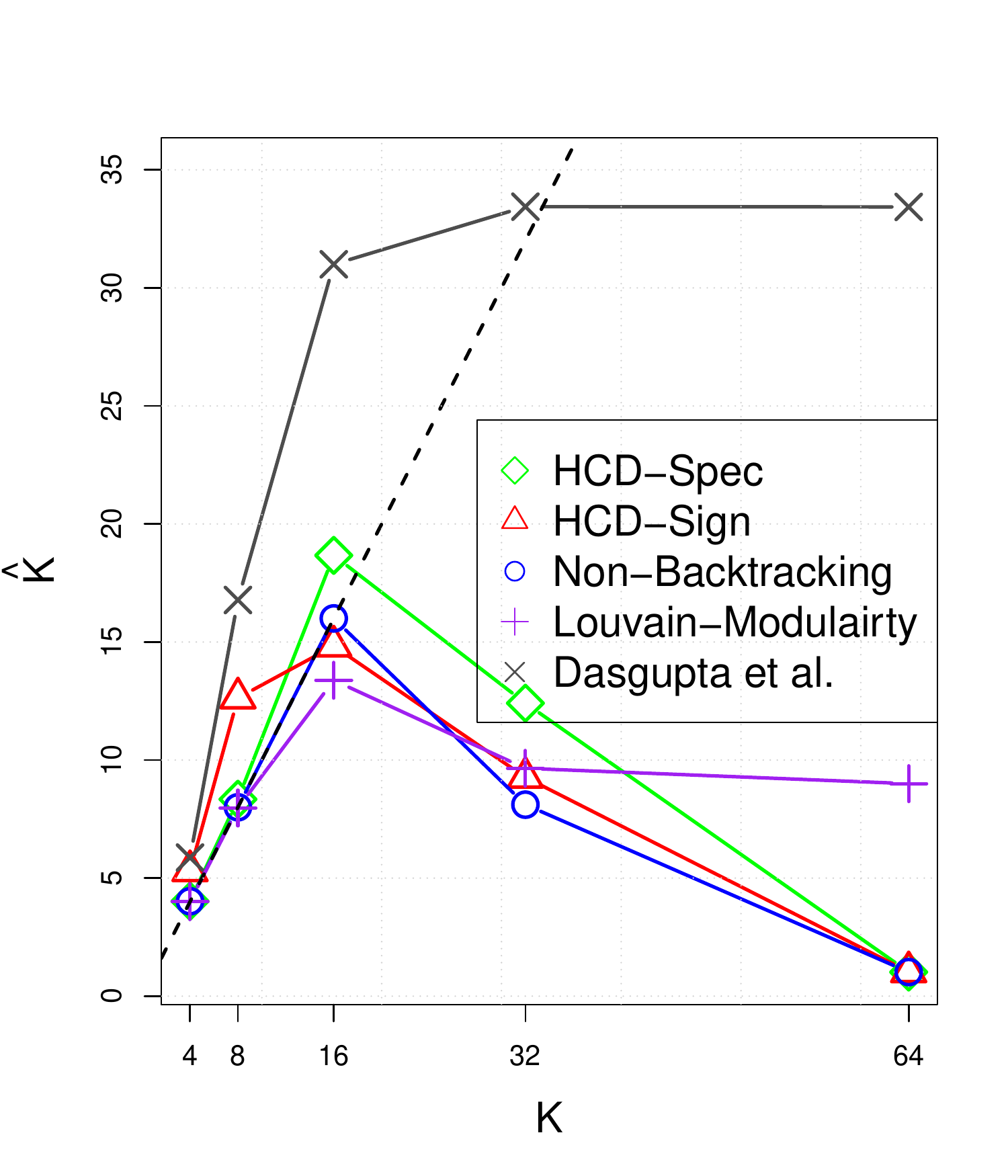}
\caption{$\hat{K}$}
\label{fig:SBM-VaryK-hatK}
\end{subfigure}%
\caption{Results for all methods on the SBM with no hierarchy with $K$ varied and the average degree being 50.}
\label{fig:SBM-VaryK}
\end{figure}

\begin{figure}[H]
\centering
\begin{subfigure}[t]{0.33\textwidth}
\centering
\includegraphics[width=\textwidth]{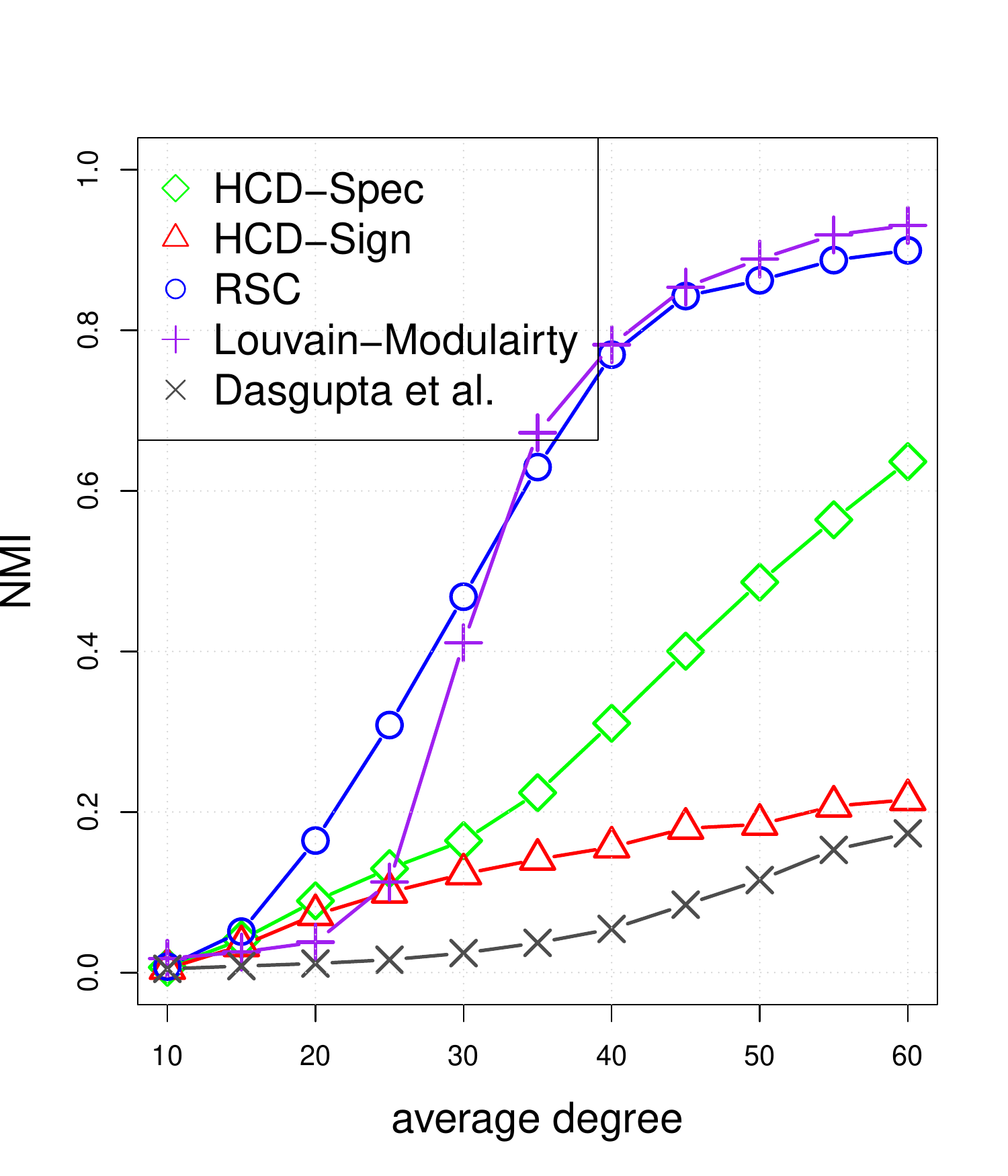}
\caption{NMI}
\label{fig:SBM-VarySparsity-Coclust}
\end{subfigure}
\hfill
\begin{subfigure}[t]{0.33\textwidth}
\centering
\includegraphics[width=\textwidth]{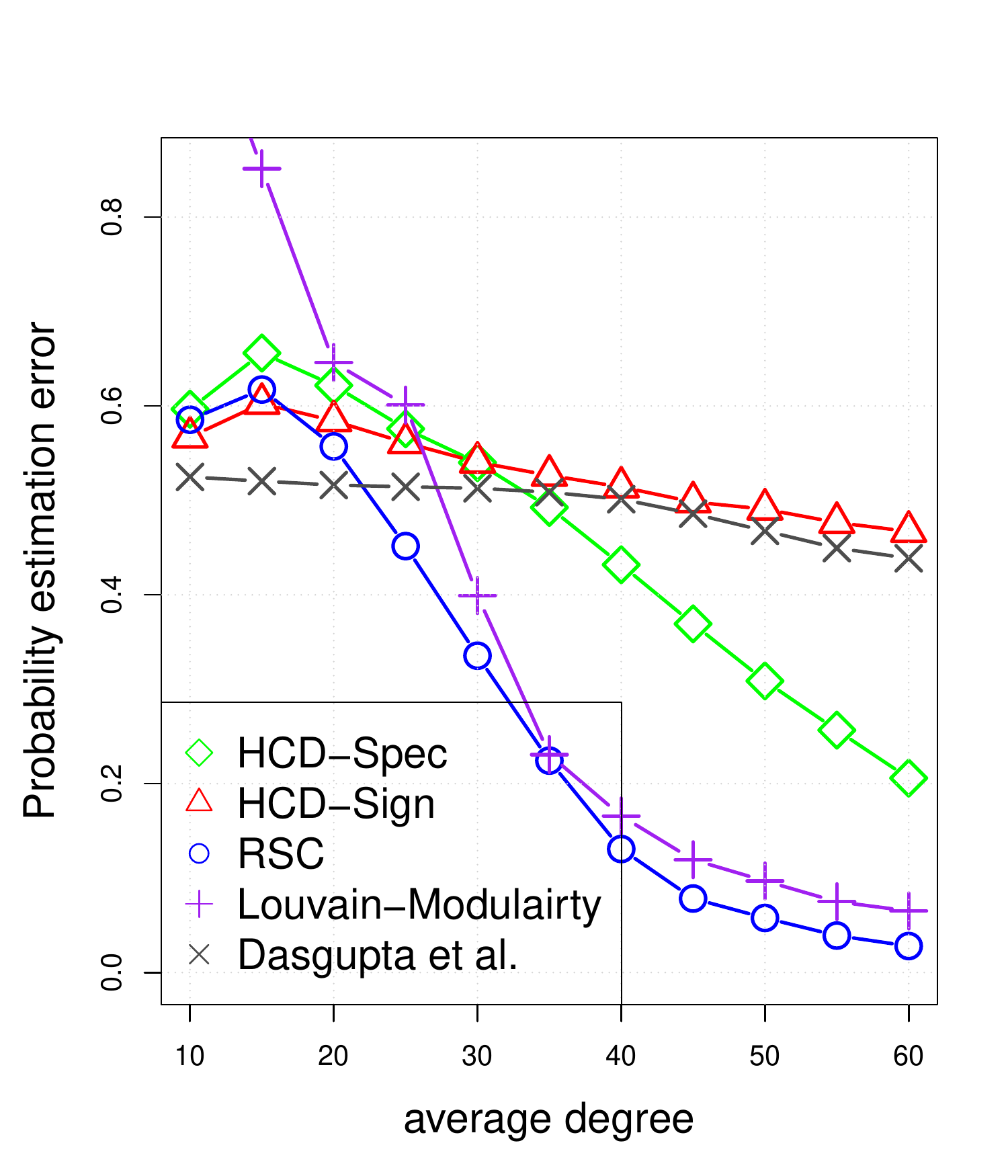}
\caption{Error in $\hat{P}$}
\label{fig:SBM-VarySparsity-Prob}
\end{subfigure}%
\hfill
\begin{subfigure}[t]{0.33\textwidth}
\centering
\includegraphics[width=\textwidth]{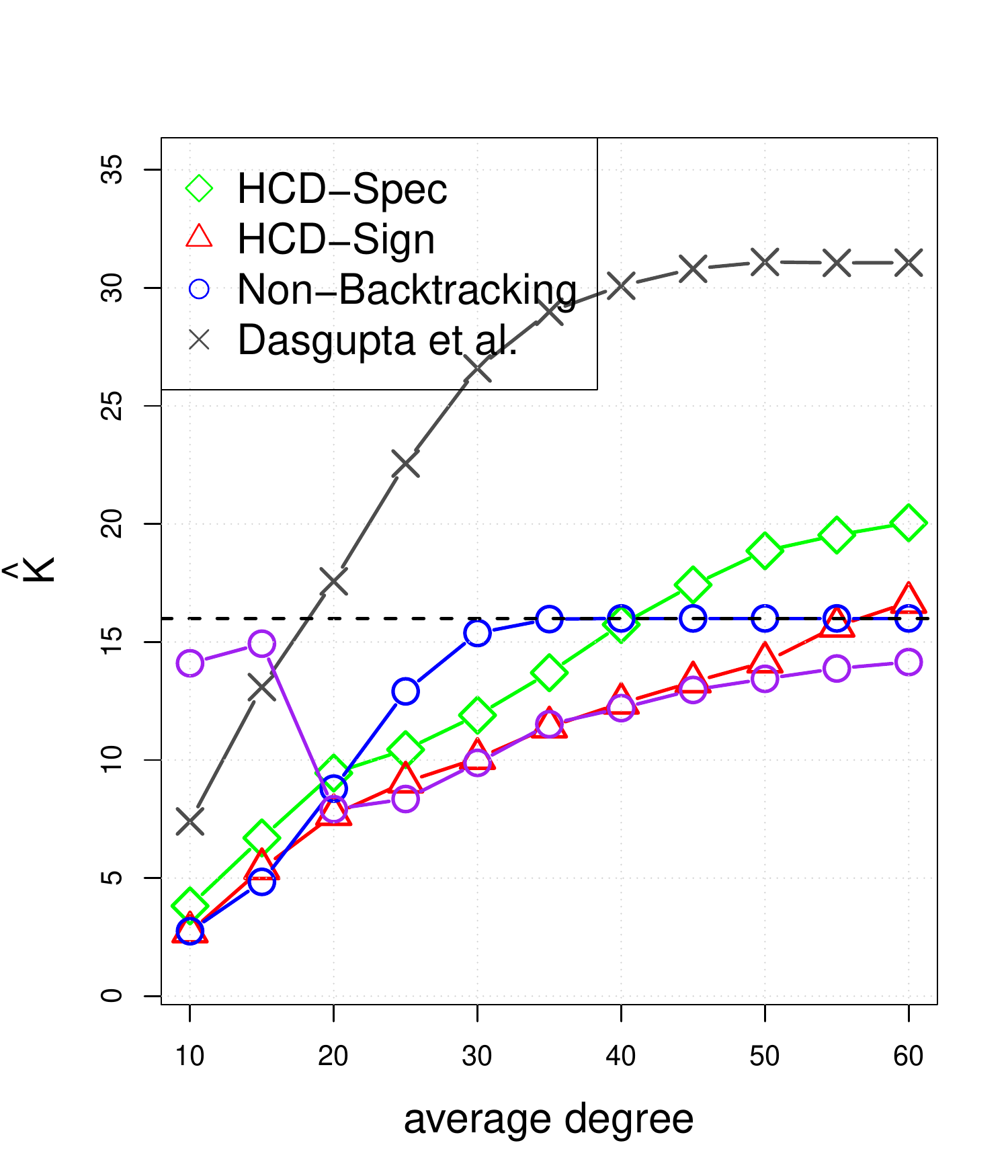}
\caption{$\hat{K}$}
\label{fig:SBM-VarySparsity-hatK}
\end{subfigure}%

\caption{Results for all methods on the SBM with no hierarchy with average degree varied. $K$ is fixed to be 16 in this example.}
\label{fig:SBM-VarySparsity}
\end{figure}


 \section{Hierarchical communities in genes associated with anemia}\label{sec:app}

We illustrate the HCD method by fitting hierarchical communities to a network of genes that have been found to associate with anemia.  Another detailed example of hierarchical communities of topics in statistical research literature is included in Appendix~\ref{sec:stat-net}.  

Anemia is a blood disease that is caused by either a decrease in the number of red blood cells, or a general lower capability of the blood for oxygen transport.  Arguably the most common blood disorder, it was estimated by the World Health Organization to affect roughly a quarter of humans globally in 2005 \citep{de2008worldwide}. Causes for anemia include acute blood loss, an increased breakdown of red blood cell populations or a decreased production of red blood cells. The causes for an increased breakdown of red blood cells also include genetic conditions, e.g., sickle-cell anemia. Here, we investigate associations between genes in the context of anemia using the DigSee framework \citep{kim2017analysis}, which identifies disease-related genes from Medline abstracts. We look for associations of anemia with biological events related to mutation, gene expression, regulation and transcription. A total of 4449 papers were found that link genes to anemia. We only retain papers that link at least two genes to anemia, resulting in a set of 1580 papers that relate anemia to the biological events defined above. In total, 1657 genes are mentioned across these papers. We represent the gene co-occurrence information by a network in which two genes are connected if they were mentioned together in at least two papers. Following \cite{wang2016discussion}, we focus on the 3-core of the largest connected component of the network, ignoring the periphery which typically carries little information about the structure of the core. The resulting network, shown in Figure 7a, has 140 nodes (genes) and the average node degree is 7.03.    The hierarchical clustering result on the 140 genes is robust to the core extraction step, since it is similar to the result on the largest connected component of the initial network (see Appendix~\ref{sec:RSC-gene}).

The community labels and hierarchy returned by the HCD-Sign algorithm are given in Figures~\ref{fig:GeneNetwork} and \ref{fig:GeneTree}, respectively. The gene communities turn out to be associated with different types of anemia or different processes causing the disease. We use the MSigDB database \citep{liberzon2011molecular} to perform enrichment analysis by computing overlaps with known functional gene sets, as a way to interpret the communities. Table~\ref{tab:geneSets} provides a summary of high-level interpretations.

\begin{figure}[H]
\centering
\begin{subfigure}[t]{0.5\textwidth}
\centering
\includegraphics[width=\textwidth]{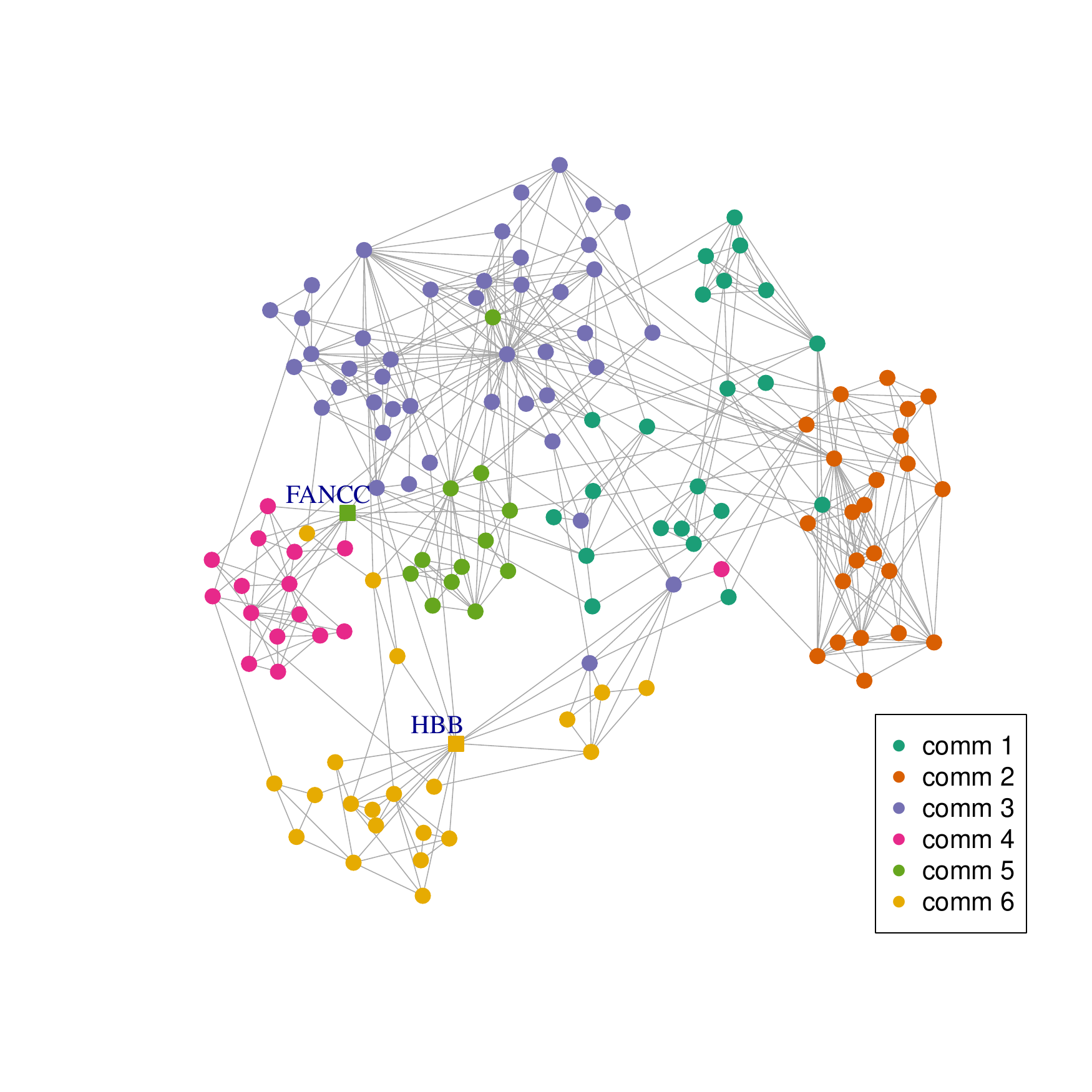}
\caption{The network and detected communities.}
\label{fig:GeneNetwork}
\end{subfigure}%
\hfill
\begin{subfigure}[t]{0.5\textwidth}
\centering
\includegraphics[width=\textwidth]{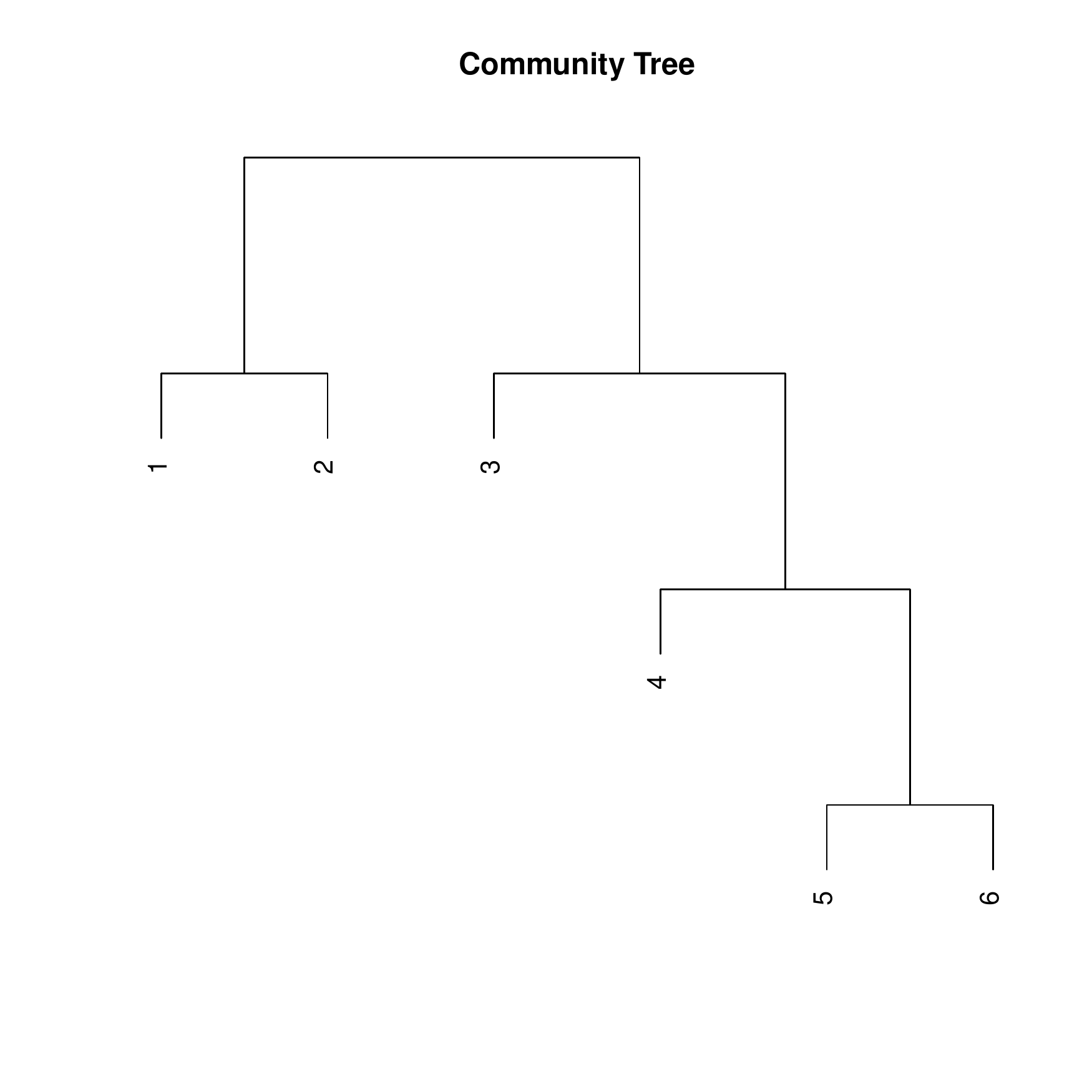}
\caption{Estimated hierarchy of communities.}
\label{fig:GeneTree}
\end{subfigure}
\caption{Hierarchical communities in a network of 140 genes associated with anemia.}
\label{fig:AnemiaResult}
\end{figure}

The left branch of the dendrogram consists of communities 1 and 2. The major biological processes related to these communities are both involved in the activation of cells due to exposure to a factor, leading to the activation of immune cells and hence contributing to an immune response \citep{liberzon2011molecular}. Indeed, many genes in these communities are associated with ``cluster of differentiation molecules", which are genes coding for cell surface markers that can be used by the immune system to activate cells and communicate, amongst others \citep{zola2007cd}.

The right branch of the hierarchical splitting tree consists of four communities resulting from three sequential splits. Community 3 contains several genes coding for interleukins, the signaling molecules that are mainly used by the immune system, and other signaling related genes such as MAP kinases. The most significant gene set for this community, response to cytokine, confirms that these genes are related to molecular signaling within the immune system. Community 4 is well separated in the network. Nine out of the 16 genes in the community are related to the FANC protein core complex, a set of proteins that have been linked to fanconi anemia. The fanconi anemia pathway is involved in removing DNA interstrand crosslinks (ICL), a form of DNA repair, during DNA replication and transcription \citep{ceccaldi2016fanconi}. Studies have related the mutations in these genes (typically FANCA, FANC and FANCG) to bone marrow failure \citep{schneider2017gli1+}. Related, the most significantly enriched gene ontology term  of this community is DNA repair. The conjecture that the genes in community 4 are related to fanconi anemia can be verified by checking the keywords in the titles. Specifically, for each community, we define its ``signal'' papers to be the set of papers that link at least two genes within the same community. Within this set, we found 68\% of the signal papers for community 4 have the word ``fanconi" in their titles, significantly more than other communities (see Figure~\ref{fig:Fanconi}). The most significantly enriched biological process for genes in community 5 is related to cytokine signaling. It contains one of the core genes of the FANC protein core complex (the gene ``FANCC" in Figure \ref{fig:GeneNetwork}), which interestingly connects many FANC genes with the other genes in the network, therefore acting as a connection between the fanconi anemia community and the remainder of the network. Finally, the genes in community 6 are related to oxygen transport. Note that impaired oxygen transport can be a result of insufficient healthy red blood cells to transport oxygen around the body. Sickle-cell anemia, caused by a mutation in the Hemoglobin- Beta (HBB) gene, results in red blood cells with a sickle-like morphology that impairs their oxygen transport. We hypothesize that the genes in this community are thus related to sickle-cell anemia. Indeed, the sickle-cell anemia  causative gene, HBB, is a hub in community 6 (see Figure \ref{fig:GeneNetwork}). Correspondingly, the fraction of signal papers for community $6$ that mention the word `sickle' or `thalassemia' is significantly higher than in other communities (see Figure \ref{fig:Sickle}).
\begin{table}
\caption{\label{tab:geneSets} Most significantly enriched gene sets for each community.}
\begin{center}
\begin{tabular}{ c|l } 
 \hline
 Community & Most significantly enriched gene set  \\ \hline
1 & Cell activation involved in immune response  \\ 
2 & Cell activation  \\ 
3 & Response to cytokine\\
 4 & DNA repair\\
5 & Cytokine mediated signaling pathway\\
6 & Oxygen transport\\
 \hline
\end{tabular}
\end{center}
\end{table}

For comparison, we also applied regularized spectral clustering to this the dataset, with results reported in Appendix~\ref{sec:RSC-gene}. The clustering labels of the RSC match 72\% of the HCD labels. However, community sizes from RSC are very unbalanced and interpretation is more difficult, and of course there is no hierarchy. 
\begin{figure}[H]
\centering
\begin{subfigure}[t]{0.5\textwidth}
\centering
\includegraphics[width=\textwidth]{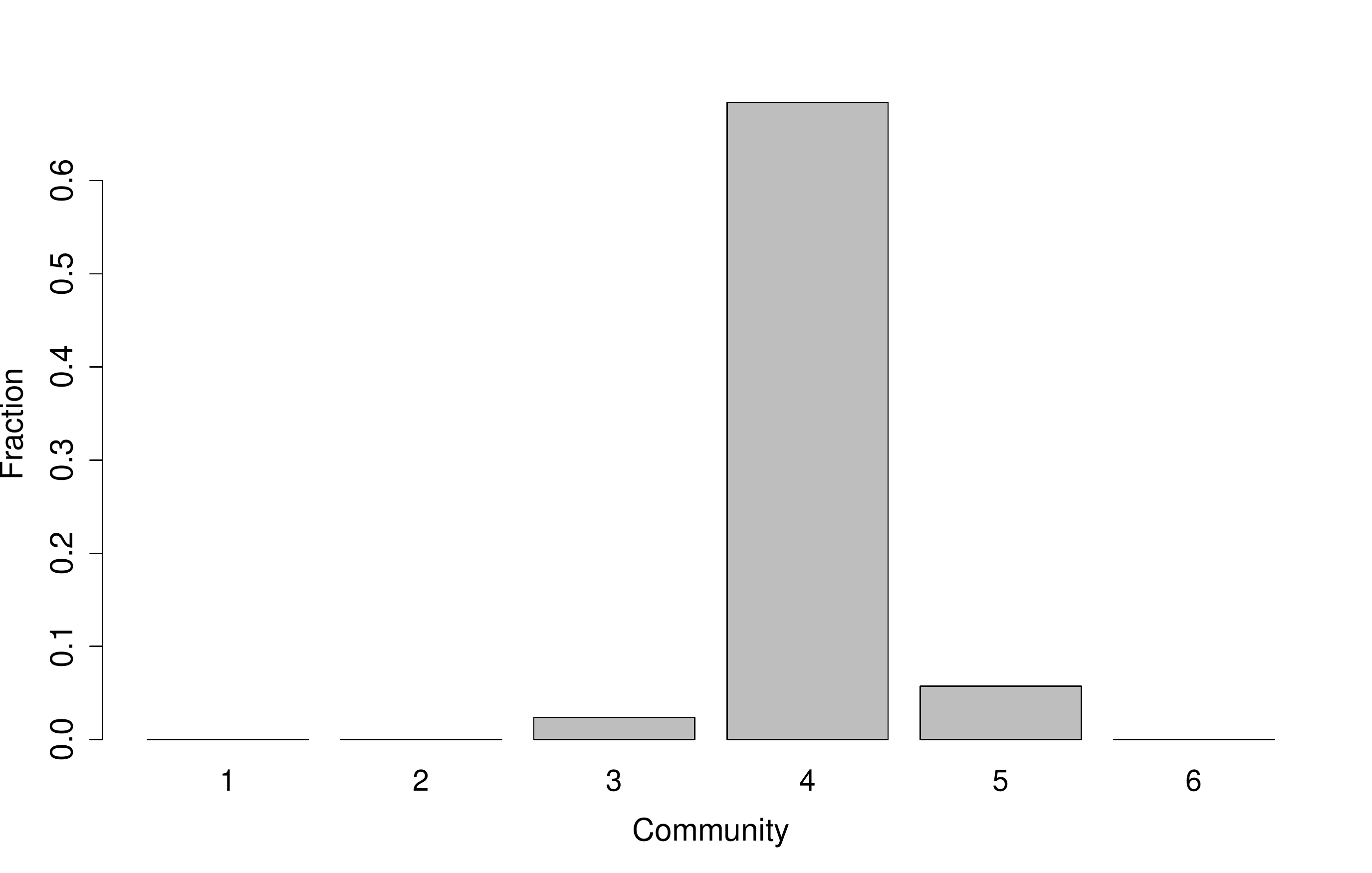}
\caption{Papers with ``fanconi" in the title.}
\label{fig:Fanconi}
\end{subfigure}%
\hfill
\begin{subfigure}[t]{0.5\textwidth}
\centering
\includegraphics[width=\textwidth]{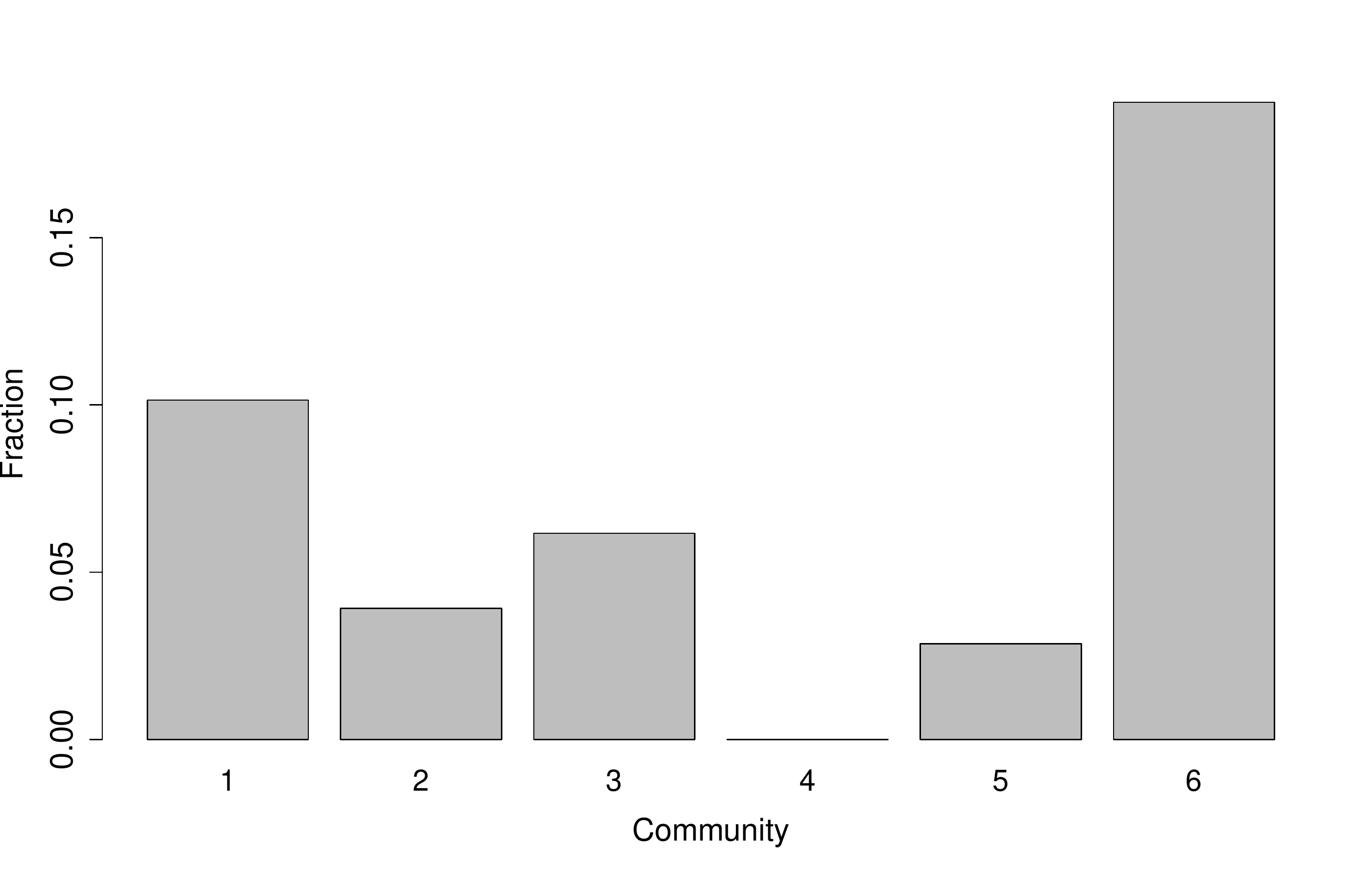}
\caption{Papers with ``sickle" or ``thalassemia" in the title.}
\label{fig:Sickle}
\end{subfigure}
\caption{The proportion of papers with certain key words in the title, out of the signal papers for each community (papers that mention at least 2 genes within that community).}
\label{fig:paper-proportion}
\end{figure}


\section{Discussion}\label{sec:discussion}

We studied recursive partitioning as a framework for hierarchical community detection and proposed two specific algorithms for implementing it, using either spectral clustering or sign splitting.    This framework requires a stopping rule to decide when to stop splitting communities, but otherwise is tuning-free.  We have shown that in certain regimes recursive partitioning has significant advantages in computational efficiency, community detection accuracy,  and hierarchal structure recovery, compared with $K$-way partitioning.   An important feature of hierarchical splitting is that it can recover high-level mega-communities correctly even when all $K$ smaller communities cannot be recovered.  It also provides a natural interpretable representation of the community structure,  and induces a tree-based similarity measure that does not depend on community label permutations and allows us to quantitatively compare entire hierarchies of communities.   The algorithm itself is model-free, but we showed it works under a new model we introduced, the binary tree SBM.   Under this model,  the hierarchical algorithm based on sign splitting is consistent for estimating both individual communities and the entire hierarchy.    We conjecture that the advantage of hierarchical clustering carries over to more general models;  more work will be needed to establish this formally.

\ack{We thank the Associate Editor and the referees for helpful and constructive comments.  }

\bibliography{CommonBib}{}
\bibliographystyle{abbrvnat}

\newpage


\newpage

\setcounter{page}{1}

\begin{center}
\bf{\Large Appendix}
\end{center}

\begin{appendix}
\setcounter{thm}{0}
\renewcommand{\thethm}{\Alph{section}.\arabic{thm}}

\setcounter{lem}{0}
\renewcommand{\thelem}{\Alph{section}.\arabic{lem}}

\setcounter{figure}{0}
\renewcommand{\thefigure}{\thesection.\arabic{figure}} 

\setcounter{equation}{0}
\renewcommand{\theequation}{\thesection.\arabic{equation}}

\section{More on the eigenstructure of the BTSBM}\label{app:eigen}
Recall that $Z\in \bR^{n\times K}$ is the membership matrix with the $i$-th row given by $Z_{i} = e_{\mathcal{I}(c(i))}$, where $e_{j}$ is the $j$-th canonical basis vector in $\bR^{K}$ and $\mathcal{I}$ is the integer given by the binary representation. Under the BTSBM, we have
\[P = \e A = ZBZ^{T} - p_{0}I.\]
Without loss of generality, we can rearrange the nodes so that
\[Z =
  \begin{bmatrix}
    \one_{n_{1}} & 0 & \cdots & 0\\
    0 & \one_{n_{2}} & \cdots & 0\\
    \vdots & \vdots & \ddots & \vdots\\
    0 & 0 & \cdots & \one_{n_{K}}
  \end{bmatrix},
\]
where $\one_{m}$ denotes an $m$-dimensional vector with all entries equal to $1$. Under \eqref{eq:equal_block}, $Z / \sqrt{m}$ has orthonormal columns and we can rewrite $P$ as
\begin{equation}\label{eq:PBZ}
  P = \frac{Z}{\sqrt{m}} (m B)\frac{Z^{T}}{\sqrt{m}} - p_{0}I = \frac{Z}{\sqrt{m}}(m B - p_{0}I)\frac{Z^{T}}{\sqrt{m}}.
\end{equation}
Therefore $\tilde{P}=ZBZ^T$ has the same eigenvalues as $m B$ and the same  eigenvectors as $(Z / \sqrt{m})V$, where $V\in \bR^{K\times K}$ is any basis matrix that gives the eigenspace of $mB$. Therefore, Theorem \ref{thm:eigen2} in Section \ref{sec:theory} is a direct result from the following eigenstructure of $B$ that will be proved at the end of this section.

\begin{thm}\label{thm:eigenB}
  Let $B$ be the $K\times K$ community connection probability matrix of the BTSBM with $K = 2^d$. 
  Then the following holds.  \\
  1. (Eigenvalues) The distinct nonzero eigenvalues of $B$, denoted by $\lambda_{d, 0} \ge  \lambda_{d, 1} \ge  \ldots \ge \lambda_{d, d}$,    are given by 
\begin{equation}\label{eq:lambdaB}
  \lambda_{d, 0} = p_0 + \sum_{r=1}^{d}2^{r-1}p_r, \quad \lambda_{d, q} = p_0 + \sum_{r=1}^{d - q} 2^{r-1}p_r - 2^{d - q}p_{d - q + 1}, ~~q = 1, \ldots, d.
\end{equation}
2. (Eigenvectors) For any $q \in \{1, \ldots, d\}$ and each $x \in S_{q-1}$, let $\nu_{x}^{(d,q)}$ be a $K$-dimensional vector, such that for any $k \in [K]$, 
\[\nu_{x, k}^{(d,q)} = \left\{
    \begin{array}{rl}
      1 & (\mbox{if }\mathcal{I}^{-1}(k) \in S_d\cap \{x0+\}) \, , \\
      -1 & (\mbox{if }\mathcal{I}^{-1}(k) \in S_d\cap \{x1+\}) \, , \\
      0 & (\mbox{otherwise})   \, . 
    \end{array}
\right.,\]
where $\mathcal{I}^{-1}$ maps the integer to its binary representation and $\{x0+\}, \{x1+\}$ are defined above Theorem \ref{thm:eigen2}. Then the eigenspace corresponding to $\lambda_{d,q}$ is spanned by $\{\nu_{x}^{(d,q)}: x \in S_{q-1}\}$.
\end{thm}

Theorem \ref{thm:eigenB} gives one representation of the BTSBM eigenspace. While it is easy to describe, it is not technically convenient, because all the eigenvectors except the first two are sparse, with $\ell_{\infty}$ norm much larger than $\frac{1}{\sqrt{n}}$. However, since $P$ has at most $d + 1 <\!\!< n$ distinct eigenvalues, there are many ways to represent the eigenspace, and in particular, we can find a representation such that all entries of all eigenvectors corresponding to non-zero eigenvalues are $\pm \frac{1}{\sqrt{n}}$.

We start from $d = 1$. In this case,
$B_1 =
  \begin{bmatrix}
    p_0 & p_1 \\
    p_1 & p_0
  \end{bmatrix}
.$ 
Let 
$U_1 =\frac{1}{\sqrt{2}} \begin{bmatrix}
    1 & 1 \\
    1 & -1
  \end{bmatrix}$.
Then
$$B_1 =  U_1
  \begin{bmatrix}
    p_0 + p_1 & 0 \\
    0 & p_0 - p_1
  \end{bmatrix} U_1.$$
Next, consider $d = 2$. By definition,
\[B_2 =
  \begin{bmatrix}
    p_0 & p_1 & p_2 & p_2\\
    p_1 & p_0 & p_2 & p_2\\
    p_2 & p_2 & p_1 & p_0\\
    p_2 & p_2 & p_0 & p_1
  \end{bmatrix}
=
\begin{bmatrix}
  B_1 & p_2\one_{2}\one_{2}^{T}\\
  p_2\one_{2}\one_{2}^{T} & B_1
\end{bmatrix}
.\]
Note that 
\[B_{1} = U_1
  \begin{bmatrix}
    p_0 + p_1 & 0\\
    0 & p_0 - p_1
  \end{bmatrix}
U_1, \quad p_3\one_{2}\one_{2}^{T} = U_1
  \begin{bmatrix}
    2p_3 & 0\\
    0 & 0
  \end{bmatrix} U_1.
\]
We can rewrite $B_2$ as
\[B_2 =
  (I_{2}\otimes U_1)
  \begin{bmatrix}
    p_0 + p_1 & 0 & 2p_2 & 0\\
    0 & p_0 - p_1 & 0 & 0\\
    2p_2 & 0 & p_0 + p_1 & 0\\
    0 & 0 & 0 & p_0 - p_1
  \end{bmatrix}
  (I_{2}\otimes U_1)
\]
where $I_2 =
\begin{bmatrix}
  1 & 0\\
  0 & 1
\end{bmatrix}
$ and $\otimes$ denotes the Kronecker product. 
Using the block diagonal structure and observing that 
\begin{align*}
  \begin{bmatrix}
    p_0 + p_1 & 2p_2\\
    2p_2 & p_0 + p_1\\
  \end{bmatrix} &  = U_1   \begin{bmatrix}
    p_0 + p_1 + 2p_2 & 0\\
    0 & p_0 + p_1 - 2p_2\\
  \end{bmatrix}U_1,
\\  
  \begin{bmatrix}
    p_0 - p_1 & 0 \\
    0 & p_0 - p_1
  \end{bmatrix} & = U_1
    \begin{bmatrix}
    p_0 - p_1 & 0 \\
    0 & p_0 - p_1
  \end{bmatrix} U_1
\end{align*}
we obtain 
\[B_2 = U_2
  \begin{bmatrix}
    p_0 + p_1 + 2p_2 & 0 & 0 & 0\\
    0 & p_0 - p_1 & 0 & 0\\
    0 & 0 & p_0 + p_1 - 2p_2 & 0\\
    0 & 0 & 0 & p_0 - p_1
  \end{bmatrix} U_2,\]
where   $U_2 = (I_{2}\otimes U_1) (U_1 \otimes I_{2}) = U_1\otimes U_1$.  
In other words, $U_2$ is the standard $4\times 4$ Hadamard matrix rescaled by $\frac{1}{2}$.    The following theorem describes a similar structure for the general case.
  \begin{thm}\label{thm:eigen}
For any positive integer $m$, let $\ord_{2}(m) = \max\{q: m \mbox{ is divisible by }2^{q}\}$.  Let $\Lambda_{d}$ be a $2^{d}\times 2^{d}$ diagonal matrix with
\begin{equation}\label{eq:Lambdad}
  \Lambda_{d, ii} = \left\{
    \begin{array}{cc}
      \lambda_{d, 0} & (i = 1)\\
      \lambda_{d, d - \ord_{2}(i - 1)} & (i > 1)
    \end{array}
\right.
\end{equation}
where $\lambda$'s are defined in \eqref{eq:lambdaB}. Then for any $d$,
\begin{equation}\label{eq:eigen1}
B_{d} = U_1^{\otimes d} \Lambda_{d} U_1^{\otimes d}.
\end{equation}
where $U_1^{\otimes d} = \frac{1}{2^{d / 2}}H_{K}$ is symmetric and orthogonal, with  $H_{K}$ the standard $K\times K$ Hadamard matrix.
\end{thm}
Theorem \ref{thm:eigen} gives a complete and precise description of the eigenstructure of $B$.   There are $d + 1$ distinct eigenvalues and there exists an eigenvector basis such that all entries are on the scale of $1 / \sqrt{m} = \sqrt{K / n}$.  
This representation will be important in later proofs, but note that it does not reflect the hierarchical community structure.  

\begin{proof}[\textbf{Proof of Theorem \ref{thm:eigen}}]
  We prove \eqref{eq:eigen1} by induction. We have already shown the result holds for $d = 1$ and $2$.  Suppose it holds for $d - 1$.   Then 
\[B_{d} =
  \begin{bmatrix}
    B_{d-1} & p_{d}\one_{2^{d-1}}\one_{2^{d-1}}^{T}\\
    p_{d}\one_{2^{d-1}}\one_{2^{d-1}}^{T} & B_{d - 1}
  \end{bmatrix}.\]
By the induction hypothesis,
\[B_{d-1} = U_{1}^{\otimes (d-1)}\Lambda_{d - 1}U_{1}^{\otimes (d-1)}, \,\, p_{d}\one_{2^{d-1}}\one_{2^{d-1}}^{T} = U_{1}^{\otimes (d-1)}
  \begin{bmatrix}
    2^{d-1}p_{d} & 0\\
    0 & 0
  \end{bmatrix}U_{1}^{\otimes (d-1)}.
\]
Write
\[\Lambda_{d - 1} =
  \begin{bmatrix}
    \lambda_{d - 1, 0} & 0\\
    0 & \td{\Lambda}_{d-1}
  \end{bmatrix}
\]
where $\lambda_{d-1, 0}\in \bR$ and $\td{\Lambda}_{d-1}$ is a $(2^{d - 1} - 1)$-dimensional diagonal matrix. Then 
\begin{align*}
B_{d} &= \lb I_2 \otimes U_{1}^{\otimes (d-1)}\rb
  \begin{bmatrix}
    \lambda_{d-1, 0} & 0 & 2^{d-1}p_d & 0\\
    0 & \td{\Lambda}_{d-1} & 0 & 0\\
    2^{d-1}p_d & 0 & \lambda_{d-1, 0} & 0\\
    0 & 0 & 0 & \td{\Lambda}_{d-1}
  \end{bmatrix}\lb I_2 \otimes U_{1}^{\otimes (d-1)}\rb
\end{align*}
Note that
\begin{align*}
  \begin{bmatrix}
    \lambda_{d-1, 0} & 2^{d-1}p_d\\
    2^{d-1}p_d & \lambda_{d-1, 0}
  \end{bmatrix} &= U_1
  \begin{bmatrix}
    \lambda_{d-1, 0} + 2^{d - 1}p_d & 0\\
    0 & \lambda_{d-1, 0} - 2^{d - 1}p_d
  \end{bmatrix}U_1
\end{align*}
and for any $j$,
\[
  \begin{bmatrix}
    \td{\Lambda}_{jj} & 0\\
    0 & \td{\Lambda}_{jj}
  \end{bmatrix}
 = U_1   \begin{bmatrix}
    \td{\Lambda}_{jj} & 0\\
    0 & \td{\Lambda}_{jj}
  \end{bmatrix} U_1.\]
Therefore,
\begin{align*}
B_{d} &= U_{d}
  \begin{bmatrix}
    \lambda_{d-1, 0} + 2^{d-1}p_d & 0 & 0 & 0\\
    0 & \td{\Lambda}_{d-1} & 0 & 0\\
    0 & 0 & \lambda_{d-1, 0} - 2^{d-1}p_d & 0\\
    0 & 0 & 0 & \td{\Lambda}_{d-1}
  \end{bmatrix}U_{d}
\end{align*}
where 
\[U_{d} = \lb I_2 \otimes U_{1}^{\otimes (d-1)}\rb (U_1 \otimes I_{2^{d-1}}) = U_{1}^{\otimes d}.\]
Finally, by definition, 
\[\lambda_{d-1, 0} + 2^{d-1}p_d = \lambda_{d, 0}, \quad \lambda_{d-1, 0} - 2^{d-1}p_d = \lambda_{d, 1}.\]
Then it is easy to verify that 
\[\Lambda_{d} =   \begin{bmatrix}
    \lambda_{d-1, 0} + 2^{d-1}p_d & 0 & 0 & 0\\
    0 & \td{\Lambda}_{d-1} & 0 & 0\\
    0 & 0 & \lambda_{d-1, 0} - 2^{d-1}p_d & 0\\
    0 & 0 & 0 & \td{\Lambda}_{d-1}
  \end{bmatrix}.\]
This completes the proof.  
\end{proof}

Using Theorem \ref{thm:eigen} we can now prove Theorem \ref{thm:eigen2}.

\begin{proof}[\textbf{Proof of Theorem \ref{thm:eigen2}}]
Part (i) has been proved in Theorem \ref{thm:eigen} and now we need to prove part (ii). It is easy to see that $\{\nu_{1}^{(q)}, \ldots, \nu_{2^{q-1}}^{(q)}\}$ are mutually orthogonal. By Theorem \ref{thm:eigen}, $\lambda_{d, q}$ has multiplicity $2^{q-1}$. Thus it remains to prove that $\nu_{j}^{(q)}$ is an eigenvector. For any $i\in \{1, \ldots, 2^{d}\}$, let $B_{i}^{T}$ be the $i$-th row of $B$. Then
\begin{equation}\label{eq:eigen2}
B_{i}^{T} \nu_{j}^{(q)} = \sum_{h = 1}^{2^{d}}B_{ih}\nu_{j, h}^{(q)} = \sum_{h\in S_{2j-2}^{(q)}}B_{ih} - \sum_{h\in S_{2j-1}^{(q)}}B_{ih}.
\end{equation}
Noting that
\[S_{2j-1}^{(q)} = \{i + 2^{d - q}: i \in S_{2j-2}^{(q)}\},\]
we can rewrite \eqref{eq:eigen2} as
\[B_{i}^{T} \nu_{j}^{(q)} = \sum_{h\in S_{2j-2}^{(q)}}\lb B_{ih} - B_{i(h + 2^{d-q})}\rb.\]
Let $i_1 i_2 \ldots i_d$ (resp. $h_1 h_2 \ldots h_d$) be the binary representation of $i$ (resp. $h$), with zeros added to the front of the list  whenever necessary. Then 
\[h\in S_{2j-2}^{(q)}\Longleftrightarrow h_1 h_2\ldots h_{q} \mbox{ is the binary representation of }2j - 2.\]
As a consequence, $h_{q} = 0$ and hence
\[h_1 h_2 \ldots (h_{q} + 1) h_{q+1}\ldots h_{d}\mbox{ is the binary representation of }h + 2^{d - q}.\]
By definition,
\[B_{ih} = p_{w(i, h)}, \quad \mbox{where} \quad w(i, h) = \left\{
    \begin{array}{ll}
     d + 1 - \min\{r: i_{r} \not = h_{r}\} & (i \not = h)\\
      0 & i = h
    \end{array}
\right..\]
Thus, if $(i_1, \ldots, i_{q-1})\not = (h_1, \ldots, h_{q-1})$, then
\[w(i, h) = w(i, h + 2^{d-q}) = d + 1 - \min\{r\le q - 1: i_{r}\not = h_{r}\} \Longrightarrow B_{ih} - B_{i(h + 2^{d-q})} = 0.\]
Thus, 
\begin{equation}
  \label{eq:case1}
  B_{i}^{T} \nu_{j}^{(q)} = 0.
\end{equation}
If $(i_1, \ldots, i_{q-1}) = (h_1, \ldots, h_{q-1})$, we consider two cases. If $i_{q} = 0$, then
\[w(i, h + 2^{d - q}) = d + 1 - q, \quad w(i, h) = \left\{\begin{array}{ll}
d + 1 - \min\{r > q: i_{r} \not = h_{r}\} & \mbox{ if } i \not = h\\
0 & \mbox{ if } i = h
\end{array}\right.\]
Then 
\begin{align}
B_{i}^{T} \nu_{j}^{(q)} &= (p_0 - p_{d+1 - q}) + \sum_{h\in S_{2j-2}^{(q)}\setminus \{i\}}\lb p_{d + 1 - \min\{r > q: i_{r} \not = h_{r}\}} - p_{d + 1 - q}\rb  \nonumber\\
& = (p_0 - p_{d+1 - q}) + \sum_{j = 1}^{d - q}2^{j-1}\lb p_{j} - p_{d + 1 - q}\rb = \lambda_{q}\label{eq:case2}.
\end{align}
Similarly, if $i_{q} = 1$, then
\begin{equation}
  \label{eq:case3}
  B_{i}^{T} \nu_{j}^{(q)} = -\lambda_{q}.
\end{equation}
Putting \eqref{eq:case1} - \eqref{eq:case3} together, we conclude that $B\nu_{j}^{(q)} = \lambda_{q}\nu_{j}^{(q)}$.  
\end{proof}
Finally, Corollary \ref{cor:eigen} is a direct consequence of Theorems \ref{thm:eigen2} and \ref{thm:eigen}.
\begin{proof}[\textbf{Proof of Corollary \ref{cor:eigen}}]
 By Theorem \ref{thm:eigen}, under \eqref{eq:assortative}, we have
 \begin{align*}
   \lambda_{d, 0} - \lambda_{d, 1} & = 2^{d}p_{d} > 0,  \\
   \lambda_{d, q} - \lambda_{d, q + 1} & = 2^{d - q}(p_{d - q} - p_{d -
   q + 1}) > 0
   \end{align*}
for $q = 1, \ldots, d - 1$.   Thus $\lambda_{d, q}$ is decreasing in
$q$. Since $P$ has the same eigenstructure as $mB - p_{0}I$, the
second largest eigenvalue of $P$ is
\[\lambda_{2} = m\lambda_{d, 1} - p_{0} = (m - 1)p_{0} + m\sum_{i=1}^{d-1}2^{i-1}p_{i} - m2^{d - 1}p_{d},\] 
and the eigengap is
\begin{align*}
\Delta_{2} = &m\min\{\lambda_{d, 0} - \lambda_{d, 1}, \lambda_{d, 1} - \lambda_{d, 2}\} = m\min\{2^{d}p_{d}, 2^{d-1}(p_{d-1} - p_{d})\} \\
= & n\min\{p_{d}, (p_{d-1} - p_{d}) / 2\} = n\rho_{n}\min\{a_{d}, (a_{d-1} - a_{d}) / 2\}.
\end{align*}
If on the other hand we have dis-assortative sequence \eqref{eq:disassortative}, the above inequalities are reverted thus $\lambda_{d, q}$ is increasing in
$q$ but we have
$$\lambda_{d, d} < 0.$$
Thus $\lambda_{d,1} < \lambda_{d,2} \cdots, < \lambda_{d,d} <0.$ However, notice that $|\lambda_{d,1}|$ is still larger than $|\lambda_{d,q}|, q>1$ so it is still the one with second largest magnitude.

To prove the claim about $u_2$, note that 
\[H_{d} = \begin{bmatrix}
    1 & 1 \\
    1 & -1
  \end{bmatrix} \otimes H_{d - 1} = 
  \begin{bmatrix}
    H_{d-1} & H_{d-1}\\ 
    H_{d-1} & -H_{d-1}
  \end{bmatrix}.
\]
It is easy to prove by induction in $d$ that the first column of $H_{d}$ is
$\one_{2^{d}}$. The $(2^{d-1} + 1)$-th column is
\[
  \begin{bmatrix}
    \one_{2^{d-1}}\\ -\one_{2^{d-1}}
  \end{bmatrix}.
\]
By definition, 
\[u_{2} = \frac{1}{\sqrt{m}}Z \lb\frac{1}{\sqrt{K}}H_{d}[, 2^{d-1} + 1]\rb = \frac{1}{\sqrt{n}}  \begin{bmatrix}
    \one_{n / 2} \\ -\one_{n / 2}
  \end{bmatrix}.\]

\end{proof}

\section{Proof of exact recovery by the HCD-Sign algorithm}
Throughout this section, we denote the eigenvalues of $P$ and $A$  by $\lambda_{1}\ge \ldots \ge \lambda_{n}$  and $\td{\lambda}_{1}\ge \ldots \ge \td{\lambda}_{n}$, respectively.   Let $u_{1}, \ldots, u_{n}$ (resp. $\td{u}_{1}, \ldots, \td{u}_{n}$) be an orthonormal set of eigenvectors (not necessarily unique) satisfying $Pu_{s} = \lambda_{s}u_{s}$, and similarly, $A\td{u}_{s} = \td{\lambda}_{s}\td{u}_{s}$.   Let
\[H = A - P.\]
For completeness, we state some results on concentration of adjacency matrices we will use.   

\begin{lem}[Theorem 17 of \cite{eldridge2017unperturbed}]\label{lem:FixTerm}
Let $X$ be an $n\times n$ symmetric random matrix such that $\e X = (0)$ and all of its entries on and above the diagonals are independent,  $\e|X_{ij}|^p \le \frac{1}{n}$ for all $i, j \in [n]$ and $p\ge 2$. Let $u$ be an $n$-vector with $\norm{u}_{\infty}=1$. Given constants $\xi >1$ and $0 < \kappa < 1$, with probability at least $1-n^{-\frac{1}{4}(\log_{\mu}{n})^{\xi-1}(\log_{\mu}{e})^{-\xi}}$, 
$$\norm{X^pu}_{\infty} < (\log{n})^{p\xi} \text{~~for all~~} p\le \frac{\kappa}{8}(\log^{\xi}{n})$$
where $\mu = \frac{2}{\kappa+1}$.
\end{lem}

We also need the matrix perturbation result based on the Neumann trick, again from \cite{eldridge2017unperturbed}. However, the original paper contains a mistake in the proof of Theorem 12 where the authors claimed in their equation (10) that 
\[\lb\frac{|\lambda_t|}{|\lambda_t| - \|H\|}\rb^{k}\le \frac{|\lambda_t|}{|\lambda_t| - \|H\|}, \quad \forall k \ge 1.\]
This is clearly wrong as the RHS is bigger than $1$. Fortunately, this error can be fixed, though the conclusion has to be changed. In addition, there is a missing assumption in the original statement of \cite{eldridge2017unperturbed}.     Here we state the corrected version with a slightly different conclusion.  A self-contained proof can be found in Appendix~\ref{app:eldridge}.

\begin{lem}[Modified Theorem 9 of \cite{eldridge2017unperturbed}]\label{lem:deterministicControl}
Given any $t \in [n]$, suppose $\lambda_{t}$ has multiplicity $1$. Define the eigengap as 
\[\Delta_t = \min\{|\lambda_{t} - \lambda_{s}|: s\not = t\}.\] 
If $\norm{H} < |\lambda_t| / 2$, then for any set of eigenvectors $u_1, \ldots, u_n$, it holds for all $j \in [n]$ that
\begin{align*}
  &\min_{\zeta\in \{-1, +1\}}|(\td{u}_{t} - \zeta u_{t})_{j}|\\
\le& \lb\frac{4\|H\|^{2}}{\Delta_{t}^{2}} + \frac{2\|H\|}{|\lambda_{t}|}\rb|u_{t, j}| + 2\xi_{j}(u_{t}; H, \lambda_{t}) + \frac{4\sqrt{2}\|H\|}{\Delta_{t}}\sum_{s\not = t}\bigg|\frac{\lambda_{s}}{\lambda_{t}}\bigg|\cdot \lb |u_{s, j}| + \xi_{j}(u_{s}; H, \lambda_{t})\rb.
\end{align*}
where $\xi(u;H,\lambda)$ is a $n$-vector whose $j$th entry is defined to be 
$$\xi_j(u;H,\lambda) = \sum_{p\ge 1}\lb\frac{2}{|\lambda|}\rb^{p}|(H^{p}u)_{j}|.$$
\end{lem}
Finally, we will need a concentration result for the adjacency matrix, stated here in the form proved by  \cite{lei2014consistency}.
\begin{lem}[\cite{lei2014consistency}]\label{lem:BasicConcentration}
If $n\max_{ij}P_{ij} \ge c_{0}\log n$ for some universal constant $c_{0} > 0$, then given any $\phi > 0$, there exists a constant $\tilde{C}(\phi, c_{0})$, which only depends on $\phi$ and $c_0$, such that 
\[\|H\|\le \tilde{C}(\phi, c_{0})\sqrt{n\max_{ij}P_{ij}},\]
with probability at least $1-n^{-\phi}$.
\end{lem}

We are now ready to prove the following lemma about the concentration of adjacency matrix eigenvectors in $\ell_{\infty}$ norm. 
\begin{lem}\label{lem:KConcentration}
Let  $A\in \bR^{n\times n}$ be generated from the BTSBM with parameters $(n, \rho_{n}; a_{1}, \ldots, a_{d})$ as defined in \eqref{eq:rhon}, with $n=Km = 2^dm$. Let $\eta_{d} = \min\{a_{d}, |a_{d-1} - a_{d}| / 2\}$.
Assume that
\begin{equation}
  \label{eq:KConcentration_cond}
  \frac{\log^{\xi}n\max\{1, a_d\}}{\sqrt{n\rho_n}\eta_{d}}\le \frac{1}{4},
\end{equation}
for some $\xi > 1$. If one of the following conditions is true:
\begin{enumerate}
\item  the model is assortative as \eqref{eq:assortative}, 
\item the model is dis-assortative as \eqref{eq:disassortative},
\end{enumerate} 
 then for any $\phi > 0$, there exists a constant $C_1(\phi)$, which only depend on $q$, such that
\[\min_{\zeta\in \{-1, +1\}}\norm{\tilde{u}_2 - \zeta u_2}_{\infty} \le C_1(\phi)\frac{1}{\sqrt{n}}\frac{\max\{1,a_d\}}{\sqrt{n\rho_n}\eta_{d}} \max\left\{\log^{\xi} n, \frac{\max\{1, a_{d}\}}{\eta_{d}}\right\}, \]
with probability at least $1-2n^{-\phi}$ for sufficiently large $n$.
\end{lem}

\begin{proof}[Proof of Lemma~\ref{lem:KConcentration}]
Let $H = A-P$ be the symmetric Bernoulli noise matrix. Recall that $\tilde{P} = P + p_{0}I$ and define $\tilde{A} = A+ p_{0}I$. As discussed, $\tilde{A}$ and $\tilde{P}$ have the same eigenvectors as $A$ and $P$ respectively and $H = A - P = \tilde{A} - \tilde{P}$. Since we only need to study the perturbation bound on the eigenvectors of $A$, we can apply the lemmas to $\tilde{A}$ and $\tilde{P}$ instead of $A$ and $P$. Therefore, we would ignore the trivial difference caused by enforcing zeros on diagonals of $A$ 
for the rest of the proof.

We first show the result in the assortative setting. By Corollary \ref{cor:eigen}, 
\begin{align}
  \lambda_2 &= mp_{0} + m\sum_{i=1}^{d-1}2^{i-1}p_{i} - m2^{d - 1}p_{d} = n\rho_{n} \frac{1 + \sum_{i=1}^{d-1}2^{i-1}a_{i} - 2^{d-1}a_{d}}{K}\notag\\
& = n\rho_{n} \frac{(1 - a_{d}) + \sum_{i=1}^{d-1}2^{i-1}(a_{i} - a_{d})}{K}\notag\\
& \ge n\rho_{n} (a_{d - 1} - a_{d})\frac{1 + \sum_{i=1}^{d-1}2^{i-1}}{K}\notag\\
& = n\rho_{n} \frac{a_{d - 1} - a_{d}}{2} \ge n\rho_{n}\eta_{d}\label{eq:lambda2}.
\end{align}

By Theorem \ref{thm:eigen}, 
\[B = \lb\frac{H_{K}}{2^{d / 2}}\rb \Lambda_{d} \lb\frac{H_{K}}{2^{d / 2}}\rb.\]
By \eqref{eq:PBZ} and recalling that $P$ has been replaced by $P + p_{0}I$,
\[P = \bar{U} \Lambda_{d} \bar{U}^{T}, \quad \mbox{where }\bar{U} = \frac{ZH_{K}}{\sqrt{m2^{d}}} = \frac{ZH_{K}}{\sqrt{n}}.\]
Set $u_{1}, \ldots, u_{K}$ as the columns of $\bar{U}$ and $u_{K + 1}, \ldots, u_{n}$ to be arbitrary eigenvectors of the eigenvalue $0$. Without loss of generality, we assume that $\langle \td{u}_{2}, u_{2}\rangle\ge 0$. Then
\begin{equation}
  \label{eq:usj}
  |u_{s, j}| \equiv \frac{1}{\sqrt{n}}, \quad \forall s \in [K], j\in [n].
\end{equation}

When $n\ge 3$, 
\[n\rho_{n}\ge \frac{\log^{2\xi} n}{\eta_{d}^{2}}\ge \log^{2\xi}n \ge \log n.\]
Then Lemma \ref{lem:BasicConcentration} implies that there exists $C_0(\phi) > 0$ that only depends on $q$ such that
\begin{equation}
  \label{eq:pE0}
  \p(E_{0}) \ge 1 - n^{-\phi},
\end{equation}
where
\begin{equation}
  \label{eq:E0}
  E_{0} = \left\{ \norm{H} \le C_0(\phi) \sqrt{n\rho_n} \right\}.
\end{equation}
Since $\sqrt{n\rho_{n}}\eta_{d} \ge 4\log^{\xi} n$, for sufficiently large $n$ (e.g., $n \ge \exp(C_0(\phi) / \xi)$), 
\begin{equation}\label{eq:c0q}
  2C_0(\phi) \sqrt{n\rho_n}\le n\rho_{n}\eta_{d}.
\end{equation}

Note that \eqref{eq:c0q} and \eqref{eq:lambda2} imply that $\|H\|\le |\lambda_{2}| / 2$. Thus Lemma \ref{lem:deterministicControl} holds for $t = 2$. By Corollary \ref{cor:eigen}, 
$$\delta_{2} = \min\{|\lambda_{2} - \lambda_{s}|: s\not = 2\} = n\rho_{n} \min\{a_{d}, (a_{d-1} - a_{d}) / 2\}= n\rho_{n} \eta_{d}.$$
Thus, by Lemma \ref{lem:deterministicControl},
\begin{align}
  |(\tilde{u}_2-u_2)_j| \le&  \frac{1}{\sqrt{n}}\left( \frac{4\norm{H}^{2}}{(n\rho_{n}\eta_{d})^{2}}+\frac{2\norm{H}}{n\rho_{n}\eta_{d}}\right) + 2\xi_j(u_2; H, \lambda_2) \notag\\
& + \frac{4\sqrt{2}\|H\|}{n\rho_{n}\eta_{d}}\sum_{s\not = 2, s\le K}\bigg|\frac{\lambda_{s}}{\lambda_{2}}\bigg|\cdot \lb \frac{1}{\sqrt{n}} + \xi_{j}(u_{s}; H, \lambda_{2})\rb.
  \label{eq:deterministic1}
\end{align}
The summands with $s > K$ in the last term can be dropped because $\lambda_{s} = 0$. 


Next we derive a bound for $\xi_j(u_s; H, \lambda_2)=\sum_{p\ge 1}\lb\frac{2}{|\lambda_2|}\rb^{p}|(H^{p}u_s)_{j}|$. First let
\[\gamma = \sqrt{n\rho_{n}}, \quad X = H / \gamma.\]
For sufficiently large $n$ (e.g., $n\ge 3$), we have $\gamma > 1$ and hence $|H_{ij} / \gamma| \le 1$, since $|H_{ij}|\le 1$. Then
\begin{equation}\label{eq:EHijp}
  \e|H_{ij}/\gamma|^p \le \e |H_{ij}/\gamma|^2 = \frac{P_{ij}(1-P_{ij})}{\gamma^2} \le \frac{\max_{ij}P_{ij}}{\gamma^2} = \frac{\rho_n}{\gamma^2} = \frac{1}{n}. 
\end{equation}
Let
\begin{equation}\label{eq:E1s}
  E_{1s} = \left\{\norm{X^p u_s }_{\infty} < (\log{n})^{p\xi}\norm{u_s}_{\infty}, \text{~~for all~~} p\le \frac{1}{16}(\log^{\xi}{n})\right\}.
\end{equation}
Then by Lemma \ref{lem:FixTerm} with $\kappa = 1/2$, 
\begin{equation}
  \label{eq:pE1s}
  \p(E_{1s})\ge 1-n^{-\frac{1}{4}(\log_{\mu}{n})^{\xi-1}(\log_{\mu}{e})^{-\xi}}, \quad \mbox{where }\mu = 4 / 3.
\end{equation}
This implies 
\begin{equation}
  \label{eq:E1}
  \p(E_{1})\ge 1-n^{-\frac{1}{4}(\log_{\mu}{n})^{\xi-1}(\log_{\mu}{e})^{-\xi} + 1}, \quad \mbox{where }E_{1} = \bigcap_{s=1}^{K}E_{1s}.
\end{equation}
Since $\xi > 1$, for sufficiently large $n$, the above probability is larger than $1 - n^{-\phi}$. Now let 
\begin{equation}
  \label{eq:E}
  E = E_{0}\cap E_{1}, \quad \mbox{then }\p(E)\ge 1 - 2n^{-\phi}.
\end{equation}

Let $M = \lfloor \log^{\xi} n / 16\rfloor$.  If event $E$ holds, by \eqref{eq:usj} we have
\[\|X\| \le C_0(\phi), \quad \|X^{p}u_s\|_{\infty}\le \frac{(\log n)^{p\xi}}{\sqrt{n}}, \,\, \forall p \in [M], s\in [K].\]
In addition, using  $\|v\|_{\infty}\le \|v\|_{2}$ for any vector $v$, for all $p\ge 1$,
\[\|X^{p}u_s\|_{\infty}\le \|X^{p}u_s\|_{2}\le \|X^{p}\| \le C_0(\phi)^{p}.\]
Using these properties, we have
\begin{align}
\xi_j  (u_s; H, \lambda_2) &  = \sum_{p\ge 1}\lb\frac{2\gamma}{\lambda_2}\rb^{p}|(X^p u_s)_j|
  = \sum_{p = 1}^{M}\lb\frac{2\gamma}{\lambda_2}\rb^{p}|(X^p u_s)_j| + \sum_{p > M}\lb\frac{2\gamma}{\lambda_2}\rb^{p}|(X^p u_s)_j|\notag\\
                           & \le \frac{1}{\sqrt{n}}\sum_{p = 1}^{M}\lb\frac{2\gamma (\log n)^{\xi}}{\lambda_2}\rb^{p} + \sum_{p > M}\lb\frac{2\gamma C_0(\phi)}{\lambda_2}\rb^{p} \notag \\
&                             \le \frac{1}{\sqrt{n}}\sum_{p = 1}^{M}\lb\frac{2\log^{\xi} n}{\sqrt{n\rho_{n}}\eta_{d}}\rb^{p} + \sum_{p > M}\lb\frac{2C_0(\phi)}{\sqrt{n\rho_{n}}\eta_{d}}\rb^{p}\label{eq:Z-decomposition},
\end{align}
where the last line uses \eqref{eq:lambda2}. Since $\frac{\log^{\xi} n}{\sqrt{n\rho_{n}}\eta_{d}} \le \frac{1}{4}$, for sufficiently large $n$ (e.g., $n\ge e^{C_0(\phi)}$), 
\begin{equation}\label{eq:logxin}
  \frac{2\max\{\log^{\xi} n, C_0(\phi)\}}{\sqrt{n\rho_{n}}\eta_{d}} \le \frac{1}{2}.
\end{equation}
Then \eqref{eq:Z-decomposition} implies that
\begin{align}
  \xi_j(u_s; H, \lambda_2) &\le \frac{1}{\sqrt{n}}\frac{2\log^{\xi} n}{\sqrt{n\rho_{n}}\eta_{d}}\sum_{p = 1}^{M}2^{-(p - 1)} + \frac{2C_0(\phi)}{\sqrt{n\rho_{n}}\eta_{d}}\sum_{p > M}2^{-(p - 1)}\notag\\
\le &\frac{1}{\sqrt{n}}\frac{4\log^{\xi} n}{\sqrt{n\rho_{n}}\eta_{d}} + \frac{4C_0(\phi)}{\sqrt{n\rho_{n}}\eta_{d}} 2^{-M}.  \label{eq:zetaj1}
\end{align}
Since $M = \log^{\xi}n / 16$ and $\xi > 1$, $2^{-M}$ decays faster than any polynomial of $n$. Thus for sufficiently large $n$,
\begin{equation}\label{eq:2-M}
  2^{-M}\le \frac{1}{\sqrt{n}}, \quad C_0(\phi)\le \log^{\xi}n.
\end{equation}
This simplifies \eqref{eq:zetaj1} to 
\begin{equation}
  \label{eq:zetaj}
  \xi_j(u_s; H, \lambda_2) \le \frac{1}{\sqrt{n}}\frac{8\log^{\xi} n}{\sqrt{n\rho_{n}}\eta_{d}}.
\end{equation}

By \eqref{eq:deterministic1} and \eqref{eq:zetaj}, assuming event $E$, 
\begin{align}
  |(\tilde{u}_2-u_2)_j| \le&  \frac{1}{\sqrt{n}}\left( \frac{4C_0(\phi)^{2}}{n\rho_{n}\eta_{d}^{2}}+\frac{2C_0(\phi)}{\sqrt{n\rho_{n}}\eta_{d}}\right) + \frac{1}{\sqrt{n}}\frac{16\log^{\xi} n}{\sqrt{n\rho_{n}}\eta_{d}} \notag\\
& + \frac{1}{\sqrt{n}}\frac{4\sqrt{2}C_0(\phi)}{\sqrt{n\rho_{n}}\eta_{d}}\lb 1 + \frac{8\log^{\xi} n}{\sqrt{n\rho_{n}}\eta_{d}}\rb \frac{\sum_{s\not = 2, s\le K}|\lambda_{s}|}{\lambda_{2}}\notag\\
\le & \frac{1}{\sqrt{n}}\frac{22\log^{\xi} n}{\sqrt{n\rho_{n}}\eta_{d}} + \frac{1}{\sqrt{n}}\frac{12\sqrt{2}C_0(\phi)}{\sqrt{n\rho_{n}}\eta_{d}}\frac{\sum_{s\not = 2, s\le K}|\lambda_{s}|}{\lambda_{2}}\notag\\
\le & \frac{1}{\sqrt{n}}\frac{1}{\sqrt{n\rho_{n}}\eta_{d}} \lb 22\log^{\xi}n + 12\sqrt{2}C_0(\phi)\frac{\sum_{s=1}^{K}|\lambda_{s}|}{\lambda_{2}}\rb.\label{eq:deterministic3}
\end{align}
where the second inequality uses \eqref{eq:c0q}, \eqref{eq:logxin} and \eqref{eq:2-M}. By Theorem \ref{thm:eigen}, 
\[\lambda_{1} = m\lb p_0 + \sum_{r=1}^{d}2^{r-1}p_r\rb\]
and for any $j\in [d]$,
\[\lambda_{2^{j-1}+1} = \ldots = \lambda_{2^{j}} = m\lb p_0 + \sum_{r=1}^{d - j} 2^{r-1}p_r - 2^{d - j}p_{d - j + 1}\rb.\]
Since $p_{0} > p_{1} > \ldots > p_{d}$, 
\[\lambda_{2^{j}} = m(p_{0} - p_{d - j + 1}) + m\sum_{r = 1}^{d - j}2^{r - 1}(p_r - p_{d - j + 1})\ge 0.\]
As a result,
\begin{align*}
  \frac{1}{m}\sum_{s=1}^{K}|\lambda_{s}| & = \frac{1}{m}\sum_{s = 1}^{K}\lambda_{s} = \frac{1}{m}\lambda_{1} + \frac{1}{m}\sum_{j=1}^{d}2^{j-1}\lambda_{2^{j}}\\
& = p_0 + \sum_{r=1}^{d}2^{r-1}p_r + \sum_{j=1}^{d}2^{j-1}\lb p_0 + \sum_{r=1}^{d - j} 2^{r-1}p_r - 2^{d - j}p_{d - j + 1}\rb\\
& = p_0 \lb 1 + \sum_{j=1}^{d}2^{j - 1}\rb + \sum_{r=1}^{d}2^{r - 1}p_{r} \lb 1 + \sum_{j=1}^{d - r}2^{j - 1} - 2^{d - r}\rb\\
& = 2^{d}p_{0} = 2^{d}\rho_{n}.
\end{align*}
Further, by \eqref{eq:lambda2} we have
\[\frac{\sum_{s=1}^{K}|\lambda_{s}|}{\lambda_{2}} = \frac{m2^{d}\rho_{n}}{\lambda_{2}} \le \frac{n\rho_{n}}{n\rho_{n}\eta_{d}}\le \eta_{d}^{-1}.\]
Plugging this into \eqref{eq:deterministic3}, we conclude that, assuming event $E$, 
\[|(\tilde{u}_2-u_2)_j| \le\frac{1}{\sqrt{n}}\frac{1}{\sqrt{n\rho_{n}}\eta_{d}} \lb 22\log^{\xi}n + 12\sqrt{2}C_0(\phi)\eta_{d}^{-1}\rb.\]
The proof of assortative case is then completed by setting $C_1(\phi) = 22 + 12\sqrt{2}C_0(\phi)$.

 For the dis-assortative case, the proof is almost the same and we only point out the steps with slight different arguments. First, \eqref{eq:lambda2} would become
\begin{equation}\label{eq:lambda2-dis}
| \lambda_2 | \ge n\rho_{n}\eta_{d}.
\end{equation}

 Lemma \ref{lem:BasicConcentration} implies that there exists $C_0(\phi) > 0$ that only depends on $q$ such that
\begin{equation}
  \label{eq:pE0}
  \p(E_{0}) \ge 1 - n^{-\phi},
\end{equation}
where
\begin{equation}
  \label{eq:E0}
  E_{0} = \left\{ \norm{H} \le C_0(\phi) \sqrt{n\rho_na_d} \right\}.
\end{equation}
From $\frac{\log^{\xi}na_d}{\sqrt{n\rho_n}\eta_{d}}\le \frac{1}{4}$, for sufficiently large $n$, we have
$$a_d \le \frac{1}{4}\frac{\sqrt{n\rho_n}\eta_{d}}{\log^{\xi}{n}}$$
and
$$\sqrt{n\rho_n}\eta_d \log^{\xi}{n} \ge 4a_d\log^{2\xi}{n}>1/4 $$
Thus for sufficiently large $n$
\begin{align}\label{eq:c0q}
  2C_0(\phi) \sqrt{n\rho_na_d} &\le 2C_0(\phi) \sqrt{n\rho_n\frac{1}{4}\frac{\sqrt{n\rho_n}\eta_{d}}{\log^{\xi}{n}}}\notag\\
  &  = n\rho_n\sqrt{\frac{1}{4}\frac{\eta_{d}}{\sqrt{n\rho_n}\log^{\xi}{n}}}\notag=n\rho_n\sqrt{\frac{1}{4}\frac{\eta_{d}^2}{\sqrt{n\rho_n}\eta_d\log^{\xi}{n}}}\notag\\
  & < n\rho_n\sqrt{\frac{1}{4}4\eta_d^2}= n\rho_{n}\eta_{d}.
\end{align}
Thus the condition $\norm{H} < |\lambda_2|/2$ is still true under the event $E_0$. Next, for the bound of $\xi_j(u_s; H, \lambda_2)=\sum_{p\ge 1}\lb\frac{2}{|\lambda_2|}\rb^{p}|(H^{p}u_s)_{j}|$. We instead take
\[\gamma = \sqrt{n\rho_{n}a_d}, \quad X = H / \gamma.\]
For sufficiently large $n$ 
\begin{equation}\label{eq:EHijp-dis-ass}
  \e|H_{ij}/\gamma|^p \le \e |H_{ij}/\gamma|^2 = \frac{P_{ij}(1-P_{ij})}{\gamma^2} \le \frac{\max_{ij}P_{ij}}{\gamma^2} = \frac{\rho_na_d}{\gamma^2} = \frac{1}{n}. 
\end{equation}

As with \eqref{eq:E1s}, \eqref{eq:pE1s}, \eqref{eq:E1} and \eqref{eq:E},
\begin{align}
\xi_j  (u_s; H, \lambda_2)                            & \le \frac{1}{\sqrt{n}}\sum_{p = 1}^{M}\lb\frac{2\gamma (\log n)^{\xi}}{\lambda_2}\rb^{p} + \sum_{p > M}\lb\frac{2\gamma C_0(\phi)}{\lambda_2}\rb^{p} \notag \\
&                             \le \frac{1}{\sqrt{n}}\sum_{p = 1}^{M}\lb\frac{2a_d\log^{\xi} n}{\sqrt{n\rho_{n}}\eta_{d}}\rb^{p} + \sum_{p > M}\lb\frac{2C_0(\phi)a_d}{\sqrt{n\rho_{n}}\eta_{d}}\rb^{p}\label{eq:Z-decomposition-dis-ass}.
\end{align}
Since $\frac{a_d\log^{\xi} n}{\sqrt{n\rho_{n}}\eta_{d}} \le \frac{1}{4}$, for sufficiently large $n$
\begin{equation}\label{eq:logxin-dis-ass}
  \frac{2\max\{a_d\log^{\xi} n, a_dC_0(\phi)\}}{\sqrt{n\rho_{n}}\eta_{d}} \le \frac{1}{2}.
\end{equation}
Then \eqref{eq:Z-decomposition-dis-ass} implies that
\begin{align}
  \xi_j(u_s; H, \lambda_2) &\le \frac{1}{\sqrt{n}}\frac{2a_d\log^{\xi} n}{\sqrt{n\rho_{n}}\eta_{d}}\sum_{p = 1}^{M}2^{-(p - 1)} + \frac{2a_dC_0(\phi)}{\sqrt{n\rho_{n}}\eta_{d}}\sum_{p > M}2^{-(p - 1)}\notag\\
\le &\frac{1}{\sqrt{n}}\frac{4a_d\log^{\xi} n}{\sqrt{n\rho_{n}}\eta_{d}} + \frac{4C_0(\phi)a_d}{\sqrt{n\rho_{n}}\eta_{d}} 2^{-M}.  \label{eq:zetaj1-dis-ass}
\end{align}

Again, as $M = \log^{\xi}n / 16$ and $\xi > 1$, $2^{-M}$ decays faster than any polynomial of $n$, for sufficiently large $n$,
\begin{equation}\label{eq:2-M-dis-ass}
  2^{-M}\le \frac{1}{\sqrt{n}}, \quad C_0(\phi)\le \log^{\xi}n.
\end{equation}
This simplifies \eqref{eq:zetaj1-dis-ass} to 
\begin{equation}
  \label{eq:zetaj-dis-ass}
  \xi_j(u_s; H, \lambda_2) \le \frac{1}{\sqrt{n}}\frac{8a_d\log^{\xi} n}{\sqrt{n\rho_{n}}\eta_{d}}.
\end{equation}

By \eqref{eq:deterministic1} and \eqref{eq:zetaj-dis-ass}, on event $E$, 
\begin{align}
  |(\tilde{u}_2-u_2)_j| \le&  \frac{1}{\sqrt{n}}\left( \frac{4C_0(\phi)^{2}a_d^2}{n\rho_{n}\eta_{d}^{2}}+\frac{2C_0(\phi)a_d}{\sqrt{n\rho_{n}}\eta_{d}}\right) + \frac{1}{\sqrt{n}}\frac{16a_d\log^{\xi} n}{\sqrt{n\rho_{n}}\eta_{d}} \notag\\
& + \frac{1}{\sqrt{n}}\frac{4\sqrt{2}C_0(\phi)a_d}{\sqrt{n\rho_{n}}\eta_{d}}\lb 1 + \frac{8a_d\log^{\xi} n}{\sqrt{n\rho_{n}}\eta_{d}}\rb \frac{\sum_{s\not = 2, s\le K}|\lambda_{s}|}{|\lambda_{2}|}\notag\\
\le & \frac{1}{\sqrt{n}}\frac{22a_d\log^{\xi} n}{\sqrt{n\rho_{n}}\eta_{d}} + \frac{1}{\sqrt{n}}\frac{12\sqrt{2}C_0(\phi)a_d}{\sqrt{n\rho_{n}}\eta_{d}}\frac{\sum_{s\not = 2, s\le K}|\lambda_{s}|}{|\lambda_{2}|}\notag\\
\le & \frac{1}{\sqrt{n}}\frac{a_d}{\sqrt{n\rho_{n}}\eta_{d}} \lb 22\log^{\xi}n + 12\sqrt{2}C_0(\phi)\frac{\sum_{s=1}^{K}|\lambda_{s}|}{\lambda_{2}}\rb.\label{eq:deterministic3}
\end{align}
where the second inequality uses \eqref{eq:c0q}, \eqref{eq:logxin-dis-ass} and \eqref{eq:2-M-dis-ass}. By Theorem \ref{thm:eigen}, 
\[\lambda_{1} = m\lb p_0 + \sum_{r=1}^{d}2^{r-1}p_r\rb\]
and for any $j\in [d]$,
\[\lambda_{2^{j-1}+1} = \ldots = \lambda_{2^{j}} = m\lb p_0 + \sum_{r=1}^{d - j} 2^{r-1}p_r - 2^{d - j}p_{d - j + 1}\rb.\]
Since $p_{0} < p_{1} < \ldots < p_{d}$, 
\[\lambda_{2^{j}} = m(p_{0} - p_{d - j + 1}) + m\sum_{r = 1}^{d - j}2^{r - 1}(p_r - p_{d - j + 1})\le 0.\]
As a result,
\begin{align*}
  \frac{1}{m}\sum_{s=1}^{K}|\lambda_{s}| & = \frac{1}{m}\lambda_{1} - \frac{1}{m}\sum_{j=1}^{d}2^{j-1}\lambda_{2^{j}}\\
& = p_0 + \sum_{r=1}^{d}2^{r-1}p_r - \sum_{j=1}^{d}2^{j-1}\lb p_0 + \sum_{r=1}^{d - j} 2^{r-1}p_r - 2^{d - j}p_{d - j + 1}\rb\\
& = p_0 \lb 1 - \sum_{j=1}^{d}2^{j - 1}\rb + \sum_{r=1}^{d}2^{r - 1}p_{r} \lb 1 - \sum_{j=1}^{d - r}2^{j - 1} + 2^{d - r}\rb\\
& = \lb\sum_{j=1}^{d}2^{r}p_{r}\rb - (2^{d} - 2)p_{0}\le 2^{d + 1}p_{d}.
\end{align*}
Thus, 
\[\frac{\sum_{s=1}^{K}|\lambda_{s}|}{|\lambda_{2}|}\le \frac{m2^{d+1}p_{d}}{n\rho_{n}\eta_{d}} = \frac{2a_{d}}{\eta_{d}}.\]

Plugging this into \eqref{eq:deterministic3}, we conclude that, on event $E$, 
\[|(\tilde{u}_2-u_2)_j| \le\frac{1}{\sqrt{n}}\frac{a_d}{\sqrt{n\rho_{n}}\eta_{d}} \lb 22\log^{\xi}n + 24\sqrt{2}C_0(\phi)a_{d}\eta_{d}^{-1}\rb.\]
The proof of dis-assortative case is then completed by setting $C_1(\phi) = 22 + 24\sqrt{2}C_0(\phi)$.

\end{proof}


\begin{coro}\label{cor:KConcentration}
 Fix $\phi > 0$. Under the settings of Lemma \ref{lem:KConcentration}, there exists a constant $C_{2}(\phi) > 0$ which only depends on $\phi$ and $\zeta\in\{-1, +1\}$ such that, for some $\xi > 1$, if 
\[\frac{\max\{1,a_d\}}{\sqrt{n\rho_n}\eta_{d}} \max\left\{\log^{\xi} n, \frac{\max\{1, a_{d}\}}{\eta_{d}}\right\}\le C_{2}(\phi),\]
then
\[ \sign(\td{u}_{2j}) = \sign(u_{2j})\zeta\mbox{ for all }j\]
with probability at least $1-2n^{-\phi}$.
\end{coro}

\begin{proof}
Without loss of generality, assume $\langle \td{u}_{2}, u_{2}\rangle \ge 0$ and take $s = 1$. Let $C_{2}(\phi) = 1 / \max\{4, 2C_1(\phi)\}$ where $C_1(\phi)$ is defined in Lemma \ref{lem:KConcentration}. Under our condition, the condition \eqref{eq:KConcentration_cond} in Lemma \ref{lem:KConcentration} is satisfied since
\[\frac{\max\{1,a_d\}}{\sqrt{n\rho_n}\eta_{d}}\le \frac{\max\{1,a_d\}\log^{\xi}n}{\sqrt{n\rho_n}\eta_{d}}\le C_{2}(\phi) \le 1/4.\]
Thus by Lemma \ref{lem:KConcentration}, 
\[\norm{\tilde{u}_2 - u_2}_{\infty} = \min_{\zeta\in\{-1, +1\}}\norm{\tilde{u}_2 - \zeta u_2}_{\infty}\le C_1(\phi)\frac{C_{2}(\phi)}{\sqrt{n}}\le \frac{1}{2\sqrt{n}},\]
with probability at least $1-2n^{-\phi}$. The proof is completed by noting $|u_{2j}|\equiv 1 / \sqrt{n}$.
\end{proof}

\begin{proof}[\textbf{Proofs of Theorem~\ref{thm:consistency-ass} and Theorem~\ref{thm:consistency-dis-ass}}]
We first prove the theorem assuming a perfect stopping rule without error is used. Define
\[T(\ell; \phi) = 1 + 2^{\phi+1}T(\ell-1; \phi), \quad T(0; \phi) = 0.\]
Then
\[\frac{T(\ell; \phi)}{2^{\ell(\phi + 1)}} = \frac{T(\ell-1;\phi)}{2^{(\ell-1)(\phi+1)}} + \frac{1}{2^{\ell(\phi + 1)}} = \ldots = T(0; \phi) + \sum_{i=1}^{\ell}\frac{1}{2^{i(\phi + 1)}} \le 1.\]
Thus $T(\ell; \phi)\le 2^{\ell(\phi + 1)} = (2^{\ell})^{\phi + 1} = K^{\ell (\phi + 1) / d}$. It is left to show that Theorem \ref{thm:consistency-ass} and Theorem \ref{thm:consistency-dis-ass} hold with probability $1 - 2T(\ell; \phi)n^{-\phi}$, with $C(\phi) = C_{2}(\phi)$ where $C_{2}(\phi)$ is defined in Corollary \ref{cor:KConcentration}.

We prove this by induction in $d$ and $\ell$. Starting with the case of $\ell = 1$, let $E_{\mathrm{top}}$ be the event that the top two mega-communities which form the first split are exactly recovered by HCD-Sign. For convenience we write $\td{a}_{\ell}$ for $\max\{1, a_{\ell}\}$. By definition, $\eta_{(1)} = \eta_d$ and $\nu_{(1)} = 2^{-1/2}\td{a}_{d}$.    Note that

\begin{align*}
 \frac{\td{a}_{d} \max\{\log^{\xi}n, \td{a}_{d}\eta_{d}^{-1}\}}{\sqrt{n\rho_n}\eta_{d}}  & =  \frac{\sqrt{2}}{\sqrt{n\rho_n}\eta_{d}}2^{-1/2} \td{a}_{d}\max\{\log^{\xi}n, \td{a}_{d}\eta_{d}^{-1}\} \\
  &=  \sqrt{\frac{K^{1/d}}{n\rho_{n}}} \frac{\max\{\log^{\xi} n, \td{a}_{d}\eta_{(1)}^{-1}\}}{\eta_{(1)}} (2^{-1/2} \td{a}_{d}).
\end{align*}
In the assortative case, $\td{a}_{d} = 1$. By \eqref{eq:main_cond} in Theorem \ref{thm:consistency-ass}, we have
\[\sqrt{\frac{K^{1/d}}{n\rho_{n}}} \frac{\max\{\log^{\xi} n, \td{a}_{d}\eta_{(1)}^{-1}\}}{\eta_{(1)}} (2^{-1/2} \td{a}_{d})\le \sqrt{\frac{K^{1/d}}{n\rho_{n}}} \frac{\max\{\log^{\xi} n, \eta_{(1)}^{-1}\}}{\eta_{(1)}} < C(\phi).\]
In the dis-assortative case, $\td{a}_{d} = a_{d}$. By \eqref{eq:main_cond-dis-ass} in Theorem \ref{thm:consistency-dis-ass}, we have
\[\sqrt{\frac{K^{1/d}}{n\rho_{n}}} \frac{\max\{\log^{\xi} n, \td{a}_{d}\eta_{(1)}^{-1}\}}{\eta_{(1)}} (2^{-1/2} \td{a}_{d})\le \sqrt{\frac{K^{1/d}}{n\rho_{n}}} \frac{\max\{\log^{\xi} n, \sqrt{2}\nu_{(1)}\eta_{(1)}^{-1}\}}{\eta_{(1)}}\nu_{(1)} < \sqrt{2}C(\phi).\]
Thus, for both cases, 
\[ \frac{\td{a}_{d} \max\{\log^{\xi}n, \td{a}_{d}\eta_{d}^{-1}\}}{\sqrt{n\rho_n}\eta_{d}} < C_{2}(\phi).\]

Then Corollary \ref{cor:KConcentration} implies that
\begin{equation}
  \label{eq:Etop}
  \p(E_{\mathrm{top}})\ge 1 - 2n^{-\phi}. 
\end{equation}
When $d = 1$, the claim holds because $T(1; \phi) = 1$. For $d > 1$, let $T_{\ell}$ and $T_{r}$ be the left and the right branches of the root node.   Let $E_{\mathrm{left}}$  and $E_{\mathrm{right}}$ be the events that $T_{\ell}$ and $T_{r}$, respectively, are exactly recovered by HCD-Sign. Note that each branch is itself a BTSBM with parameters $(n / 2; \rho_{n}, a_{1}, \ldots, a_{d - 1})$.    We can again apply Corollary \ref{cor:KConcentration} to each branch with $n' = n / 2$, $d' = d-1$, $\ell' = \ell - 1$. Notice that in both assortative and dis-assortative setting, if we use the definitions of \eqref{eq:delta} and \eqref{eq:nu-new}
\begin{align*}
  \eta'_{(\ell')} &= \min\{\eta_{d'-r+1}: r\in [\ell']\}\\
                    &=  \min\{\eta_{d-r+1}: 2\le r \le \ell\}\\
& \ge  \min\{\eta_{d-r+1}: r\in [\ell]\}  = \eta_{(\ell)}.
\end{align*}
and
\begin{align*}
  \nu'_{(\ell')} &= \max\{2^{-\frac{\ell'-r+1}{2}}\max\{1,a_{d'-r+1}\}: r\in [\ell']\}\\
  & =  \max\{2^{-\frac{\ell-1-r+1}{2}}\max\{1,a_{d-1-r+1}\}: r\in [\ell-1]\}\\
  & = \max\{2^{-\frac{\ell-r+1}{2}}\max\{1,a_{d-r+1}\}:2\le r \le \ell\}\\
& \le  \max\{2^{-\frac{\ell-r+1}{2}}\max\{1,a_{d-r+1}\}:r \in [\ell]\} = \nu_{(\ell)}.
\end{align*}

Thus, in the assortative case, 
\[\sqrt{\frac{K^{1/d}}{n\rho_{n}}} \frac{\max\{\log^{\xi} n, \eta_{(\ell')}^{-1}\}}{\eta_{(\ell')}} \le \sqrt{\frac{K^{1/d}}{n\rho_{n}}} \frac{\max\{\log^{\xi} n, \eta_{(\ell)}^{-1}\}}{\eta_{(\ell)}} < C(\phi)\]
and in the dis-assortative case,
\[\sqrt{\frac{2^{\ell'}}{n'\rho_{n}}} \frac{\nu'_{(\ell')}\max\{\log^{\xi} n', \nu'_{(\ell')}\eta_{(\ell')}^{'-1}\}}{\eta'_{(\ell')}}\le \sqrt{\frac{2^{\ell}}{n\rho_{n}}} \frac{\nu_{(\ell)}\max\{\log^{\xi} n, \nu_{(\ell)}\eta_{(\ell)}^{-1}\}}{\eta_{(\ell)}} < C(\phi).\]


Then the induction hypothesis implies that
\[\p(E_{\mathrm{left}}) \ge 1 - 2T(\ell'; \phi)(n / 2)^{-\phi}, \quad \p(E_{\mathrm{right}}) \ge 1 - 2T(\ell'; \phi)(n / 2)^{-\phi}.\]
Let $E = E_{\mathrm{top}}\cap E_{\mathrm{left}}\cap E_{\mathrm{right}}$ be the event that $T$ is exactly recovered, i.e., 
\[\min_{\Pi \in \text{Perm}(q)}\Pi(\mc(\hat{T},\hat{c},q) )= \mc(T,c,q), \quad \text{for all~~}q\le \ell\]
where $\text{Perm}(q)$ contains all potential permutations of labels in $S_q$. A union bound then gives
\[\p(E) \ge 1 - 2(1 + T(\ell'; \phi)2^{\phi + 1})n^{-\phi} = 1 - 2T(\ell; \phi)n^{-\phi} , \]
as claimed.

Finally, if one uses a consistent stopping rule with rate $\phi'$, the only additional components we need to insert is we need to make sure that the stopping does not make error up to level $\ell$. By the union of events, it is easy to see that the error probability is upper bounded by 
$$\sum_{r=0}^\ell2^{(\phi'+1)r}n^{-\phi'} = K^{(\phi'+1)\ell/d}n^{-\phi'}.$$

Finally, at the end of $\ell=d$, the stopping rule make false split the final $K$ communities with probability no more than $K(n/K)^{-\phi'} = K(1+\phi')n^{-\phi}$.

This completes the proof.

 \end{proof}

Before proving Proposition~\ref{prop:NB-consistency}, we will need the result about model selection consistency of the non-backtracking matrix from \cite{le2015estimating}, as stated next.

\begin{lem}[Modified version of Theorem 4.2 of \cite{le2015estimating}]\label{lem:LeConsistency}
Given a network $A$ generated from probability matrix $P$ as the inhomogeneous Erd\"{o}s-Renyi model, assume all nodes have the same expected degree, with $\mu = \sum_{j=1}^nP_{ij}$ for all $i \in [n]$. Furthermore, assume 
\begin{equation}\label{eq:NB-condition3}
\mu^5\cdot \max_{ij}P_{ij} \le n^{-1/13}
\end{equation}

and for some constant $C>0$, 
\begin{equation}\label{eq:NB-condition1}
\mu \ge C\log{n}.
\end{equation}
Also, assume $P$ is rank $K$ with nonzero eigenvalues $|\lambda_1| \ge |\lambda_2| \cdots \ge |\lambda_K| > 0$. We have the following two properties:
\begin{enumerate}
\item If $K \ge 2$ and 
\begin{equation}\label{eq:NB-condition2}
|\lambda_1| \ge |\lambda_2| \ge 4\sqrt{\mu}+C\sqrt{\log{n}},
\end{equation}
 then for sufficiently large $n$, with probability at least $1-1/n$, the non-backtracking matrix has at least 2 real eigenvalues larger than 
$$\lb 1+C\lb\frac{\log{n}}{\mu}\rb^{1/4}\rb\sqrt{\mu}.$$
\item If $K = 1$ and $|\lambda_1|\ge 4\sqrt{\mu}+C\sqrt{\log{n}}$, then for sufficiently large $n$, with probability at least $1-1/n$, the non-backtracking matrix has exactly 1 real eigenvalue larger than 
$$\lb 1+C\lb\frac{\log{n}}{\mu}\rb^{1/4}\rb\sqrt{\mu}.$$
\end{enumerate}
\end{lem}

\begin{proof}[Proof of Proposition~\ref{prop:NB-consistency}]
For notational simplicity, we will ignore the difference between $P$ and $\tilde{P}$. As can be seen clearly, this does not change our result, except on slightly different constants.  We assume $\ell = d$ without loss of generality. It is left to show that the stopping rule is consistent with rate $1$ for any node. Note that the subgraph corresponding to a mega-community at layer $r$ is drawn from a BTSBM with parameters $(p_{0}, \ldots, p_{d - r + 1})$. Fix a mega- community at layer $r$.

By Theorem~\ref{thm:eigen2} and noting that each block in the mega-community is of size $m$, we have
$$\mu  = m\rho_n\lb 1+\sum_{j=1}^{d - r + 1}2^{j-1}a_j\rb,$$
while
$$|\lambda_2| =  m\rho_n\bigg|1+\sum_{j=1}^{d - r}2^{r-1}a_r-2^{d- r + 1}a_{d - r + 1}\bigg|.$$
First we verify \eqref{eq:NB-condition3}. Since the megacommunity has size $2^{d - r + 1}m$ and the maximum probability is $p_{0} = \rho_{n}\max\{1, a_{d - r + 1}\}\le \rho_{n}\max\{1, a_{d}\}$. Additionally, 
\[\mu\le m\rho_{n}2^{d - r + 1}\max\{1, a_{d}\}\le n\rho_{n}\max\{1, a_{d}\}.\]
By \eqref{eq:stopping-rule-3'},
\[\mu^{5}(n\max P_{ij}) \le (n\rho_{n}\max\{1, a_{d}\})^{6}\le n^{12/13}\Longrightarrow \mu^{5}\max P_{ij}\le n^{-1/13}\le (2^{d-r+1}m)^{-1/13}.\]
Thus, \eqref{eq:NB-condition3} holds. By \eqref{eq:stopping-rule-2'}, 
\[\log n\le m\rho_{n}\le m\rho_{n}\lb 1+\sum_{j=1}^{d- r +1}2^{j-1}a_j\rb = \mu.\]
Thus \eqref{eq:NB-condition1} holds with $C = 1$. On the other hand, by definition of $\zeta_{(\ell)}$, 
\[\frac{25}{m\rho_{n}} \le \frac{(1+\sum_{r=1}^{d-r}2^{r-1}a_r-2^{d-r}a_{d-r+1})^2}{1+\sum_{r=1}^{d-r+1}2^{r-1}a_r} = \frac{\lambda_{2}^{2}}{m\rho_{n} \mu}.\]
Thus, \eqref{eq:NB-condition2} holds with $C = 1$ because 
\[\lambda_{2}\ge 5\sqrt{\mu} \ge 4\sqrt{\mu} + \sqrt{\log n}.\]
Since the conditions \eqref{eq:NB-condition3}, \eqref{eq:NB-condition1} and \eqref{eq:NB-condition2} hold for any $r$, the stated consistency is a result of Lemma \ref{lem:LeConsistency}.

\end{proof}

  \section{The eigenstructure of the BTSBM with unequal block sizes}
  
 Our main presentation assumed the communities have equal sizes for simplicity and clarity. This assumption can be relaxed to community sizes of the same order within the same theoretical framework, at the cost of more involved calculations.  
  
  \begin{thm}\label{thm:unequal}
  Let $P$ be the probability matrix of an assortative or a dis-assortative BTSBM with $K = 2^{d}$ communities and parameters $p_0, p_1, \ldots, p_d$. Let $n_{i}$ denote the size of the $i$-th community. Let $\lambda_{2}$ be the second largest eigenvalue of $P$ in absolute value, and $u_{2}$ the corresponding eigenvector, and $\delta_{2}$ the gap between $\lambda_{2}$ and $\lambda_3$. Then there exists a constant $c > 0$ that only depends on $K$ and $p_0, \ldots, p_d$ such that if
  \begin{equation}
    \label{eq:1+c}
    \max_{i}\bigg|\sqrt{\frac{Kn_{i}}{n}} - 1\bigg|\le c, 
  \end{equation}
  then $\lambda_{2}$  has multiplicity 1,
  \begin{equation}
    \label{eq:Delta2}
    \delta_{2}\ge \left\{
      \begin{array}{ll}
        n\min\{p_d / 2, (p_{d-1} - p_{d}) / 4\} & (\mbox{assortative case})\\
        n(p_{d} - p_{d - 1}) / 4 & (\mbox{disassortative case})
      \end{array}\right.,
  \end{equation}
  and
  \begin{equation}
    \label{eq:u2i}
    \frac{1}{2\sqrt{n}}\le \min_{i}|u_{2i}|\le \max_{i}|u_{2i}|\le \frac{2}{\sqrt{n}}.
  \end{equation}
\end{thm}

\begin{proof}

  First note that we can replace $P$ by $\td{P} = P + p_{0}I$ since this does not affect the eigenvalue multiplicity, the eigengap and the eigenvector.  Then 
  \[\td{P} = ZBZ^{T}.\]
  Let $\Pi = \diag(\sqrt{Kn_{1}/n}, \ldots, \sqrt{Kn_{K}/n})$ and
  \[Q =
    \begin{bmatrix}
      \one_{n_{1}} / \sqrt{n_{1}} & 0 & \cdots & 0\\
      0 & \one_{n_{2}} / \sqrt{n_{2}} & \cdots & 0\\
      \vdots & \vdots & \ddots & \vdots\\
      0 & 0 & \cdots & \one_{n_{K}} / \sqrt{n_{K}}
    \end{bmatrix}.
  \]
  Then $Q^{T}Q = I$ and
  \[\td{P} = \frac{n}{K}Q^{T}(\Pi B\Pi)Q.\]
  Let $\lambda_{j}'$ be the $j$-th largest eigenvalue of $\Pi B\Pi$ and $v_{j}$  the corresponding eigenvector. Further let $\delta_{j}' = \min\{\lambda_{j-1}' - \lambda_{j}', \lambda_{j}' - \lambda_{j+1}'\}$. Then
  \begin{equation}
    \label{eq:Delta2prime}
    \lambda_{1} = \frac{n}{K}\lambda_{1}', \quad \lambda_{2} = \frac{n}{K}\lambda_{2}', \quad \lambda_{3} = \frac{n}{K}\lambda_{3}' \Longrightarrow \delta_{2} = \frac{n}{K}\delta_{2}',
  \end{equation}
  and
  \begin{equation}
    \label{eq:u2}
    u_{2} =     \begin{bmatrix}
      v_{12}\one_{n_{1}} / \sqrt{n_{1}}\\
      v_{22}\one_{n_{2}} / \sqrt{n_{2}}\\
      \vdots\\
      v_{k2}\one_{n_{K}} / \sqrt{n_{K}}
    \end{bmatrix}.
  \end{equation}
  Let $c \le 1$. Then
  \begin{align*}
    \|\Pi B\Pi - B\| &\le 2\|B (\Pi - I)\| + \|(\Pi - I)B(\Pi - I)\|\\
    & \le (2c + c^{2})\|B\| \le 3c\|B\|.
  \end{align*}
  Let $\lambda_j(B)$ be the $j$-th largest eigenvalue of $B$.   Consider 
  \[c < \frac{\min\{\lambda_1(B) - \lambda_2(B), \lambda_2(B) - \lambda_3(B)\}}{12\|B\|},\]
  By Weyl's inequality,
  \[|\lambda_{j}' - \lambda_j(B)|\le \|\Pi B\Pi - B\|\le 3c\|B\| < \frac{\min\{\lambda_1(B) - \lambda_2(B), \lambda_2(B) - \lambda_3(B)\}}{4}.\]
  This implies
  \begin{align*}
    \lambda_{1}' - \lambda_{2}' \ge \lambda_{1} - \lambda_{2} - \frac{\min\{\lambda_1(B) - \lambda_2(B), \lambda_2(B) - \lambda_3(B)\}}{2}\ge \frac{1}{2}(\lambda_{1} - \lambda_{2}),
  \end{align*}
  and
  \begin{align*}
    \lambda_{2}' - \lambda_{3}' \ge \lambda_{2} - \lambda_{3} - \frac{\min\{\lambda_1(B) - \lambda_2(B), \lambda_2(B) - \lambda_3(B)\}}{2}\ge \frac{1}{2}(\lambda_{2} - \lambda_{3}).
  \end{align*}
  Thus $\lambda_{2}'$ has multiplicity $1$ and the eigengap is at least the half of the eigengap of $B$. By Theorem \ref{thm:eigenB},
  \[\min\{\lambda_1(B) - \lambda_2(B), \lambda_2(B) - \lambda_3(B)\} \ge \left\{
      \begin{array}{ll}
        K\min\left\{p_{d}, (p_{d-1} - p_{d}) / 2\right\} & (\mbox{assortative case})\\
        K(p_{d-1} - p_{d}) / 2 & (\mbox{disassortative case})\\
      \end{array}\right..\]
  Therefore, \eqref{eq:Delta2} follows from \eqref{eq:Delta2prime}.

Let $v_2(B)$ denote the second eigenvector of $B$ and assume $v_{2}^{T}v_2(B)\ge 0$ without loss of generality.   From the Davis-Kahan theorem ~\citep[e.g.][Theorem 6.9]{cape2019two}
  \[\|v_2(B) - v_{2}\|_{2}\le \frac{2\|\Pi B \Pi - B\|}{\min\{\lambda_1(B) -\lambda_2(B), \lambda_2(B) - \lambda_3(B)\}}\le \frac{6c\|B\|}{\min\{\lambda_1(B) -\lambda_2(B), \lambda_2(B) - \lambda_3(B)\}}.\]
Consider 
\[c < \frac{\min\{\lambda_1(B) -\lambda_2(B),\lambda_2(B) - \lambda_3(B)\}}{24\|B\|\sqrt{K}},\]
then
\[\|v_2(B) - v_{2}\|_{\infty}\le \|v_2(B) - v_{2}\|_{2}\le \frac{1}{4\sqrt{K}}.\]
By Theorem \ref{thm:eigenB},
\[|v_{2i}(B)| = \frac{1}{\sqrt{K}}\Longrightarrow \frac{3}{4\sqrt{K}}\le |v_{2i}|\le \frac{5}{4\sqrt{K}}.\]
By \eqref{eq:u2}, the condition \eqref{eq:1+c} and the restriction  $c\le 1$,
\[|u_{2i}| \ge \frac{1}{\max_{i}\sqrt{n_{i}}}\min_{i}|v_{2i}|\ge \sqrt{\frac{K}{(1 + c)n}}\frac{3}{4\sqrt{K}}\ge \frac{1}{2\sqrt{n}}\]
and
\[|u_{2i}|\le \frac{1}{\min_{i}\sqrt{n_{i}}}\max_{i}|v_{2i}|\le \sqrt{\frac{(1 + c)K}{n}}\frac{5}{4\sqrt{K}}\le \frac{2}{\sqrt{n}}\]
\end{proof}
Theorem \ref{thm:unequal} implies a certain degree of robustness of the eigenstructure to unequal community sizes. When $\max_{i}n_{i}$ differs from $\min_{i}n_{i}$ by a constant, the eigengap and the eigenvector behave essentially the same as they do with equal community sizes, except for a multiplicative constant. We can then prove Theorem \ref{thm:consistency-ass} and Theorem \ref{thm:consistency-dis-ass} for this case using exactly the same arguments.

\section{Proof of Lemma \ref{lem:deterministicControl}}\label{app:eldridge}
Our Lemma \ref{lem:deterministicControl} is a corrected version of Theorem 9 of \cite{eldridge2017unperturbed}, whose original proof contained a mistake we fixed.   After we made the earlier draft of this paper online,  \cite{eldridge2017unperturbed} also posted a similar corrected version. The Lemma \ref{lem:deterministicControl} we use for the proof has a slightly different form compared to that of \cite{eldridge2017unperturbed}. For completeness, we give a full proof of Lemma \ref{lem:deterministicControl} here.

\begin{lem}[Theorem 8 of \cite{eldridge2017unperturbed}]\label{lem:neumann}
  Fix any $u_{1}, \ldots, u_{n}$, $\td{u}_{1}, \ldots, \td{u}_{n}$ and $t\in [n]$. Suppose that $\|H\| < \td{\lambda}_{t}$.  Then 
\[\td{u}_{t} = \sum_{s=1}^{n}\frac{\lambda_{s}}{\td{\lambda}_{t}} \langle  \td{u}_{t}, u_{s}\rangle \sum_{p\ge 0}\lb\frac{H}{\td{\lambda}_{t}}\rb^{p}u_{s}.\]
\end{lem}



\begin{lem}[Theorem 2 of \cite{yu2015useful}]\label{lem:davis_kahan}
  Let $\Sigma, \hat{\Sigma} \in \bR^{p\times p}$ be symmetric matrices with eigenvalues $\lambda_{1}\ge \lambda_{2}\ge \ldots \ge \lambda_{p}$ and $\td{\lambda}_{1}\ge \td{\lambda}_{2}\ge \ldots \ge \td{\lambda}_{p}$ respectively. Fix $1 \le r \le s \le p$ and assume that $\min\{\lambda_{r - 1} - \lambda_{r}, \lambda_{s} - \lambda_{s + 1}\} > 0$ where $\lambda_{0} = \infty$ and $\lambda_{p+1} = -\infty$. Let $d = s - r + 1$ and let $V = (v_{r}, v_{r+1}, \ldots, v_{s}) \in \bR^{p\times d}$ and $\hat{V} = (\hat{v}_{r}, \hat{v}_{r+1}, \ldots , \hat{v}_{s}) \in \bR^{p\times d}$ have orthonormal columns satisfying $\Sigma v_{j} = \lambda_{j}v_{j}$ and $\hat{\Sigma}\hat{v}_{j} = \hat{\lambda}_{j}\hat{v}_{j}$ for $j = r, r + 1, \ldots, s$. Then
\[\|\sin \Theta(V, \hat{V})\|_{\mathrm{F}}\le \frac{2\sqrt{d}\|\hat{\Sigma} - \Sigma\|_{\mathrm{op}}}{\min\{\lambda_{r - 1} - \lambda_{r}, \lambda_{s} - \lambda_{s + 1}\}},\]
where $\|\cdot\|_{\mathrm{F}}$ denotes the Frobenius norm and $\Theta(V, \hat{V})$ denotes the principal angle matrix between $V$ and $\hat{V}$. 
\end{lem}

\begin{proof}[\textbf{Proof of Lemma \ref{lem:deterministicControl}}]

Note that $u_{t}$ is unique up to the sign, because $\lambda_{t}$ is unique and both sides of the bound are invariant to signs.    This allows us to assume $\langle u_{t}, \td{u}_{t}\rangle \ge 0$ without loss of generality. Let $r = s = t$ in Lemma \ref{lem:davis_kahan}.  Since $d = 1$ we have
\begin{equation}
  \label{eq:sintheta}
  \sin \Theta(\td{u}_{t}, u_{t}) = \|\sin \Theta(\td{u}_{t}, u_{t})\|_{\mathrm{F}}\le \frac{2\|H\|}{\delta_{t}}.
\end{equation}
This implies 
\begin{equation}\label{eq:basis1}
  1 - \langle \td{u}_{t}, u_{t}\rangle\le 1 - \langle \td{u}_{t}, u_{t}\rangle^{2} = \sin^{2} \Theta(\td{u}_{t}, u_{t}) \le \frac{4\|H\|^{2}}{\delta_{t}^{2}}.
\end{equation}
Let $u_{1}, \ldots, u_{t - 1}, u_{t + 1}, \ldots, u_{n}$ be any set of orthonormal vectors that span the orthogonal subspace of $u_{t}$. Then 
\begin{align}
  \sum_{s\not = t}\langle \td{u}_{t}, u_{s}\rangle^{2} = 1 - \langle \td{u}_{t}, u_{t}\rangle^{2} \le 1 - \lb 1 - \frac{4\|H\|^{2}}{\delta_{t}^{2}}\rb^{2}\le \frac{8\|H\|^{2}}{\delta_{t}^{2}} \, ,
  \end{align}
  and therefore
\begin{align}
  \max_{s\not = t}|\langle \td{u}_{t}, u_{s}\rangle|\le \frac{2\sqrt{2}\|H\|}{\delta_{t}}.
    \label{eq:basis2}
\end{align}
Since $2\|H\| < \lambda_{t}$, by Weyl's inequality,
\begin{equation}
  \label{eq:tdlambdat_lambdat}
  |\td{\lambda}_{t}| \ge |\lambda_{t}| - \|H\| > \|H\|, \quad \mbox{and}\quad |\td{\lambda}_{t}| \ge |\lambda_{t}| - \|H\| > |\lambda_{t}| / 2.
\end{equation}
By Lemma \ref{lem:neumann}, for any $j$,
\begin{equation}\label{eq:psi}
  \td{u}_{t, j} = \sum_{s=1}^{n}\psi_{st, j}, \quad \mbox{where }\psi_{st} = \frac{\lambda_{s}}{\td{\lambda}_{t}} \langle  \td{u}_{t}, u_{s}\rangle \sum_{p\ge 0}\lb\frac{H}{\td{\lambda}_{t}}\rb^{p}u_{s}.
\end{equation}
Then 
\begin{align}
  |(\td{u}_{t} - u_{t})_{j}| & = \bigg|u_{t, j} - \sum_{s=1}^{n}\psi_{st, j}\bigg| \le \underbrace{|u_{t, j} - \psi_{tt, j}|}_{I_{1}} + \underbrace{\bigg|\sum_{s\not = t}\psi_{st, j}\bigg|}_{I_{2}}.\label{eq:thm12_1}
\end{align}
By definition, 
\begin{align}
  I_{1} &= \left|\lb u_{t} - \frac{\lambda_{t}}{\td{\lambda}_{t}}\langle \td{u}_{t}, u_{t}\rangle \sum_{p\ge 0}\lb\frac{H}{\td{\lambda}_{t}}\rb^{p}u_{t}\rb_{j}\right|\notag\\
& \le \left|u_{t, j} - \frac{\lambda_{t}}{\td{\lambda}_{t}}\langle \td{u}_{t}, u_{t}\rangle u_{t, j}\right| + \left|\frac{\lambda_{t}}{\td{\lambda}_{t}}\langle \td{u}_{t}, u_{t}\rangle \sum_{p\ge 1}\left[\lb\frac{H}{\td{\lambda}_{t}}\rb^{p}u_{t}\right]_{j}\right|\notag\\
& \le \left|1 - \frac{\lambda_{t}}{\td{\lambda}_{t}} \langle \td{u}_{t}, u_{t}\rangle\right||u_{t, j}| + 2|\langle \td{u}_{t}, u_{t}\rangle| \cdot \sum_{p\ge 1}\frac{|[H^{p}u_{s}]_{j}|}{|\td{\lambda}_{t}|^{p}}\label{ineq(i)}\\
& = \left|1 - \frac{\lambda_{t}}{\td{\lambda}_{t}} \langle \td{u}_{t}, u_{t}\rangle\right||u_{t, j}| + 2|\langle \td{u}_{t}, u_{t}\rangle| \cdot \sum_{p\ge 1}\lb\frac{|\lambda_{t}|}{|\td{\lambda}_{t}|}\rb^{p}\frac{|[H^{p}u_{s}]_{j}|}{|\lambda_{t}|^{p}}\notag\\
& \le \left|1 - \frac{\lambda_{t}}{\td{\lambda}_{t}} \langle \td{u}_{t}, u_{t}\rangle\right||u_{t, j}| + 2|\langle \td{u}_{t}, u_{t}\rangle| \cdot \xi_{j}(u_{t}; H, \lambda_{t}).\label{eq:thm12_2}
\end{align}
where \eqref{ineq(i)} uses the first part of \eqref{eq:tdlambdat_lambdat} and \eqref{eq:thm12_2}  uses the second part of \eqref{eq:tdlambdat_lambdat} and the definition of $\xi_{j}$. The first term in \eqref{eq:thm12_2} can be bounded as follows:  
\begin{align}
  \left|1 - \frac{\lambda_{t}}{\td{\lambda}_{t}} \langle \td{u}_{t}, u_{t}\rangle\right| & \le  |1 - \langle \td{u}_{t}, u_{t}\rangle| +  \left|1 - \frac{\lambda_{t}}{\td{\lambda}_{t}}\right| |\langle \td{u}_{t}, u_{t}\rangle|\notag\\
& \le \frac{4\|H\|^{2}}{\delta_{t}^{2}} + \frac{|\td{\lambda}_{t} - \lambda_{t}|}{|\td{\lambda}_{t}|}|\langle \td{u}_{t}, u_{t}\rangle|\label{ineq2(i)}\\
& \le \frac{4\|H\|^{2}}{\delta_{t}^{2}} + \frac{|\td{\lambda}_{t} - \lambda_{t}|}{|\td{\lambda}_{t}|}\label{ineq2(ii)}\\
&\le \frac{4\|H\|^{2}}{\delta_{t}^{2}} + \frac{\|H\|}{|\td{\lambda}_{t}|}\label{ineq2(iii)}\\
& \le \frac{4\|H\|^{2}}{\delta_{t}^{2}} + \frac{2\|H\|}{|\lambda_{t}|}, \label{ineq2(iv)}
\end{align}
where \eqref{ineq2(i)} uses \eqref{eq:basis1}, \eqref{ineq2(ii)} uses the Cauchy-Schwartz inequality, $|\langle \td{u}_{t}, u_{t}\rangle|\le \|\td{u}_{t}\|\|u_{t}\| = 1$, \eqref{ineq2(iii)} uses Weyl's inequality, $|\lambda_{t} - \td{\lambda}_{t}|\le \|H\|$ and \eqref{ineq2(iv)} uses the second part of \eqref{eq:tdlambdat_lambdat}. Thus, \eqref{eq:thm12_2} implies
\begin{align}
  I_{1} &\le \lb \frac{4\|H\|^{2}}{\delta_{t}^{2}} + \frac{2\|H\|}{|\lambda_{t}|}\rb|u_{t, j}| + 2|\langle \td{u}_{t}, u_{t}\rangle|\cdot \xi_{j}(u_{t}; H, \lambda_{t})\notag\\
&\le \lb\frac{4\|H\|^{2}}{\delta_{t}^{2}} + \frac{2\|H\|}{|\lambda_{t}|}\rb|u_{t, j}| + 2\xi_{j}(u_{t}; H, \lambda_{t}).\label{eq:I1}
\end{align}

Next we derive a bound for $I_{2}$. By definition \eqref{eq:psi} of $\psi_{st}$ and \eqref{eq:basis2}, it holds for any $j$ that
\begin{align}
  |\psi_{st, j}| & = \bigg|\frac{\lambda_{s}}{\td{\lambda}_{t}}\bigg| |\langle \td{u}_{t}, u_{s}\rangle|\cdot \bigg|\left[\sum_{p\ge 0}\lb\frac{H}{\td{\lambda}_{t}}\rb^{p}u_{s}\right]_{j}\bigg|\notag \le \frac{2\sqrt{2}\|H\|}{\delta_{t}}\bigg|\frac{\lambda_{s}}{\td{\lambda}_{t}}\bigg|\cdot \bigg|\left[\sum_{p\ge 0}\lb\frac{H}{\td{\lambda}_{t}}\rb^{p}u_{s}\right]_{j}\bigg|\notag\\
&\le \frac{4\sqrt{2}\|H\|}{\delta_{t}}\bigg|\frac{\lambda_{s}}{\lambda_{t}}\bigg|\cdot \bigg|\left[\sum_{p\ge 0}\lb\frac{H}{\td{\lambda}_{t}}\rb^{p}u_{s}\right]_{j}\bigg|\label{ineq3:(i)}\\
& \le \frac{4\sqrt{2}\|H\|}{\delta_{t}}\bigg|\frac{\lambda_{s}}{\lambda_{t}}\bigg|\cdot \lb |u_{s, j}| + \sum_{p\ge 1}\frac{|[H^{p}u_{s}]_{j}|}{\td{\lambda}_{t}^{p}}\rb\notag\\
& \le\frac{4\sqrt{2}\|H\|}{\delta_{t}}\bigg|\frac{\lambda_{s}}{\lambda_{t}}\bigg|\cdot \lb |u_{s, j}| + \xi_{j}(u_{s}; H, \lambda_{t})\rb, \label{ineq3:(ii)}
\end{align}
where \eqref{ineq3:(i)} uses the second part of \eqref{eq:tdlambdat_lambdat} and \eqref{ineq3:(ii)} uses the same argument as in \eqref{eq:thm12_2}.   Thus,
\begin{equation}\label{eq:I2}
  I_{2}\le \frac{4\sqrt{2}\|H\|}{\delta_{t}}\sum_{s\not = t}\bigg|\frac{\lambda_{s}}{\lambda_{t}}\bigg|\cdot \lb |u_{s, j}| + \xi_{j}(u_{s}; H, \lambda_{t})\rb.
\end{equation}
The proof is then completed by combining \eqref{eq:I1} and \eqref{eq:I2}.
\end{proof}

\section{Additional results for the anemia gene network}
\label{sec:RSC-gene}

\begin{figure}[H]
\centering
\begin{subfigure}[t]{0.5\textwidth}
\centering
\includegraphics[width=\textwidth]{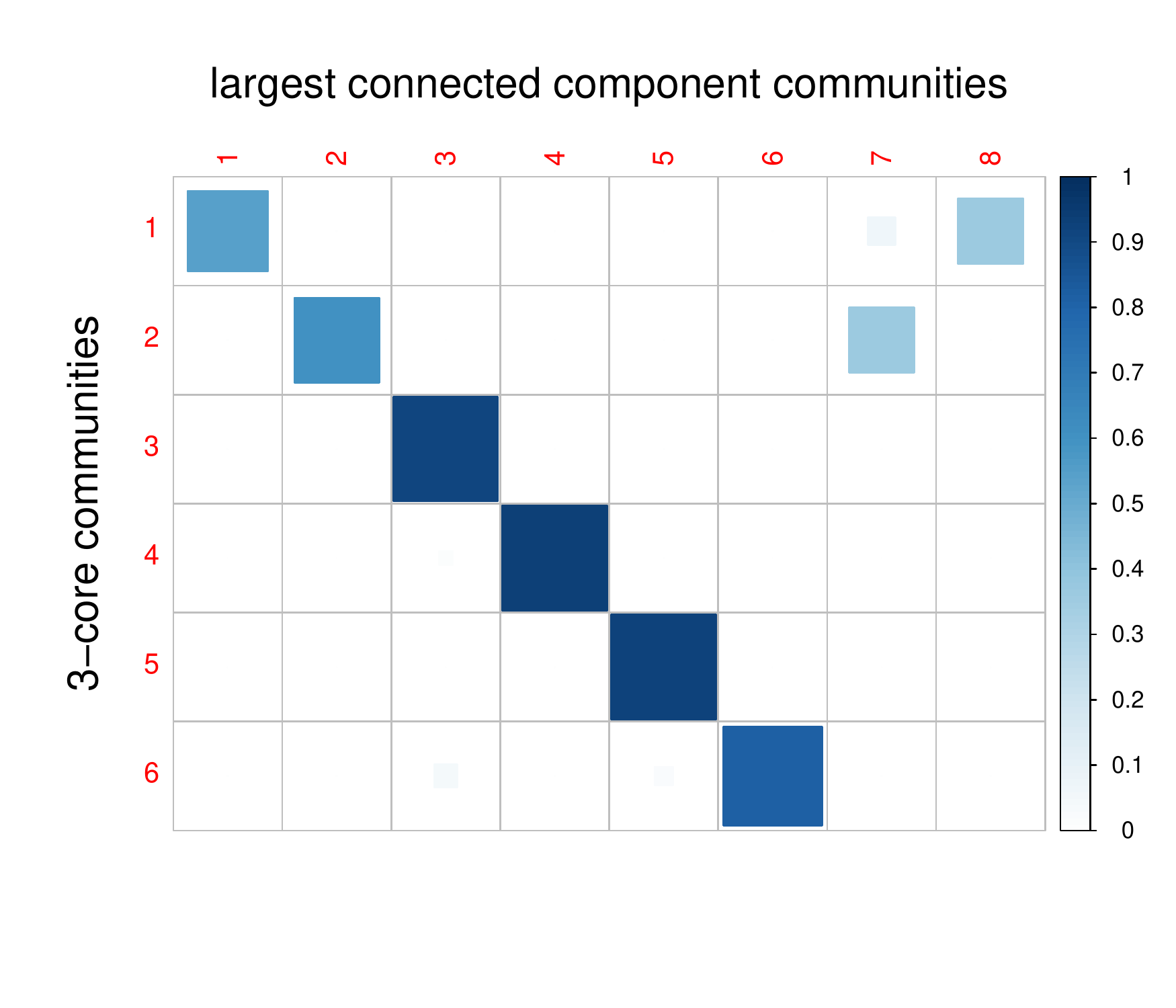}
\caption{Confusion matrix by Jaccard index.}
\label{fig:Confusion}
\end{subfigure}%
\hfill
\begin{subfigure}[t]{0.5\textwidth}
\centering
\includegraphics[width=\textwidth]{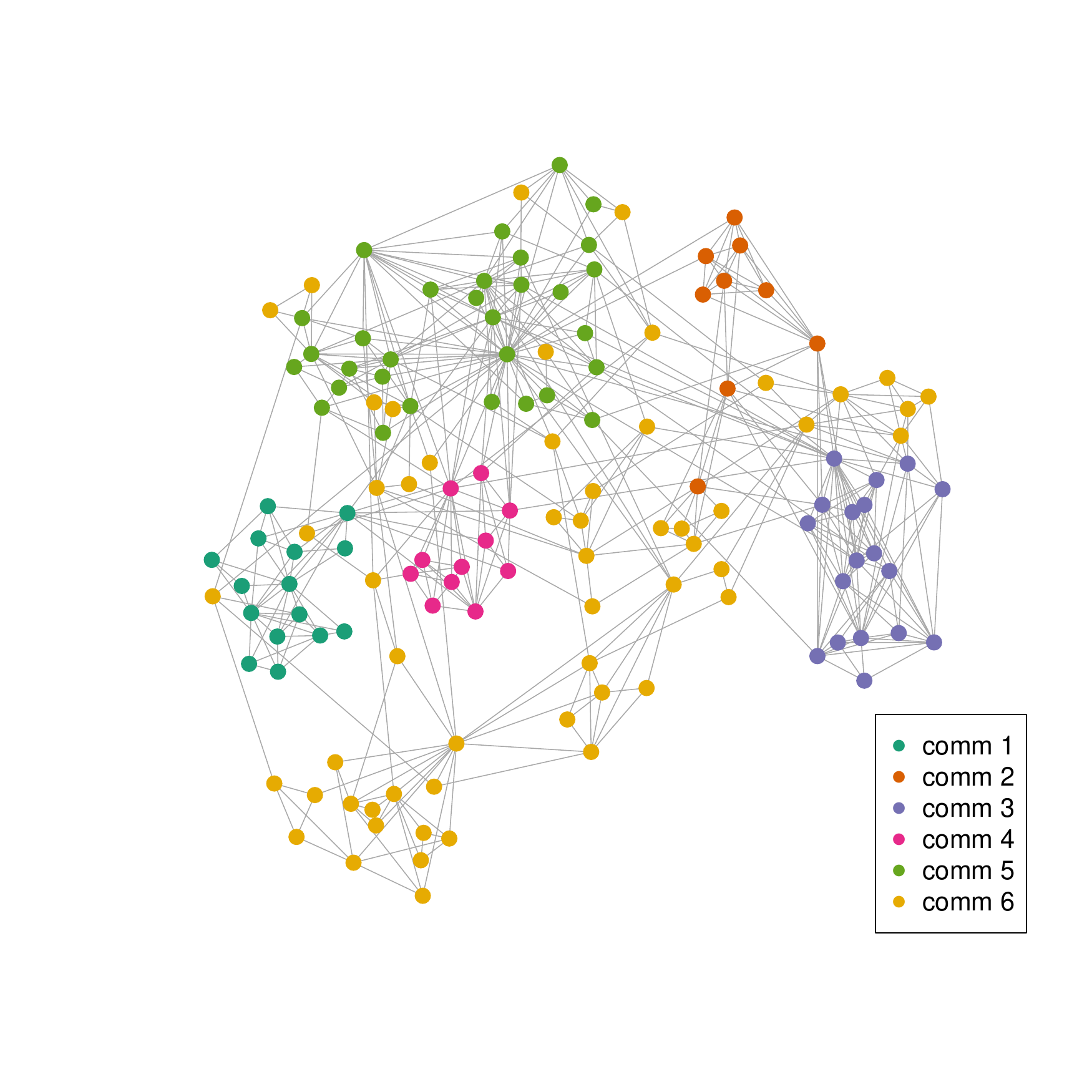}
\caption{Estimated communities by  the RSC.}
\label{fig:RSC-gene}
\end{subfigure}%
\caption{Additional results about the gene network analysis. (a): The confusion matrix represented by the Jaccard index between the communities of the 140 genes from the 3-core of the largest connected component ($n=140$ nodes) and the largest connected component ($n=285$ nodes). The 3-core community labels are show along the y-axis while the largest connected component community labels are shown along the axis. In both situations, the number communities is determined by the HCD algorithm. It can be seen that that clustering of the 140 genes are very consistent, as indicated by the high Jaccard index for matching communities. (b): Communities found by regularized spectral clustering on the anemia gene network.  Compared to HCD, these communities are less balanced and less interpretable.}
\end{figure}
%
%
%
%

\section{Hierarchical communities in a statistics citation network}\label{sec:stat-net}

This dataset \citep{ji2016coauthorship} contains information on statistics papers from four journals considered top (the Annals of Statistics, Biometrika, Journal of the American Statistical Association: Theory
and Methods, and Journal of the Royal Statistical Society Series B) for the period from 2003 to 2012.    For each paper, the dataset contains the title, authors, year of publication, journal information, the DOI, and citations.      We constructed a citation network  by connecting two authors if there is at least one citation (in either direction) between them in the dataset.  Following \cite{wang2016discussion}, we focused on the 3-core of the largest connected component of the network, ignoring the periphery which frequently does not match the community structure of the core. The resulting network, shown in Figure~\ref{fig:CitationNet}, has 707 nodes (authors) and the average degree is 9.29. 

The two HCD algorithms (HCD-Spec and HCD-SS) give the same result on
this network. We used the edge cross-validation (ECV) method \citep{li2016network} as the stopping rule instead of the
non-backtracking method, because ECV does not rely on the block model
assumption.   In this particular problem, ECV chooses a deeper and
more informative tree with 15
communities, shown in Table~\ref{tab:Keywords}, compared to the non-backtracking estimate of 11 communities.   
For a closer look at the communities, see the listing of 10 authors with the highest degrees in  each community in Table~\ref{tab:top10}.    Community labels in Table~\ref{tab:Keywords} were constructed
semi-manually from research keywords associated with people in
each community.   The keywords were obtained by collecting research
interests of 20 statisticians with the highest degrees in each
community,  from personal webpages, department research pages, Google
Scholar and Wikipedia (sources listed in order of inspection), with
stop word filtering and stemming applied.  The three most frequent keywords from research interests in each community are shown in Table~\ref{tab:Keywords}.    Note  that the citations are from publications between 2003 and 2012,  while the research interests were collected in 2018, so there is potentially a time gap. However, it is evident to anyone familiar with the statistics literature of that period that the communities detected largely correspond to real research communities, looking both at the people in each community and the associated keywords.

\begin{figure}[H]
\begin{center}
\vspace{-1.5cm}
\includegraphics[width=0.7\textwidth]{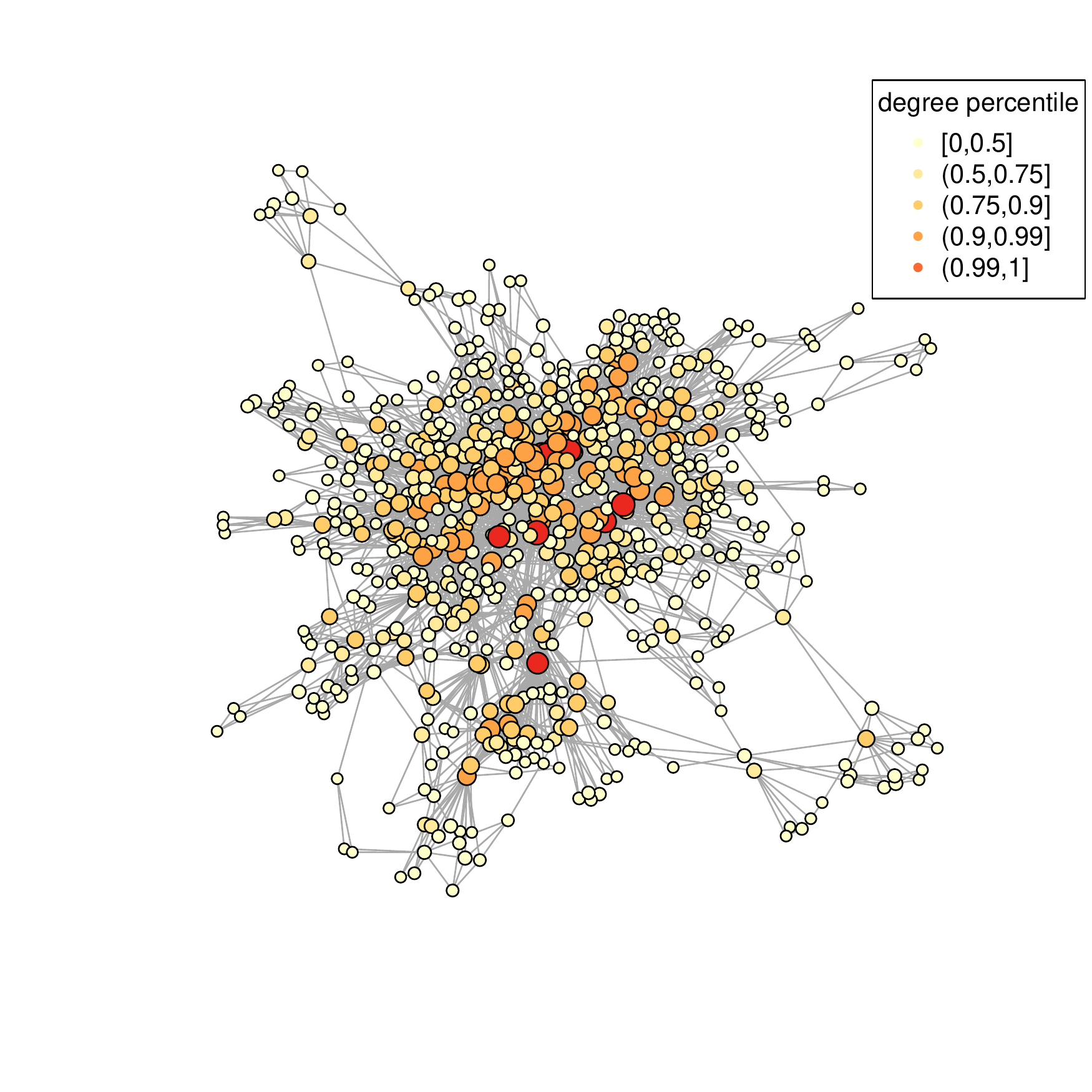}
\vspace{-2.5cm}
\end{center}
\caption{The 3-core of the statistics citation network.   Node
  size is proportional to its degree and the color indicates degree percentile.} 
\label{fig:CitationNet}
\end{figure}

The hierarchical tree of research communities contains a lot of
additional information, shown in
Figure~\ref{fig:StatisticianDendrogram}, and clearly reflects many well-known patterns. For example, Bayesian
statistics and design of experiments  split off very high up in the
tree, whereas various high-dimensional communities cluster together,
multiple testing is closely related to neuroimaging (which served as
one of its main motivations), functional analysis and
non/semi-parametric methods cluster together, and so on.   These
relationships between communities are just as informative as the
communities themselves, if not more, and could not have been obtained
with any ``flat'' $K$-way  community detection method.   

\begin{table}[H]
\centering
\caption{Statistics research communities detected by the HCD algorithm, identified by research area constructed from research interests.  The community size is shown in brackets and the three most frequently used research keywords were obtained from 20 highest degree nodes in each community. }
\label{tab:Keywords}
{\footnotesize
\begin{tabular}{l|lll}

Community research area [size]&  \multicolumn{3}{c}{Top three research interests keywords (from webpages)}\\ 
  \hline
  \hline
Design of experiments [16] & design & experiment & theory \\ 
\hline
Bayesian statistics [98] & Bayesian & model & inference \\ 
\hline
 Biostatistics and bio applications [35] &model & inference & sampling  \\ 
\hline
 Causal inference and shape (mixed)  [15] & inference & estimation & causal \\ 
\hline
 Nonparametrics and wavelets [26] & model & nonparametric & estimation \\ 
\hline
 Neuroimaging [18] & imaging & Bayesian & model \\ 
\hline
 Multiple testing/inference [92] & inference & multiple & test \\ 
\hline
 Clinical trials and survival analysis [45] & survival & clinical & trial \\ 
\hline
 Non/semi-parametric methods [38] & model & longitudinal & semi-parametric \\ 
\hline
 Functional data analysis [96] & functional &  model & measurement \\ 
\hline
 Dimensionality reduction [35] & dimension & reduction & regression \\ 
\hline
 Machine learning  [21] & machine learning & biological & mining \\ 
 \hline
 (High-dim.) time series and finance [36] & financial & econometrics & time \\ 
 \hline
 High-dimensional theory [29] & high-dimensional & theory & model \\ 
\hline
 High-dimensional methodology [107] & high-dimensional & machine learning & model \\ 
  \hline
\end{tabular}
}
\end{table}

This network was also studied by \cite{ji2016coauthorship}, though they did not extract the core. 
Our results are not easy to compare since there is no ground truth available and the number of communities they found is different.   Briefly, they found three communities initially, and then upon finding that one of them is very mixed, further broke it up into three, which can be seen as manually hierarchical clustering.   They interpreted the resulting five communities as ``Large-Scale Multiple Testing",``Variable Selection", ``Nonparametric spatial/Bayesian statistics", ``Parametric spatial statistics", and ``Semiparametric/Nonparametric statistics".    Though some of the labels coincide with our communities  in Table~\ref{tab:Keywords},  it seems more mixed, and  the hierarchical information which allows you to see which communities are close and which are far apart is not available from a flat partition.

\begin{figure}[H]
\begin{center}
\vspace{-1.5cm}
\includegraphics[width=0.8\textwidth]{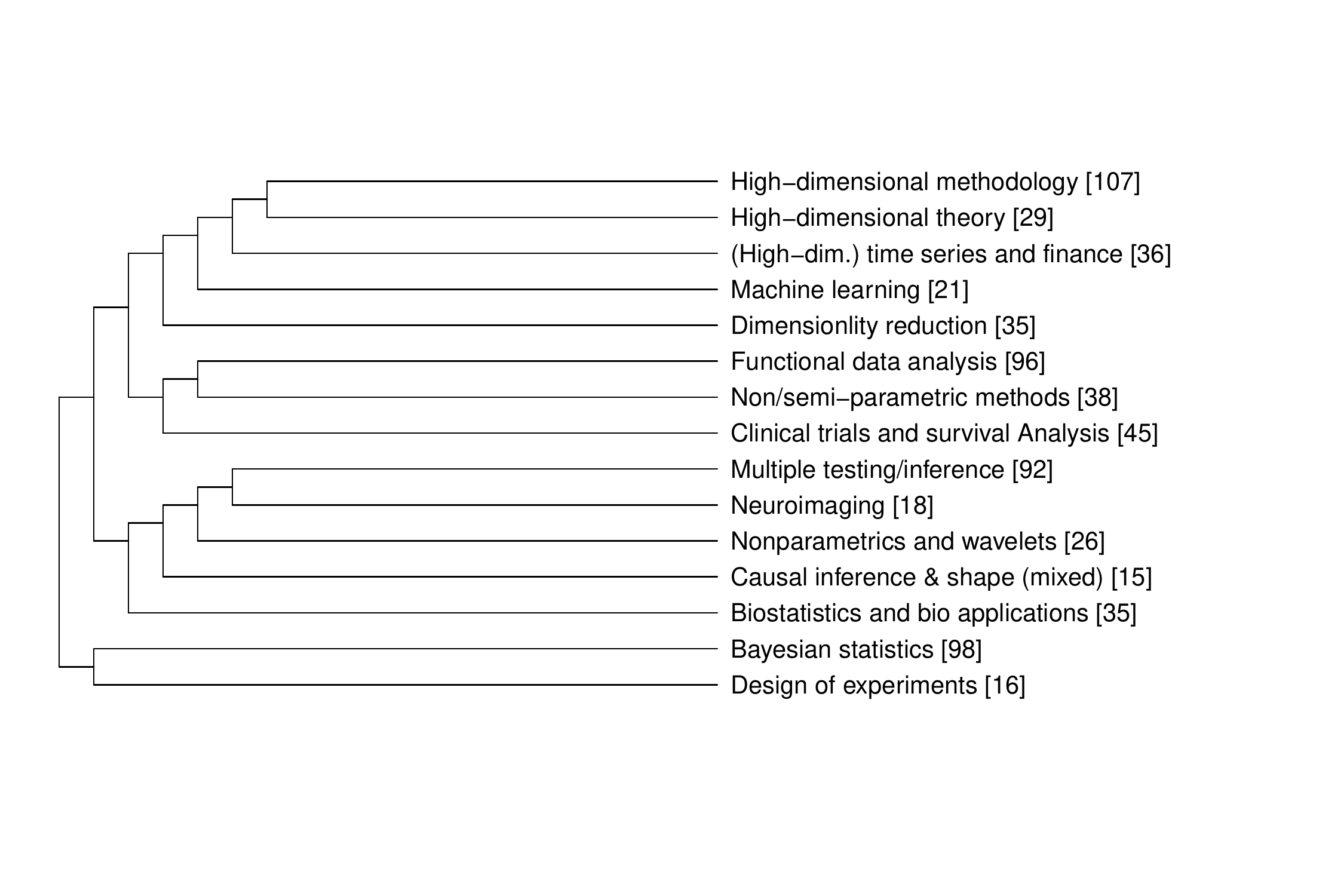}
\end{center}
\vspace{-1.5cm}
\caption{The dendrogram of 15 communities obtained by HCD from the 3-core of the statistics citation network.  Research areas are manually labeled based on keywords from webpages, and community sizes are shown in brackets.}
\label{fig:StatisticianDendrogram}
\end{figure}

\begin{table}[H]
{\tiny
\centering
\caption{The ten statisticians with the highest degrees for each of the 15 communities shown in Table~\ref{tab:Keywords}, ordered by degree within each community.}
\label{tab:top10}
\begin{tabular}{p{1.5in}|p{4.5in}}
  \hline
 Community [size] & Ten  highest degree nodes, ordered by degree \\ 
  \hline
Design of experiments [16] & Boxin Tang, Rahul Mukerjee, Derek Bingham, Hongquan Xu, Peter Z G Qian, Randy R Sitter, C Devon Lin, David M Steinberg, Dennis K J Lin, Neil A Butler \\ 
\hline
 Bayesian statistics [98] & David Dunson, Alan E Gelfand, Abel Rodriguez, Gareth Roberts, Hemant Ishwaran, Michael I Jordan, James O Berger, Marc G Genton, Peter Muller, Gary L Rosner \\ 
\hline
Biostatistics and bio applications [35] & James R Robins, Haibo Zhou, Ian L Dryden, Dylan S Small, Huiling Le, Paul R Rosenbaum, Tapabrata Maiti, Zhiqiang Tan, Andrew T A Wood, Hua Yun Chen \\ 
\hline
 Causal inference \& shape (mixed) [15] & Marloes H Maathuis, Jon A Wellner, Sungkyu Jung, Nilanjan Chatterjee, Piet Groeneboom, Alfred Kume, Fadoua Balabdaoui, Richard Samworth, Lutz Dumbgen, Mark van der Laan \\ 
\hline
 Nonparametrics and wavelets [26] & Iain M Johnstone, Bernard W Silverman, Subhashis Ghosal, Felix Abramovich, Aad van der Vaart, Axel Munk, Marc Raimondo, Theofanis Sapatinas, Anindya Roy, Dominique Picard \\ 
\hline
 Neuroimaging [18] & Jonathan E Taylor, Joseph G Ibrahim, Hongtu Zhu, Armin Schwartzman, Stephan F Huckemann, Heping Zhang, Bradley S Peterson, Stuart R Lipsitz, Ming-Hui Chen, Rabi Bhattacharya \\ 
\hline
 Multiple testing \& inference [92] & T Tony Cai, Larry Wasserman, Jiashun Jin, Christopher Genovese, Bradley Efron, Yoav Benjamini, David L Donoho, John D Storey, Sanat K Sarkar, Joseph P Romano \\ 
\hline
 Clinical trials and survival analysis [45] & Dan Yu Lin, Donglin Zeng, Xuming He, Lu Tian, L J Wei, Jing Qin, Michael R Kosorok, Guosheng Yin, Jason P Fine, Tianxi Cai \\
\hline
 Non- \& semi-parametric methods [38] & Raymond J Carroll, Naisyin Wang, Xihong Lin, Enno Mammen, Wenyang Zhang, Zongwu Cai, Annie Qu, Jianxin Pan, Xiaoping Xu, Arnab Maity \\ 
\hline
 Functional data analysis [96] & Peter Hall, Hans-Georg Muller, Fang Yao, Jane-Ling Wang, Lixing Zhu, Hua Liang, Gerda Claeskens, Jeffrey S Morris, Yanyuan Ma, Yehua Li \\ 
\hline
 Dimensionality reduction [35] & Hansheng Wang, Bing Li, R Dennis Cook, Yingcun Xia, Liping Zhu, Chih-Ling Tsai, Lexin Li, Liqiang Ni, Francesca Chiaromonte, Peng Zeng \\ 
\hline
 Machine learning [21] & Hao Helen Zhang, Xiaotong Shen, Yufeng Liu, Yichao Wu, J S Marron, Jeongyoun Ahn, Michael J Todd, Junhui Wang, Wing Hung Wong, Yongdai Kim \\ 
\hline
 (High-dim.) time series and finance [36] & Jianqing Fan, Qiwei Yao, Song Xi Chen, Yacine Ait-Sahalia, Yazhen Wang, Cheng Yong Tang, Bo Li, Jian Zou, Liang Peng, Sam Efromovich \\ 
\hline
 High-dimensional theory [29] & Peter J Bickel, Cun-Hui Zhang, Alexandre B Tsybakov, Jian Huang, Emmanuel J Candes, Martin J Wainwright, Terence Tao, Sara van de Geer, Alexandre Belloni, Lukas Meier \\ 
\hline
 High-dimensional methodology [107] & Hui Zou, Runze Li, Peter Buhlmann, Nicolai Meinshausen, Yi Lin, Elizaveta Levina, Ming Yuan, Trevor J Hastie, Jianhua Z Huang, Ji Zhu \\ 
   \hline
   \end{tabular}
}
\end{table}

\end{appendix}


\end{document}